\documentclass[10pt]{article}
\usepackage[utf8]{inputenc}
\usepackage{subfiles}
\usepackage[english]{babel}
\usepackage[a4paper,top=3 cm,bottom=3 cm,inner=2.5 cm,outer=2.5 cm]{geometry}
\usepackage{graphicx}
\usepackage{subfigure, tikz}
\usetikzlibrary{patterns}
\usetikzlibrary{positioning}
\usetikzlibrary{decorations.markings}
\usepackage{bm}
\usepackage{mathtools}
\usepackage{amsmath}
\usepackage{amsfonts}
\usepackage{mathbbol}
\usepackage{enumitem}
\usepackage{cancel}

\setlist[itemize]{label=\textbullet}
\usepackage{float}
\usepackage[utf8]{inputenc}
\usepackage{hyperref}
\usepackage{amsmath, amssymb, amsthm} 
\usepackage{eurosym}
\newtheorem{thm}{Theorem}[section] 
\newtheorem{lem}[thm]{Lemma}
\newtheorem{prop}[thm]{Proposition} 
\newtheorem{cor}[thm]{Corollary}
\newtheorem{defn}{Definition} [section]
\newtheorem{rem}{Remark}[section]

\theoremstyle{definition}

\usepackage{wasysym}
\hypersetup{pdfborder={0 0 0}} 
\usepackage{url} 
 
\usepackage{titlesec}
\titleformat{\paragraph}
{\normalfont\normalsize\bfseries}{\theparagraph}{1em}{}
\titlespacing*{\paragraph}
{0pt}{3.25ex plus 1ex minus .2ex}{1.5ex plus .2ex}

\titleformat{\subparagraph}
{\normalfont\normalsize\bfseries}{\theparagraph}{1em}{}
\titlespacing*{\paragraph}
{0pt}{3.25ex plus 1ex minus .2ex}{1.5ex plus .2ex}

\usepackage{mathrsfs}
\usepackage{leftidx}
\usepackage{braket}
\usepackage{physics}

\newcommand{\di}{\text{d}}
\newcommand{\en}[1]{\omega_{\textbf{#1}}}

\newcommand{\bshs}[2]{\mathscr{B}^{#1^{2}}_{\textbf{#2}}}

\newcommand{\supp}[1]{\text{supp}(#1)}

\newcommand\numberthis{\addtocounter{equation}{1}\tag{\theequation}}

\usepackage{csquotes}
\usepackage[
backend=biber,
giveninits=true,
style=alphabetic,
isbn = false,
url = false,
maxbibnames = 4,
]{biblatex}
\begin{document}
\par 
\bigskip 
\LARGE 
\noindent 
\textbf{Secular growths and their relation to equilibrium states in perturbative QFT} 
\bigskip \bigskip
\par 
\rm 
\normalsize 
 
\large
\noindent 
{\bf Stefano Galanda$^{1,a}$}, {\bf Nicola Pinamonti$^{1,2,b}$}, {\bf Leonardo Sangaletti$^{3,c}$} \\
\par
\small

\noindent$^1$ Dipartimento di Matematica, Dipartimento di eccellenza 2023-2027, Universit\`a di Genova - Via Dodecaneso, 35, I-16146 Genova, Italy. \smallskip

\noindent$^2$ Istituto Nazionale di Fisica Nucleare - Sezione di Genova, Via Dodecaneso, 33 I-16146 Genova, Italy. \smallskip

\noindent$^3$ Institut f\"ur Theoretische Physik, Universit\"at Leipzig, Br\"uderstraße 16, 04103 Leipzig, Germany.\smallskip
\smallskip

\noindent E-mail: 
$^a$stefano.galanda@dima.unige.it, 
$^b$pinamont@dima.unige.it,
$^c$leonardo.sangaletti@uni-leipzig.de\\ 

\normalsize
${}$ \\ \\
 {\bf Abstract} \ \
In the perturbative treatment of interacting quantum field theories, if the interaction Lagrangian changes adiabatically in a finite interval of time, secular growths may appear in the truncated perturbative series also when the interaction Lagrangian density is returned to be constant. If this happens, the perturbative approach does not furnish reliable results in the evaluation of scattering amplitudes or expectation values. In this paper we show that these effects can be avoided for adiabatically switched-on interactions, if the spatial support of the interaction is compact and if the background state is suitably chosen. We start considering equilibrium background states and show that, when thermalisation occurs (interaction Lagrangian of spatial compact support), secular effects are avoided. Furthermore, no secular effects pop up if the limit where the Lagrangian is supported everywhere in space is taken after thermalisation (large time limit), in contrast to the reversed order. This result is generalized showing that if the interaction Lagrangian is spatially compact, secular growths are avoided for generic background states which are only invariant under time translation and to states whose explicit dependence of time is not too strong. Finally, as an application, the presented theorems are used to study a complex scalar and a Dirac field, on a background KMS state, in a classical external electromagnetic potential and the contribution to the two point-function given by a generic loop diagram arising from a second order perturbative expansion.
\bigskip
${}$

\section{Introduction}
A complete rigorous formulation of interacting relativistic quantum field theories, like quantum electrodynamics (QED), is in general not yet at disposal. 
In spite of this fact, even if no rigorous construction of interacting fields is available in many cases, perturbative methods truncated at finite order furnish predictions of scattering amplitudes or expectation values of certain observables which are in some cases in extremely high accordance with experiments. This is the case of the electromagnetic fine-structure constant or the gyromagnetic ratio of the electron predicted in QED.\\
However, when time dependent interaction Lagrangians are considered, e.g.~an external potential for an electron switched on in a finite interval of time, it often happens that the contributions to expectation values or scattering amplitudes at order $n$ in the perturbation Lagrangian grow as $t^n$. This breaks the validity of the assumption that higher order corrections are smaller than the lower order ones, at least at large times. In these cases the truncated perturbative series cannot furnish reliable results.  
These effects are referred in the literature as \textit{secular effects} or \textit{secular growths}. 

In the recent years these problems of perturbation theory were extensively studied in the literature. Here we recall some relevant results: in the context of electrodynamics or self-interacting theories \cite{Akhmedov_2014, Trunin}, for the violation of the stationary approximation in expanding universes \cite{AkhmedovMoschella}, for the failure of perturbation theory near horizons \cite{Burgess_2018}, for quantum corrections to the gravitational potential on de Sitter \cite{FrobSecular, FrobSecular2, Prokopec} in the resummation of the perturbative series for non-equilibrium states and at finite temperature \cite{AltherrResummation_NonEquilibrium, Gautier, Boyanovsky_FiniteTemperature}; see also further citations therein. 

It is believed that secular effects are artefacts of perturbation theory and are present because, in the full theory, oscillations in time arise. For this reason, the validity of the perturbative approach to describe certain interacting quantum theories or even self-interacting quantum fields propagating also on certain curved backgrounds has been questioned in the literature. In order to clarify this point of view, Section \ref{se:toy-model} is devoted to the construction of a toy model where these observations become manifest.

In this toy model and in other models presented Section 
\ref{se:scalar-ED}
and in Section
\ref{sec: Dirac example}
where an external classical electromagnetic potential interacts with a quantum charged scalar field or a Dirac field in an equilibrium state, secular effects arises because an external potential which is switched on adiabatically is considered. In these cases, the interaction Lagrangian is quadratic in the field, hence one could have analyzed directly the theory without using perturbation theory for quantum fields, thus avoiding these problems. 
However, secular effects are present also in theories where the interaction Lagrangian is more than quadratic and we do not have any other way than perturbation theory to take it into account.
This is for example the case presented in Section
\ref{sec: Loop Example}
and in Section \ref{sec: Loop Example sec}
where a self interacting scalar field in a generic translation and rotation invariant state is studied.  
In this last case, time dependent growths arise in the diagrammatic expansion at higher perturbative order 
in which at least one diagram has a vanishing combination of the frequencies associated with the internal and the external legs, e.g. loop diagram.
This last example shows secular growths which
 are extensively studied in the rich literature, see e.g. \cite{Akhmedov_2014, Akhmedov_2020}.

The aim of this work is to study the validity of perturbation theory, for compactly spatially supported interaction Lagrangians, 
for theories defined in a Minkowski spacetime
by proving the absence of secular effects. In addition, the subtle limit removing the spatial cutoff is carefully discussed. In particular, time translation invariant states after the switch on, do not exhibit any secular growth. This happens for example, if thermalisation occurs, for equilibrium states, at late time. Notably, if the state is already in thermal equilibrium for the free dynamics and if the interaction Lagrangian acts as a perturbation, equilibrium is reached again. This process of return to equilibrium is well under control in the case of $C^*$-dynamical systems \cite{BratteliRobinson}, while it is not completely clear whether it holds for generic interacting quantum field theories and if it is compatible with perturbation theory.

In \cite{FaldinoEquilibriumpAQFT} it is proven that the spatial compactness of the interaction Lagrangian is a sufficient condition for thermalisation at late times, for a self-interacting neutral scalar field. In the same paper, examples where return to equilibrium fails for spatially non-compactly supported interactions are given (See also \cite{Sangaletti, Joao2019}). In addition, the limit removing the spatial cutoff can be taken for an equilibrium state of the interacting theory treated with perturbative methods \cite{FredenhagenLindnerKMS_2014, Galanda, DragoHackPinamonti, JoaoNicNic}. Hence, the limit where the interaction Lagriangian is supported everywhere in space needs to be taken only after thermalisation is reached. In other words, in perturbation theory, the order of the large time limit and the limit removing the space cutoff cannot be changed. 

Motivated by this observation, we prove in Section \ref{sec:req} that if the interaction Lagrangian is spatially compactly supported, thermalisation occurs also for charged scalar and fermionic fields. Therefore, secular effects are absent from the perturbation series.
Later, the discussion is extended to a initial state $\omega$ out of equilibrium. Namely, in Theorem \ref{th:truncated}, we prove that if a state $\omega$ is secularily bounded (See Definition \ref{def:secularily-bounded}) and if the interaction Lagrangian is of spatially compact support, secular effects are absent at each order in the perturbation series.
In Corollary \ref{cor: NoSecInvariant} we prove the validity of the property of being secularily bounded for a quasifree state $\omega$ invariant under time translations, as a consequence of a suitable clustering condition given in \eqref{eq:clustering=Derivative}. Hence, in this case the absence of secular effects is guaranteed by an application of Theorem \ref{th:truncated}.
Finally, in Theorem \ref{theo:time dependent}, we prove that for quasifree states (not necessarily time translations invariant) the secular boundedness is implied by certain conditions on the two-point function of the state itself.

Some applications of the  results obtained in Theorem \ref{th:truncated} and
Corollary \ref{cor: NoSecInvariant} are presented in Section  \ref{se:applications}.
There, it is shown that secular effects are avoided in the case of a charged scalar and Dirac field which are defined on a Minkowski spacetime and which interact with an external electromagnetic field when its spatial support is compact. This is done
analyzing directly some contributions at finite perturbation order. Finally, in Section \ref{sec: Loop Example no sec} it is shown that secular effects are avoided in certain loop diagrams when the interaction Lagrangian has a spatial cutoff. 
It should be stressed that the results of the present paper are not applicable in some other situations of interest. Actually, in this paper, we restrict our attention to quantum field theories defined on Minkowski spacetime. Even though an adaptation to the case of stationary spacetime with infinitely extended Cauchy surfaces and with  smooth metric coefficients seems to be possible, we expect that in some cases secular effects cannot be avoided.
More precisely, to prove the results of this paper, we make extensive use of the fact that the interaction is localized in a spatial compact region, while the field theory is defined on the entire Minkowski spacetime. If this assumption is removed, for instance defining the theory in a box with periodic boundary conditions, the time decay property of the correlation functions ceases to hold and in that case we expect that secular effects cannot be avoided. 

In order to obtain the general results mentioned above, we make use of methods proper of perturbative algebraic quantum field theory (pAQFT)
\cite{BrunettiFredenhagen00, BFV03, HW02, HollandsWald2001, HW05, Fredenhagen_2012,  Fredenhagen2015}, (see also 
\cite{DutschBook, KasiaBook}). This approach treats interacting relativistic quantum fields, combining techniques of renormalization of perturbative interacting theories \cite{EpsteinGlaser} with the axiomatic algebraic approach to QFT \cite{HaagKastler, HaagLQP, BuchFredAQFT}. This framework can be directly applied to quantised field theories on generic (globally hyperbolic) spacetimes and it puts the focus on the observables, identified with elements of algebras, assigned to open regions of spacetime. Locality is encoded in the commutation relations between observables of algebras associated to different regions. Once the local algebras of the free theory have been constructed, the ones of the interacting theory can be mapped inside the $*$-algebra of the free theory in the form of formal power series in the coupling constant. States are constructed in a second step, assigning positive, normalised and linear functionals over the algebras. In addition, once a state is chosen, the \textit{GNS theorem} allows to represent the algebra of observables as operators on a Hilbert space. One of the advantages of this formulation, is that renormalization turns out to be independent from the choice of a specific state by construction.

The paper is organised as follows. At the end of this section we present a toy model in which it is made manifest that secular effects are artefacts of perturbation theory. The second section, is devoted to a brief overview of the framework in which the results are presented and to the construction of equilibrium states for interacting theories. The third section contains the main result concerning the absence of secular growths, and its adaption to some physically relevant states. In this setting, the relation with thermalisation (return to equilibrium) for equilibrium states is discussed. Finally, the fourth section is devoted to the application of these results to two concrete examples: a complex scalar and a Dirac field coupled with a strong external electromagnetic field. An explicit computation at first order in perturbation theory is presented, in order to make manifest in a concrete example the discussion done in the previous sections.

\subsection{Simple toy model - adiabatic mass changes in Klein--Gordon fields}
\label{se:toy-model}
In this section we present a simple toy model in which it is made clear that secular effects are artefacts of perturbation theory. It will clarify how they possibly pop up when the $n$-point correlation functions of the state, after the switching on of the interaction, are not invariant under time translations but rather oscillate in time.
Similar analyses are present in the PhD Thesis \cite{Joao2019} and in the Master Thesis \cite{Sangaletti}. 
\\
We work in Minkowski\footnote{The convention adopted in this work on the signature of the metric is $-+++$ and on units $\hbar = c = 1$.} spacetime $(\mathbb{M},\eta)$ and we consider a massive Klein--Gordon field $\phi$ of mass $m$.
Let $\chi \in \mathcal{C}^{\infty}(\mathbb{R})$ (or eventually the limit in the distributional sense of a sequence of smooth functions) be positive with $\mathrm{supp}(\chi) = \{ t \in \mathbb{R} : t > -\epsilon \}$, for $\epsilon \in \mathbb{R}^+$, whose derivative (eventually defined in the weak sense) $\Dot{\chi}$ has support in the region $[-\epsilon, 0]$. We also assume $\chi$ equal to $1$ for $t\geq 0$. This function governs the way in which the mass of the Klein--Gordon field on Minkowski spacetime is adiabatically changed in the interval of time $[-\epsilon,0]$. This change corresponds to considering an interaction Lagrangian density of the form $-\delta m^2\phi^2/2 $. The equation of motion satisfied by $\phi$ is:
\begin{equation*}
    (\square - m^2 - \chi \, \delta m^2) \phi = 0.
\end{equation*}
This is solvable in Fourier space and becomes an ordinary differential equation for the modes $\xi_p(t)$:
\begin{equation*}
    \Ddot{\xi}_p(t) + (|\mathbf{p}|^2 + m^2 + \delta m^2 \chi(t)) \xi_p(t) = 0.
\end{equation*}
By knowing the mode solutions, a general pure, translation invariant, Gaussian state of the corresponding quantised theory has the two-point function of the form:
\begin{equation*}
    \omega_2(x,y) =\frac{1}{(2 \pi)^3} \int_{\mathbb{R}^3}  \di^3\mathbf{p} \, \overline{\xi_p(t_x)} \xi_p(t_y) e^{i \mathbf{p} (\mathbf{x} - \mathbf{y})}.
\end{equation*}
For the sake of simplicity, we fix as initial condition for $t < -\epsilon$:
\begin{equation*}
    \xi_p(t) = \frac{e^{-i \omega_0 t}}{\sqrt{2 \omega_0}},
\end{equation*}
corresponding to the assumption of the system being prepared in its ground state. Here, $\omega_0 = \sqrt{|\mathbf{p}|^2 + m^2}$. Therefore, the general modes at $t > 0$ can be determined exactly and have the general form:
\begin{equation*}
    \xi_p(t) = \alpha_p  \frac{e^{i \omega_1 t}}{\sqrt{2 \omega_1}} + \beta_p  \frac{e^{-i \omega_1 t}}{\sqrt{2 \omega_1}},
\end{equation*}
where $\omega_1 = \sqrt{|\mathbf{p}|^2 + m^2 + \delta m^2}$ and $\alpha_p, \beta_p$ are complex functions, determined by the switch-on, and depending on the norm of the three momentum $|\mathbf{p}|$. In particular, if $\dot\chi$ is smooth and of compact support, $\beta_p$ decays rapidly for large $|\mathbf{p}|$. Furthermore, since $|\alpha_p|^2= 1 + |\beta_p|^2$, we have that also $\alpha_p \beta_p$ is of rapid decay.\\
Therefore, after the smooth switching on of the interaction the system is described by:
\begin{equation*}
    \omega_2(x,y) = \frac{1}{(2 \pi)^3}\int_{\mathbb{R}^3} \frac{\di^3 \mathbf{p}}{2 \omega_1} \, \bigg(|\alpha_p|^2 e^{-i \omega_1 (t_x - t_y)} + |\beta_p|^2 e^{i \omega_1 (t_x - t_y)} +  \overline{\alpha_p} \beta_p e^{-i\omega_1 (t_x + t_y)} + \alpha_p \overline{\beta_p} e^{+i\omega_1 (t_x + t_y)}\bigg) e^{i \mathbf{p} (\mathbf{x} - \mathbf{y})} ,
\end{equation*}
where $t_x,t_y > 0$. The first two contributions depend on $t_x-t_y$ and thus they are invariant under time translations $(t_x,t_y) \to (t_x+a,t_y+a)$. The only terms sensitive to the infrared behaviour in time are those depending on $t_x + t_y$. In particular, the dependence on $t_x + t_y$ appears only in the oscillating function, and $\overline{\alpha_p}\beta_p$ decreases rapidly. Hence, the corresponding contribution is bounded. We thus observe that, in the exact theory, no secular growths are visible: only contributions which oscillate in $t_x+t_y$.\\
The very same behaviour occurs also in the limit where $\chi(t)$ tends to the Heaviside step function  $\theta(t)$. In that case, an explicit form of $\alpha_p$ and $\beta_p$ can be obtained: 
\begin{equation*}
    \alpha_p = \frac{1}{2}\bigg( \sqrt{\frac{\omega_1}{\omega_0}} + \sqrt{\frac{\omega_0}{\omega_1}} \bigg), \qquad
    \beta_p = \frac{1}{2}\bigg( \sqrt{\frac{\omega_1}{\omega_0}} - \sqrt{\frac{\omega_0}{\omega_1}} \bigg),
\end{equation*}
and the contribution which is not invariant under time translation has the form
\begin{align*}
    O\coloneqq&\frac{2}{(2 \pi)^3} \int_{\mathbb{R}^3} \frac{\di^3 \mathbf{p}}{\omega_1} \frac{\delta m^2}{4 \omega_1 \omega_0}\cos\big( \omega_1 (t_x + t_y)\big) e^{i \mathbf{p} (\mathbf{x} - \mathbf{y})}.
\end{align*}
We now expand $O$ in powers of $\delta m^2$. Denoting by $O_{n}$, the term of order $n$ in $\delta m^2$ in the expansion, the contribution which shows the largest growth in time has the form:
\[
    O_n= C (t_x + t_y)^{n - 3/2}(\delta m^2)^{n+1},
\]
manifesting the presence of time growths.\\
Comparable  behaviours and further details were recently studied for thermal equilibrium states for scalar field theories in the Master Thesis \cite{Sangaletti} and in the PhD Thesis \cite{Joao2019}. In this context, the space localisation of the interaction, via a space cutoff, prevents the arising of such growths at any perturbative order. This is a consequence of the property of \textit{return to equilibrium} (\textit{thermalisation}), holding just for spatially localised interactions (\cite{FaldinoEquilibriumpAQFT} for the case of scalar field theories).\\

We remark here that the secular effects are not a peculiarity of quadratic interacting Lagrangian. They arise also in loop diagrams of self interacting scalar fields as for example, $\lambda \varphi^4$ on a thermal state or $\lambda \varphi_1 (\varphi_2)^2$ on the vacuum state for two scalar fields $\varphi_1$ and $\varphi_2$ with mass respectively $m$ and $m/2$. These are the kind of growths that are discussed in the rich literature on the topic \cite{Akhmedov_2014, Akhmedov_2020}. See e.g. Section \ref{sec: Loop Example sec} for a direct analysis. In all these cases, when the interaction Lagrangian has spatially compact support, the secular growths are avoided. 

\section{Perturbative approach to interacting quantum field theories in the algebraic setting}
In order to fix notation and review core results that we refer to later, we discuss with little more details the class of theories to which our results apply.\\
First of all, let us specify that, in this work, the main focus is on Minkowski spacetime ($\mathbb{M}, \eta$), metric signature $(-,+,+,+)$, but most of the presented results are generalizable to stationary (\textit{globally hyperbolic} \cite{Sanchez, BernalSanchez2005}) spacetime with non compact Cauchy surfaces and with well posed Cauchy problem for wave operators \cite{Bar_2014, BarGinPfaf}.\\
We will consider quantum field theories described by a Lagrangian density that can be factorised in the following way:
\begin{equation*}
    \mathcal{L} = \mathcal{L}_{F} + \mathcal{L}_I,
\end{equation*}
where $\mathcal{L}_{F}$ denotes the Lagrangian density of the free theory while $\mathcal{L}_{I}$ the interaction Lagrangian density. In particular, the interacting scalar field, both real ($\mathcal{L}_F = -\frac{1}{2} \partial^{\mu} \phi \partial_{\mu} \phi - \frac{m^2}{2} \phi^2$) and complex ($\mathcal{L}_F = -\partial^{\mu} \varphi \partial_{\mu} \varphi^* - m^2 \varphi \varphi^*$), and the interacting Dirac field ($\mathcal{L}_F = \overline{\psi}(i \cancel{\partial} - m)\psi$). 

\subsection{Interacting quantum field theory, S matrices and Bogoliubov maps}
In this section we briefly recall the main steps used to construct interacting quantum field theories 
with perturbation methods in the algebraic setting. 
We refer to the original works 
\cite{BrunettiFredenhagen00, BFV03, HW02, HollandsWald2001, HW05, Fredenhagen_2012,  Fredenhagen2015}
where the formalism has been developed and to the books \cite{DutschBook, KasiaBook} for recent detailed reviews. 

In this framework, the observables of the theory are seen as element of a {\bf $*$-algebra} $\mathcal{A}$ generated by normal ordered local fields smeared with compactly supported smooth functions. A generic normal ordered local field is denoted by $\Phi(f)$ where $f$ is the smearing function. The set of local fields is denoted by $\mathcal{F}_{\text{loc}}$ and typical examples of elements in $\mathcal{F}_{\text{loc}}$ are:
\[
:\phi^2:(f), \qquad :\varphi\varphi^\dagger:(f), \qquad :\psi\gamma_\mu \psi^\dagger:(f),
\]
where $:\cdot:$ indicates the (covariant) normal ordering \cite{HW02} and $\gamma_\mu$ are the standard $\gamma$ matrices. $\phi$ is a scalar neutral field, $\varphi$ is a charged scalar field and $\psi$ the linear Dirac field. 
The support of a local field coincides with the support of $f$. \\
We restrict our attention to local fields which are polynomial in the basic scalar field $\phi$, $\varphi$ or the basic Dirac fields $\psi$ and $\psi^\dagger$. $\mathcal{A}$ is then obtained as the collection of finite sum and finite products of its generators (normal ordered local fields).
Furthermore, $\mathcal{A}(\mathcal{O})\subset \mathcal{A}$ is the subalgebra generated by fields supported in $\mathcal{O}$.\\	
Concrete realizations of $\mathcal{A}$ are obtained identifying elements of $\mathcal{A}$ with functionals over smooth field configurations, see \cite{BruDutFre09} for the scalar field and \cite{BDFR22} for the extension to the case of anticommuting Dirac fields. Products in $\mathcal{A}$ are given in terms of the commutator function (Pauli-Jordan propagator) $\Delta$ of the basic fields $\phi$  and $\varphi$ in the case of scalar fields and in terms of the anticommutator function $\cancel{S}$ of the basic Dirac fields $\psi$ and $\psi^\dagger$. In particular, denoting by $\star$ the {\bf product} among generators, we have for every $A,B\in \mathcal{A}$:\\
\[
A \star B \coloneqq \mathcal{M}e^{
\frac{i}{2}\!\! \int\!\! \Delta(x,y) \frac{\delta}{\delta \phi (x)}
\otimes \frac{\delta}{\delta \phi (y)}  
+
\frac{i}{2} \!\!\int\!\! \Delta(x,y) \left( 
\frac{\delta}{\delta \varphi (x)}
\otimes  \frac{\delta}{\delta \varphi^\dagger\! (y)}
-
\frac{\delta}{\delta \varphi^\dagger\! (y)}
\otimes  \frac{\delta}{\delta \varphi (x)}
\right)
+
\frac{i}{2} \!\!\int\!\! \cancel{S}(x,y) \left( \frac{\delta_r}{\delta \psi (x)}
\otimes  \frac{\delta}{\delta \overline{\psi} (y)}
+
\frac{\delta_r}{\delta \overline{\psi} (y)}
\otimes  \frac{\delta}{\delta \psi (x)}
\right) }
A\otimes B
\]
where $\mathcal{M}(A\otimes B)$ is the graded pointwise product, $\delta/\delta \phi$ is the functional derivative, $\delta_r/\delta \psi = (-1)^{n+1}\delta/\delta \psi$ is the right functional derivative (where $n$ is the number of Dirac field generators) and $\overline{\psi} = \psi^\dagger\gamma^0$. This product implements ordinary free commutation relations among linear scalar fields and free anticommutation relations among Dirac fields. The support of $A\star B$ satisfies $\supp {A\star B} \subset \supp A \cup \supp B$. The involution on $\mathcal{A}$ is the complex conjugation for scalar fields and the adjoint in spinor space for Dirac fields. In the rest of the paper we drop the symbol $\star$ to denote the product
and replace it by juxtaposition when it is not explicitly needed.

A {\bf state} $\omega$ is a normalised, positive, linear functional on $\mathcal{A}$ whose evaluation on $A\in \mathcal{A}$ implements the mean expectation value. In order for $\omega$ to be a physical state on $\mathcal{A}$ we require it to be of Hadamard form \cite{KW91, Wald1995, Ra96, BFK}. This request is necessary to avoid divergences in the expectation values of elements of $\mathcal{A}$, even if composite fields are used as generators. Once a state is chosen, the standard picture where pure states are vectors in a suitable Hilbert space and expectation values are matrix elements of suitable operators can be recovered by the GNS construction.\\
If the state $\omega$ is quasi-free (Gaussian), the correlation functions of the theory are determined by the two-point function only. Furthermore, in this case, there is a representation of $\mathcal{A}$ in which the state $\omega$ takes a particularly simple form. Actually, in that case, (considering for simplicity only scalar fields):
\begin{equation}\label{eq:omega-deformation}
\omega(\Phi(f)\Phi(g)) = \left.\mathcal{M}e^{\int \omega_2(x,y) \frac{\delta}{\delta \phi(x)}\otimes \frac{\delta}{\delta \phi(y)} }\alpha_{\omega_2}(\Phi(f))\otimes \alpha_{\omega_2}(\Phi(g)) \right|_{\phi=0}
\end{equation}
where $\alpha_{\omega_2}$ is a $*$-isomorphism of $\mathcal{A}$ to $\mathcal{A}_{\omega_2}$ which is a $*$-algebra with a deformed product. Its explicit form on $F\in \mathcal{A}$ is: 
\begin{equation}\label{eq:alpha}
\alpha_{\omega_2} (F) \coloneqq e^{\frac{1}{2}\int \omega_2(x,y) \frac{\delta^2}{\delta \phi(x)\delta \phi(y)}} F.
\end{equation}
Below, in some proofs, we shall drop the symbol $\alpha_{\omega_2}$ keeping implicit the application of $\alpha_{\omega_2}$ in the representation of elements of $\mathcal{A}$ in $\mathcal{A}_{\omega_2}$ if not strictly necessary.

{\bf Interaction Lagrangians} $V\coloneqq\mathcal{L}_I( g) \in \mathcal{A}$ are normal ordered local fields, smeared with compactly supported smooth functions $g$ in order to avoid infrared divergences. Both the interaction Lagrangians and the observables we shall work with for Dirac theories are even in the number of Dirac field generators.

To represent elements of the interacting theory in $\mathcal{A}$ one needs to introduce a {\bf time ordered product} among local functionals which is a map:
\[
T : \mathcal{F}_{\text{loc}}^{\otimes n} \to \mathcal{A}.
\]
For a detailed list of properties and the proof of existence of these maps, also in curved spacetime, we refer to \cite{BrunettiFredenhagen00, HW02, HW05}. 
Here we simply mention the causal factorization property, that permits to represent certain time ordered products as ordinary products: $T(A\otimes B) = A B$ if $J^{+}\big(\supp{A}\big)\cap J^{-}\big(\supp{B}\big) = \emptyset$\footnote{Here $J^+$ ($J^-$) denotes the causal future (past).}. We also mention that the construction of these products is obtained using the Epstein and Glaser recursive procedure \cite{EpsteinGlaser}. At each recursive step, renormalization has to be employed to get a well defined element in $\mathcal{A}$. This is done using ideas of Steinmann \cite{Steinmann}, extending algebra valued distributions on $\mathcal{C}^{\infty}_0(\mathbb{M}^n\setminus D_n)$ to distributions on $\mathcal{C}^{\infty}_0(\mathbb{M}^n)$, where $D_n$ is the thin diagonal in $\mathbb{M}^n$.

Elements of the interacting theory are obtained as formal power series in $\lambda$ with coefficients in $\mathcal{A}$. In particular, by the above definition, the interaction Lagrangian is given by:
		\begin{equation}\label{eq:formaV}
    		V \coloneqq  \lambda \int \chi(t) h(\mathbf{x}) \mathcal{L}_I \di^4x,
		\end{equation}
where $\chi$ is a time cutoff and $h$ a space cutoff. Namely, $\chi\in \mathcal{C}^{\infty}_0(\mathbb{R})$ and $h \in \mathcal{C}^{\infty}_0(\Sigma)$, where $\Sigma$ is a Cauchy surface (set of points at fixed Minkowski time) in $\mathbb{M}$.\\
With the aid of the time ordered product we can define for a local functional $F \in \mathcal{F}_{\text{loc}}$ the corresponding {\bf $S$-matrix} as the time ordered exponential:
\begin{equation}\label{eq: def S-matrice}
    S(F) \coloneqq 1 + \sum_{n \geq 1} \frac{(i\lambda )^n}{n!} T(\underbrace{F \otimes \ldots \otimes F}_n).
\end{equation}
This series is in general not convergent and thus $S(F)$ is understood as an element of $\mathcal{A}[[\lambda]]$, the formal power series in $\lambda$ with coefficients in the algebra $\mathcal{A}$. Suppose now that $F_1, F_2, F_3 \in \mathcal{F}_{\text{loc}}$ and $J^-(\text{supp}(F_3)) \cap J^+(\text{supp}(F_1)) = \emptyset$, denoted as $F_1\gtrsim F_3$. Then, by the causal factorization property of the time ordering map follows a causal factorization for the $S$-matrix:
\begin{equation}\label{eq: causalfact}
    S(F_1 + F_2 + F_3) = S(F_1 + F_2)  S(F_2)^{-1}  S(F_2 + F_3).
\end{equation}
With the aid of $S$-matrices, we define the interacting field as \cite{Bogoliubov, Stueckelberg}:
\[
R_V(F) \coloneqq -\frac{i}{\lambda} \left. \frac{d}{dt} S(V)^{-1} S(V + t F)\right|_{t=0}, \qquad F\in \mathcal{F}_{\text{loc}} 
\]
which is an element of $\mathcal{A}[[\lambda]]$. \\
We stress that is possible to draw an analogy between this formula and the known Keldysh formalism. To make this manifest it is necessary to take into account the integration over time due to the smearing of the interaction Lagrangian and of the fields. Indeed, in $S(V+F)$ the local fields $V$ and $F$ present an integration from $-\infty$ to $+\infty$ in time, while the same integration is essentially performed backwards in $S(V)^{-1}$. See \cite{JoaoNicNic} for further details.\\
We have now all tools to construct the algebra of localised interacting fields as the subalgebra of $\mathcal{A}[[\lambda]]$ generated by $R_V(F)$, namely:
\begin{equation*}
    \mathcal{A}_I(\mathcal{O}) \coloneqq \big[ \{ R_V(F) | F \in \mathcal{F}_{\text{loc}}(\mathcal{O})\} \big] \subset \mathcal{A}[[\lambda]],
\end{equation*}
were the square bracket denotes the algebra generated by the elements enumerated in them. Notice that $\mathcal{A}_I(\mathcal{O})$ is not sensitive to changes of $\chi$ to $\chi'$ which occur after $\mathcal{O}$, namely if $\text{supp} \big((\chi-\chi')h \big)\cap J^{-}(\mathcal{O}) = \emptyset$. Furthermore, if the cutoff is changed before $\mathcal{O}$, namely if $\text{supp} \big((\chi-\chi')h \big) \cap J^{+}(\mathcal{O}) = \emptyset$, one obtains an interacting algebra $\mathcal{A}_I(\mathcal{O})'$ which is $*$-isomorphic to $\mathcal{A}_I(\mathcal{O})$. The isomorphism is realised by the adjoint action of a suitable formal unitary map. For this reason, as soon as interacting field observables localised in space and time are considered, we can easily extend the definition of $R_V$ also to time cutoff functions of the form:
\begin{equation}\label{eq:accad}
    \chi(t) = \left\{ \begin{aligned}
                          &0 \qquad \text{for $t < -2\epsilon$}\\
                          &1 \qquad \text{for $t \geq  -\epsilon$}
                      \end{aligned}\right.
\end{equation}
without altering the construction.\\
From now on, we will refer to $\chi(t)$ as the \textbf{smooth switch-on function}.

\subsection{Equilibrium states for interacting theories}
In this section we shall briefly discuss how the equilibrium states for an interacting theory are obtained starting from the corresponding equilibrium states of the free theory. 
We follow mainly the construction of \cite{FredenhagenLindnerKMS_2014}, for further aspects of thermal field theory we refer to \cite{LeBellac, KaputsaGale, LandsmanVanWert}. 
\\
In order to be able to present equilibrium states, we need to provide the algebra with a one-parameter group of automorphisms that describes {\bf time evolution} and makes it a dynamical system. In particular the free dynamics is defined on local functionals $\Phi(f) \in \mathcal{F}_{\text{loc}}$ by:
\begin{equation*}
    \tau_t(\Phi(f) ) \coloneqq \Phi(f)_t: = \Phi(f_t) \, , \quad f_t(s, \mathbf{x}) \coloneqq f(s + t, \mathbf{x}).
\end{equation*}
It is then extended to $\mathcal{A}$ by simply requiring that: 
\[
\tau_t(FG)=\tau_t(F)\tau_t(G).
\]
The {\bf interacting dynamics} $\tau_t^V$ is obtained from the free dynamics $\tau_t$ as the pull-back with respect to the Bogoliubov map, namely for every $F\in \mathcal{F}_{\text{loc}}$ 
\begin{equation}
    \tau_t^V(R_V(F)) \coloneqq R_V(\tau_t(F)),
\end{equation}
from which follows that:
\begin{equation*}
    \tau^V_t(S(V)^{-1} S(V + F)) = S(V)^{-1} S(V + F_t).
\end{equation*}

\subsubsection{Equilibrium states for the free dynamics}
The two most important categories of equilibrium states, for the free dynamics, are defined as follows:
\begin{defn}\label{def: ground}
Let $\omega^{\infty}$ be a state on a realization of the dynamical system $\mathcal{A}$ with free dynamics $\tau_t$. If the state satisfies the condition:
\begin{equation}\label{eq: ground}
    -i \partial_t \omega^{\infty}(A^* \tau_t(A))\big|_{t=0}\geq 0\, \quad \forall A\in \mathcal{A}.
\end{equation}
Then $\omega^{\infty}$ is of equilibrium with respect to $\tau_t$ and is referred to as \textit{ground} (\textit{vacuum}) state.
\end{defn}
Condition \eqref{eq: ground} is equivalent to the request that, in the GNS representation associated to the state $\omega^{\infty}$, the generator of the time evolution is positive semi-definite, with lowest energy state corresponding to the vector representation of $\omega^{\infty}$. Moreover, Equation \eqref{eq: ground} is also proven to be equivalent to the request that the function $z \to \omega^{\infty}(A \tau_z(B))$ is bounded in the region $\Im(z) \geq 0$ and analytic in $\Im(z) > 0$. For a proof of these statements see Proposition $5.3.19$ in \cite{BratteliRobinson}.
\begin{defn}\label{def: KMScond}
Let $0 \leq \beta < \infty$ and $\omega^{\beta}$ be a state on a realization of the dynamical algebra $\mathcal{A}$ with dynamics $\tau_t$. Then, it is of thermal equilibrium, or KMS (Kubo-Martin-Schwinger) state, at inverse temperature $\beta$ with respect $\tau_t$, if the functions:
\begin{equation*}
    (t_1, \ldots, t_n) \mapsto \omega^{\beta}(\tau_{t_1}(A_1) \cdots \tau_{t_n}(A_n)) \, , \quad A_1, \ldots, A_n \in \mathcal{A}
\end{equation*}
have analytic continuation to the region:
\begin{equation*}
    \{(z_1, \ldots, z_n) \in \mathbb{C}^n \, : \, 0 < \Im(z_j) - \Im(z_i) < \beta, \quad 1 \leq i < j \leq n\},
\end{equation*}
are bounded and continuous on its closure and fulfil the boundary conditions:
\begin{align*}
    \omega^{\beta}(\tau_{t_1}(A_1) \cdots \tau_{t_{k-1}}(A_{k-1}) \tau_{t_{k} + i \beta}(A_{k}) \cdots \tau_{t_n + i \beta}(A_n))\\
    = \omega^{\beta}(\tau_{t_{k}}(A_{k}) \cdots \tau_{t_n}(A_n) \tau_{t_1}(A_1) \cdots \tau_{t_{k-1}}(A_{k-1}))
\end{align*}
\end{defn}
Condition \ref{def: KMScond} is referred to as \textit{KMS condition} and is motivated by the generalization to QFT of the analogous properties satisfied by a Gibbs state in Quantum Mechanics.\\
Notice that both states are time translation invariant and the ground state, by its analiticity properties, can be understood as a thermal equilibrium state of inverse temperature $\beta = \infty$. Concrete expression of the correlation function of the free KMS state, both for the free Dirac and the free scalar field, are given below respectively in Section \ref{se:dirac-field}
and in Section \ref{se:complex-scalar}. See Formulas \eqref{eq: 2puntiKMS_+} and \eqref{eq: 2puntiKMS_-} for the Dirac field and \eqref{eq: 2pfscalar} for the scalar field.

\subsubsection{Equilibrium states for the interacting dynamics}\label{sec: Int States}
Constructing states for the interacting theory is easily done prescribing the form of the correlation functions among local interacting fields. Namely, given $F_i \in \mathcal{F}_{\text{loc}}$ and any state $\omega$ on the free algebra $\mathcal{A}$, define:
\begin{equation*}\label{partialeqstate}
    \omega^I(F_1, \ldots, F_n) \coloneqq \omega(R_V(F_1) \cdots R_V(F_n))
\end{equation*}
that is a state on $\mathcal{A}[[\lambda]]$, i.e.~on observables of the interacting theory. However, even though $\omega$ might be of equilibrium for the free dynamics at some finite inverse temperature $\beta$, the above prescription does not necessarily provide a corresponding equilibrium state for the interacting dynamics $\tau_t^V$, even if $\supp{F_i}\subset \mathcal{O} \subset \mathbb{M^{+}}$ (the region of $\mathbb{M}$ where $\chi=1$). Nevertheless, if thermalisation occurs (return to equilibrium), an equilibrium state for the interacting theory is expected to be reached for observables supported far in the future. Namely, this corresponds to the large time limit: $\lim_{t\to +\infty}\omega^{\beta}\circ\tau_{t}^{V}$. Though, this limit does not exist in general if the spatial adiabatic limit is considered $h\to 1$. See the next sections and \cite{FaldinoEquilibriumpAQFT} for further details. Nevertheless, despite this problem, in the setting of pAQFT equilibrium states are obtained making use of the approach proposed in \cite{FredenhagenLindnerKMS_2014}, at least for observalbes supported in $\mathbb{M}^+$. In that paper, the authors perturbatively construct an equilibrium state for the interacting theory of a scalar free massive field at finite temperature, generalizing ideas already known for systems with finite degrees of freedom proposed by Araki in \cite{Araki1973RelHamilt}. The connection of that construction with the known Matsubara and real-time formalism used to characterise states at finite temperature, has been analysed in \cite{JoaoNicNic}. The pivotal achievement in \cite{FredenhagenLindnerKMS_2014}, Theorem $1$, is the construction of a map $\mathbb{R}^+ \ni t \mapsto U_V(t) \in \mathcal{A}[[\lambda]]$ unitary, satisfying a cocycle condition: 
\begin{equation*}
    U_V(t) \tau_t(U_V(s)) = U_V(t + s),
\end{equation*}
with generator:
\begin{equation*}
    K \coloneqq -i \frac{d}{dt}U_V(t)\bigg|_{t = 0}.
\end{equation*}
The generator has the following explicit form in terms of the interaction Lagrangian density:
\begin{equation}\label{eq:formaK}
    K = - R_V\bigg(\lambda \int_{\mathbb{M}} \di^4x h(\mathbf{x}) \Dot{\chi}(t) \mathcal{L}_I(x)\bigg).
\end{equation}
In the previous formula $\Dot{\chi} \in \mathcal{C}^{\infty}_0(\mathbb{R})$ denotes the time derivative of the switch-on function and the generator $K$ has to be understood, as always, as a formal power series in the coupling constant $\lambda$. The relevance of the unitary cocycle lies in its intertwining property between the free and interacting dynamics:
\begin{equation*}
    \tau_t^V(F) = U_V(t) \tau_t(F) U_V(t)^{-1} \qquad  t>0, \,\, \supp{F} \subset \mathbb{M^+}. 
\end{equation*}
$U_V(t)$ can be seen, in analogy with quantum mechanics, as the time evolution operator in the interacting picture with $K$ the corresponding interaction Hamiltonian, at least when $\mathcal{L}_I$ does not contain derivatives of the fields. With the cocycle $U_V$ at disposal, and having proved suitable analiticity properties in the parameter $t$ of the function $t\mapsto\omega^\beta(U_V(t) A \, U_V(t)^*)$, for generic $A \in \mathcal{A}$, the equilibrium states at finite inverse temperature $\beta$ for the interacting dynamics $\tau_t^V$ is:
\begin{equation}\label{eq: intGeneraleq}
    \omega^{\beta,V}(A) = \frac{\omega^{\beta}\big(U_V(-i\beta/2)^{-1} \, A \, U_V(i\beta/2)\big)}{\omega^{\beta}(U_V(i\beta))}\, , \quad \forall A \in \mathcal{A}(\mathcal{O}), \quad \mathcal{O}\subset \mathbb{M^+}.
\end{equation}
In particular, for $\beta = \infty$ the ground state for the interacting theory is given by:
\begin{equation}\label{eq: intGround}
    \omega^{\infty,V}(A) = \lim_{\beta \to \infty}\frac{\omega^{\beta}\big(U_V(-i\beta/2)^{-1} \, A \,U_V(i\beta/2)\big)}{\omega^{\beta}(U_V(i\beta))}\, , \quad \forall A \in \mathcal{A}(\mathcal{O}), \quad \mathcal{O}\subset \mathbb{M^+}.
\end{equation}
Indeed, it satisfies by construction the characterization mentioned above after Definition \ref{def: ground}, with respect to the interacting time evolution $\tau_t^V$ in the sense of formal power series.\\
For $0 < \beta < \infty$, making use of the KMS condition for $\omega^{\beta}$ and of the cocycle property, the KMS state with respect to the interacting dynamics can be rewritten in the following form:
\begin{equation}\label{eq: intKMS}
    \omega^{\beta,V}(A) = \frac{\omega^{\beta}\big(A \, U_V(i\beta)\big)}{\omega^{\beta}(U_V(i\beta))}\, , \quad \forall A \in \mathcal{A}(\mathcal{O}), \quad \mathcal{O}\subset \mathbb{M^+}.
\end{equation}
In Theorem $2$ of \cite{FredenhagenLindnerKMS_2014}, this state was proven to satisfy the condition given in Definition \ref{def: KMScond} with respect to $\tau_t^V$ in the sense of formal power series.\\
Notice that the equilibrium states just constructed cease to be quasifree even though the state $\omega^{\beta}$ is.\\
For computational purposes, perturbative expansions of the introduced objects are derived starting from the presented algebraic definitions in terms of \textit{truncated} or \textit{connected} parts $\omega^{\mathcal{T}}$ of the considered state $\omega$ (see \cite{BratteliRobinson} Section $5.2.3$ and \cite{BKR78} Section $2$ for the precise definition). These, are reported in the Appendix \ref{app: espansioni} and will be useful in the following discussion.\\\\
At this point, we are interested in studying the limit $h \to 1$ (with a little nomenclature abuse we refer to it as \textit{adiabatic limit}) for the interacting states, in order to relax the assumption on the spatial support of the interactions. The essential requirement is a sufficiently fast decay rate in the spatial directions of the truncated $n$-point functions, compensating the corresponding infrared limit divergence. In the case of massive real and complex scalar fields these decay rates were shown in the works mentioned above, both for interacting ground and KMS states, and in the case of the Dirac field in Theorem $4.2$ of \cite{Galanda}. Therefore, for these kind of quantum field theories, equilibrium states exist for interactions with arbitrary spatial support.

\section{Secular effects and return to equilibrium}\label{sec: Ret Eq}
\subsection{Return to equilibrium for KMS states and absence of secular growths}\label{sec:req}
In this section we want to collect and generalise some notions and theorems concerning the stability properties of equilibrium states, already known for self interacting neutral scalar fields with spatially compact interaction Lagrangians \cite{FaldinoEquilibriumpAQFT}. In particular, we are interested in understanding under which conditions the property of \textit{return to equilibrium} is satisfied. This property asserts that a state $\omega^\beta$ in thermal equilibrium with respect to the free dynamics naturally converges to $\omega^{\beta,V}$, a KMS state for the interacting dynamics, in the limit of large times:
\begin{equation}
    \lim_{T\to\infty}\omega^\beta(\tau_T^V(A))=\omega^{\beta,V}(A)\;\;\forall{A}\in\mathcal{A}.
\end{equation}
This result holds in the context of $C^*$-dynamical systems if certain clustering properties are satisfied by the state $\omega^\beta$, see e.g. Section $5.4.2$ in \cite{BratteliRobinson} for a detailed review. Furthermore, it also holds in the context of perturbative interacting quantum field theory if certain conditions are satisfied, see Theorem $3.4$ in \cite{FaldinoEquilibriumpAQFT} whose proof can be also obtained adapting the proof of Theorem 2 in \cite{BKR78}. At each perturbative order, when the property of return to equilibrium is satisfied, it is possible to obtain expectation values on equilibrium states for the interacting dynamics without having to directly consider the effect of the interaction on the state itself. In particular, since the state obtained taking the large time limit is an equilibrium state for the interacting dynamics, we are ensured that in the perturbative expansions of the expectation values no term dependent on time translation (as, for instance, terms growing in time), will survive to the large time limit. More specifically, we have:
\begin{thm}\label{thm: NoSecEq}
Let $\omega^{\beta}$ be an equilibrium state with respect to the free dynamics $\tau_t$. Given an interaction $V$ (Equation \eqref{eq:formaV}) and the corresponding interacting dynamics $\tau_t^V$, if return to equilibrium holds at all orders in perturbation theory: 
\begin{equation*}
    \lim_{T\to\infty} \omega^\beta(\tau^V_T (A))  = \omega^{\beta,V}(A)\, , \quad \forall A \in \mathcal{A}
\end{equation*}
then secular effects are absent. Here, $\omega^{\beta,V}$ is the equilibrium state for the interacting dynamics.
\end{thm}
\begin{proof}
The proof of this theorem is an immediate consequence of the fact that a state which satisfies the KMS condition is invariant under time translations. 
\end{proof}
The property of return to equilibrium, that guarantees the absence of secular growths, as a consequence of the thermalisation process, has been proven to hold in Theorem $3.4$ of \cite{FaldinoEquilibriumpAQFT} for a massive real scalar field with spatially compact interaction:

\begin{thm}\label{thm: thermalisation}
Let $\omega^{\beta}$ be the KMS state of a real scalar QFT with respect to the free dynamics $\tau_t$, at finite inverse temperature $\beta$. Then, the state is stable under perturbations described by a spatially compact interaction $V$ of the kind \eqref{eq:formaV}. Namely:
\begin{equation*}
    \lim_{T \to \infty} \omega^{\beta}(\tau_T^V(A)) = \omega^{\beta, V}(A)\;\;\forall A\in\mathcal{A}.
\end{equation*}
\end{thm}
Once the large time limit is taken, also the adiabatic limit can be performed relying on the results obtained in \cite{FredenhagenLindnerKMS_2014}:
\begin{equation*}
    \lim_{h \to 1} \lim_{T \to \infty} \omega^{\beta}(\tau_T^V(A)) = \lim_{h \to 1} \omega^{\beta, V}(A).
\end{equation*}
This shows that, once thermalisation is reached, it cannot be altered by extending the spatial support  of the interaction. On the contrary, it is natural to ask whether or not the property of return to equilibrium is always guaranteed, in particular in the presence of non-compactly spatially supported interactions. A negative answer to this question has been already given in Section 4 of \cite{FaldinoEquilibriumpAQFT}, where the authors provide some explicit example (e.g.~perturbation of the mass) of situations in which the KMS state of the free theory does not thermalise to a KMS state of the interacting theory. This is again a consequence of the infinitely extended spatial support of the interaction. Therefore, if we studied from the beginning $\lim_{h \to 1} \omega^{\beta}(\tau_t^V(A))$, the hypotheses of Theorem \ref{thm: NoSecEq} could be violated leading to possible occurrence of secular effects. In the next subsections, we extend this result concerning the property of return to equilibrium to the case of a complex scalar field and of a fermionic field. In addition, in Section \ref{se:applications}  we will also provide new physically interesting examples in which the presence of an interaction supported everywhere in space breaks the property of return to equilibrium leading to secular effects. As we will directly check, these problems can be avoided taking in consideration the effect of the interaction on the state, i.e.~using from the very beginning the KMS state for the interacting theory.

\subsubsection{Dirac field}\label{se:dirac-field}
The beginning of the section is devoted to briefly present the main differences between the Dirac and the scalar case. First, the Hadamard two-point functions derived from the Hadamard parametrix of the free scalar field on Minkowski spacetime are (see Chapter 17 of \cite{gerardbook}):
\begin{equation*}
    \cancel{S}_2^{\pm}(z) = \pm (i\cancel{\partial}_x + m) \omega_2^{\infty}(\pm z), \qquad  z=x-y,
\end{equation*}
where $\omega_2^{\infty}(z)$ is the two-point function of the ground state for the free scalar field and $\cancel{S}_2^{\pm}(z)$ are the two-point functions defining the ground state $\cancel{\omega}^{\infty}$ for the free Dirac field. $\cancel{S}_{2}^{\pm}$ are checked to satisfy the Dirac equation with $\cancel{S}$ the Pauli-Jordan function (causal propagator):
\begin{equation*}
    \cancel{S}_2^+(x,y) + \cancel{S}_2^{-}(x,y) = i \cancel{S}(x,y).
\end{equation*}
The equilibrium state with respect to Minkowski time translations at inverse temperature $\beta$ of the free Dirac field, $\cancel{\omega}^{\beta}$, has associated two-point functions denoted by $\cancel{S}^{\beta, \pm}_2$. These, are derived imposing the KMS condition.\\
For later purposes, we write here explicitly the two-point functions of the ground state\footnote{Here, the spinor indices are kept implicit for the sake of simplicity.}:
\begin{align}
    \cancel{S}_2^{+}(x - y) &\coloneqq \cancel{\omega}^{\infty}(\psi(x) \psi^*(y)) \nonumber\\
    &= \frac{1}{(2 \pi)^3} \int \frac{\di^3\mathbf{p}}{2 \en{p}} (-\gamma^0 \en{p} - \gamma^i p_i + m)e^{-i(\en{p} (x_0 - y_0) - \mathbf{p} (\mathbf{x} - \mathbf{y}))},\label{eq: 2puntiGro_+}\\
    \cancel{S}_2^{-}(x - y) &\coloneqq \cancel{\omega}^{\infty}(\psi^*(y) \psi(x)) \nonumber\\
    &= -\frac{1}{(2 \pi)^3} \int \frac{\di^3\mathbf{p}}{2 \en{p}} (\gamma^0 \en{p} + \gamma^i p_i + m)e^{i(\en{p} (x_0 - y_0) - \mathbf{p} (\mathbf{x} - \mathbf{y}))},\label{eq: 2puntiGro_-}
\end{align}
and KMS states:
\begin{align} 
    \cancel{S}^{\beta, +}_{2}(x-y) &\coloneqq \cancel{\omega}^{\beta}(\psi(x) \psi^*(y)) \nonumber\\
    &= \frac{1}{(2 \pi)^3} \int \frac{\di^3\mathbf{p}}{2 \en{p}} \bigg( \frac{(-\gamma^0 \en{p} - \gamma^i p_i + m)e^{-i\en{p} (t_x - t_y)}}{(1 + e^{-\beta \en{p}})} - \frac{(\gamma^0 \en{p} - \gamma^i p_i + m)e^{i\en{p} (t_x - t_y)}}{(1 + e^{\beta \en{p}})}  \bigg) e^{i \mathbf{p} (\mathbf{x}- \mathbf{y})},\label{eq: 2puntiKMS_+}\\
    \cancel{S}^{\beta, -}_{2}(x-y) &\coloneqq \cancel{\omega}^{\beta}(\psi^*(y) \psi(x)) \nonumber\\
    &= \frac{1}{(2 \pi)^3} \int \frac{\di^3\mathbf{p}}{2 \en{p}} \bigg(\frac{(-\gamma^0 \en{p} - \gamma^i p_i + m)e^{-i\en{p} (t_x - t_y)}}{(1 + e^{\beta \en{p}})} - \frac{(\gamma^0 \en{p} - \gamma^i p_i + m)e^{i\en{p} (t_x - t_y)}}{(1 + e^{-\beta \en{p}})}  \bigg) e^{i \mathbf{p} (\mathbf{x}- \mathbf{y})},\label{eq: 2puntiKMS_-}
\end{align}
where $\psi, \psi^* \in \mathcal{F}_{\text{loc}}$ are the linear Dirac fields and $\en{p} = \sqrt{|\mathbf{p}|^2 + m^2}$. Interacting KMS states at finite inverse temperature $\beta$ are constructed for Fermi fields following the steps presented in \cite{Galanda}. This leads to the perturbative expansions reported in Appendix \ref{app: espansioni}, where the states for the free theory have to be understood as the one just described for the fermionic case. 

We now want to generalize Theorem \ref{thm: thermalisation} for interacting Dirac field theories. In order to do so, we first need to prove the clustering properties for expectation values of compactly supported observables at large time separation:
\begin{lem}[Clustering for $\tau_t$]\label{lem: ClusteringFree}
Consider any two interacting observables $A,B \in \mathcal{A}$ of a Dirac quantum field theory, supported after the switching on of the interaction represented by a potential $V$ of the kind given in Equation \eqref{eq:formaV}. Then, given a state $\cancel{\omega}^{\beta}$ which is KMS with respect to the free dynamics $\tau_t$ at finite inverse temperature $\beta$, it holds:
\begin{equation*}
    \lim_{T \to \infty} \cancel{\omega}^{\beta}(A \, \tau_{T}(B)) = \cancel{\omega}^{\beta}(A)\cancel{\omega}^{\beta}(B),
\end{equation*}
where the convergence of this limit is in the sense of formal power series.
\end{lem}
\begin{proof}
The proof of the Lemma follows in every step the proof of Proposition $3.1.$ in \cite{FaldinoEquilibriumpAQFT} (i.e. the analogous statement for the real scalar case). In particular, all the general assumptions concerning the compact support of the interacting observables are satisfied, as well as the Hadamard properties of the thermal state. The only non trivial generalization to the fermionic case consists in the proof of the large time decay property of the fermionic thermal two point function for timelike separation. This result is obtained in Proposition \ref{app: 1} of the present paper.
\end{proof}
Once the clustering property for $\tau_t$ is established, also the clustering property for the interacting dynamics $\tau_t^V$ is proven:
\begin{equation*}
    \lim_{T \to \infty} \bigg[ \cancel{\omega}^{\beta}(A \, \tau_T^V(B)) - \cancel{\omega}^{\beta}(A) \cancel{\omega}^{\beta}(\tau_T^V(B)) \bigg] = 0,
\end{equation*}
In fact, the proof given in Proposition $3.3$ of \cite{FaldinoEquilibriumpAQFT}, for the real scalar field case, relies on the considered QFT just in the statement for the $\tau_t$ clustering. Finally, return to equilibrium for a KMS state $\cancel{\omega}^{\beta}$ in the form of Theorem \ref{thm: thermalisation}, is proven with the clustering condition for $\tau_t^V$. Again, the reason is that the proof of Theorem $3.4.$ of the same above cited work does not depend on the considered type of QFT.

\subsubsection{Complex scalar field}\label{se:complex-scalar}
We now briefly recall the main properties of the charged (complex) scalar field $\varphi$. The aim is to explicitly provide the form of the two-point function of the equilibrium state with respect to the free time evolution at inverse temperature $\beta$. We recall that the complex scalar field $\varphi$ can be seen as the linear combination of two independent real massive Klein-Gordon scalar fields $\varphi_1$ and $\varphi_2$. Hence, the canonical commutation relations 
\[
    \left[\varphi(x),\varphi^\dagger(y)\right] = i\Delta(x,y)\, , \quad \left[\varphi(x),\varphi(y)\right] =0 \, , \quad \left[\varphi^\dagger(x),\varphi^\dagger(y)\right] =0
\]
are given in term of $\Delta(x,y)$ which is the usual Pauli-Jordan propagator (commutator function) of the Klein--Gordon massive field on Minkowski spacetime:
\begin{equation*}
    \Delta(x-y)=-\int\frac{\di^3\textbf{k}}{(2\pi)^3} \frac{\sin(\en{k}(t_x-t_y))}{\en{k}}e^{i\mathbf{k}(\mathbf{x-y})},
\end{equation*}
where $\en{k}=\sqrt{|\mathbf{k}|^2+m^2}$. We observe that using the commutation relations and the KMS property of the state, together with its invariance under space-time translations, we can derive the following expression for the two-point function:
\begin{align}\label{eq: 2pfscalar}
    \omega_2^{\beta,+}(x,y)&\coloneqq\omega^\beta(\varphi(x)\varphi^\dagger(y))=
    \omega_2^{\beta}(x-y) \nonumber\\
    &=\int \frac{1}{(2\pi)^3}\left[\frac{e^{-i\en{p}(t_x-t_y)}}{(1-e^{-\beta\en{p}})2\en{p}}-\frac{e^{i\en{p}(t_x-t_y)}}{(1-e^{\beta\en{p}})2\en{p}}\right] e^{i\mathbf{p}(\mathbf{x-y})}\;\di^3\textbf{p}.
\end{align}
To determine an expression for the two-point function $\omega^\beta(\varphi^\dagger(x)\varphi(y))$ we just need to compute a complex conjugation and to make use of the commutation relation and the explicit expression for the two-point function $\omega_2^\beta(x,y)$, to obtain that:
\begin{equation*}
    \omega_2^{\beta,-}(x,y)\coloneqq \omega(\varphi^\dagger(x)\varphi(y))=\omega_2^\beta(x,y).
\end{equation*}
We observe that both the two-point functions of the complex scalar field coincide with the KMS two-point function of the real scalar field. 
Therefore, all the results concerning the construction of the interacting KMS state directly extend from the real to the complex case. The same holds in particular for the return to equilibrium properties we are interested in.

\subsection{Absence of secular growths for general states}
After the previous discussion for equilibrium states, general conditions for the absence of secular growths are studied in this subsection.\\ 
The starting point is the following definition of a class of functions, with certain integrability and boundedness properties, useful for the statement of the main theorem: 

\begin{defn}\label{def:secularily-bounded}
A function $f(t_1, \ldots, t_n)$ of real variables $(t_1, \ldots, t_n) \in \mathbb{R}^n$
is called \textbf{secularily bounded} if it satisfies both the following two conditions:\begin{itemize}
    \item[i)] $f$ is absolutely integrable on $\mathbb{R}^{n-1}$ in the variables $t_{p_1}, \ldots, t_{p_{n-1}}$ for any choice of $\{p_1, \ldots, p_{n-1}\} \subset \{1, \ldots, n\}$ for any permutation $\{p_1, \ldots, p_{n}\} $ of the set $ \{1, \ldots, n\}$, and for every $t_{p_n}\in\mathbb{R}$:
    \begin{equation*}
    \int_{\mathbb{R}^{n-1}} |f(t_1, \ldots, t_n)| \di t_{p_1} \cdots \di t_{p_{n-1}} < \infty.
\end{equation*}
\item[ii)]  The function:
\begin{equation*} 
    g(t_{p_n}) \coloneqq \int_{\mathbb{R}^{n-1}} |f(t_1, \ldots, t_n)| \di t_{p_1} \cdots \di t_{p_{n-1}}
\end{equation*}
satisfies a bound $|g(t_{p_n})| \leq C_1$ for $C_1 \in \mathbb{R}^+$, for every $ p_n\in\{1,\ldots,n\}$.
\end{itemize}
\end{defn}

With this definition at hand, we present the main result of this section. We show that, for adiabatically switched on interactions, outside the limit $h \to 1$ secular effects are absent on states with secularily bounded truncated $n$-point functions.\begin{thm}\label{th:truncated}
    On Minkowski spacetime $(\mathbb{M} = \mathbb{R} \times \Sigma, \eta)$, consider an interaction $V$:
    \begin{equation*}
        V \coloneqq  \lambda \int_{\mathbb{M}} \chi(t) h(\mathbf{x}) \mathcal{L}_I \di^4x,
    \end{equation*}
    with $h \in \mathcal{C}^\infty_0(\Sigma)$ and $\chi \in \mathcal{C}^{\infty}(\mathbb{R})$ a smooth switch-on \eqref{eq:accad}. 
    Outside the adiabatic limit, consider $K \in \mathcal{A}[[\lambda]]$, the generator of the cocycle intertwining $\tau_t$ with $\tau_t^V$, defined as:
    \begin{equation*}
        K=-R_V\bigg(\lambda \int_{\mathbb{M}} \dot{\chi}(t) h(\mathbf{x}) \mathcal{L}_I(x) \di^4x\bigg).
    \end{equation*}
   Let $\omega$ be a state on $\mathcal{A}$ such that for every $\{A_i\}_{i\in\{1,\dots, n\}}\in \mathcal{A},n\in\mathbb{N}$ the function:
    \begin{equation}\label{eq:truncated-function-hypothesis}
    f_{A_1\dots A_n}(t_1,\dots,{t_n}) \coloneqq \omega^{\mathcal{T}}(\tau_{t_1}(A_1)\otimes \dots \otimes \tau_{t_n}(A_n)) 
    \end{equation}
   is secularily bounded, for $\omega^{\mathcal{T}}$ the truncated or connected functions of $\omega$. Then, at each fixed perturbative order $N$ and for every $A\in\mathcal{A}$ it exists a constant $C_N \in \mathbb{R}^+$ such that the following uniform bound holds:
    \[
    \left|\omega^{(N)}(\tau^V_t(A))\right| \leq C_N,
    \]
    where $\omega^{(N)}(\tau^V_t(\cdot))$ denotes the expansion at $N$-th order in $K$ of $\omega(\tau^V_t(\cdot))$. In particular, no secular effects are present.
\end{thm}

\begin{proof}
To study $\omega(\tau_t^{V}(A))$ we start considering the perturbative expansion of the interacting dynamics (see Equation \eqref{eq: espansDynam}):
\[
\tau_t^{V}(A) = \tau_t(A) + \sum_{n\geq 1}i^n 
\int_{t S_n}[ \tau_{t_1}(K),[\dots ,[\tau_{t_n}(K),\tau_t(A) ]\dots  ] ]
\di t_1 \dots \di t_n,
\]
where $tS_n=\{(t_1,\dots, t_n) \in \mathbb{R}^n | (0<t_1<t_2<\dots <t_n<t)  \}$ is the $n$ dimensional simplex.\\
The iterated commutator can be more conveniently written as: 
\[
B\coloneqq [A_n,[\dots, [A_1,A_0]\dots]] = \sum_{I\subset  \{1,\dots, n\} }
(-1)^{|I^c|} \left(\prod_{i\in \overleftarrow{I}}A_i \right)   A_0 \left(\prod_{j\in \overrightarrow{I^c}} A_j \right),
\]
where the sum is over all the possible subsets $I$ of indices in $\{ 1, \ldots, n\}$, $I^c$ is the complementary of $I$ in $\{ 1, \ldots, n\}$, $|I|$ is the cardinality of the set and, to emphasise that the indices are taken with increasing or decreasing order, we adopted the notation $\overrightarrow{I^c}$ respectively $\overleftarrow{I}$.\\
Hence, for $B$ as defined above:
\begin{equation*}
    \omega(B)=\sum_{I\subset \{1,\dots, n\} }(-1)^{|I^c|} \omega\left(\prod_{j\in (\overleftarrow{I}\cup\{0\}\cup \overrightarrow{I^c})} A_j\right) = \sum_{I\subset\{1,\dots,n\}} (-1)^{|I^c|}\omega(\overleftarrow{I}\cup\{0\}\cup \overrightarrow{I^c})
\end{equation*}
where we used the compact notation $\omega(J) = \omega(\prod_{j\in J} A_J)$,  with $J$ an ordered set of indices. If we similarly denote $\omega^{\mathcal{T}}(J) = \omega^{\mathcal{T}} (\prod^{\otimes}_{j\in J} A_J)$, we can compare $\omega(B)$ with $\omega^{\mathcal{T}}(B)$:
\begin{equation}\label{eq:decomposition}
\omega(B) - \omega^{\mathcal{T}}(B)=\sum_{I\subset\{1,\dots,n\}} (-1)^{|I^c|}
\big(\omega(\overleftarrow{I}\cup\{0\}\cup \overrightarrow{I^c})-\omega^{\mathcal{T}}(\overleftarrow{I}\cup\{0\}\cup \overrightarrow{I^c}) \big).
\end{equation}
We shall now prove that $\omega(B) = \omega^{\mathcal{T}}(B)$ showing that the terms in the sum at the right hand side of the previous equation cancel in pairs. To this end, the recursive definition of truncated functions (see e.g. Section 2 of \cite{BKR78}) implies that the generic element in the right hand side can be factorised in the following way:
\[
\omega(J)-\omega^{\mathcal{T}}(J) = \sum_{J_0\subset J,\;  0\in J_0} \omega^{\mathcal{T}}(J_0)\omega(J\setminus J_0).
\]
In our case, the ordered set $J = \overleftarrow{I}\cup\{0\}\cup \overrightarrow{I^c}$. We observe that in the previous factorisation: 
\[
J_0=\overleftarrow{J_0^1} \cup \{0\}\cup \overrightarrow{J_0^{2}} \qquad \text{and} \qquad
J\setminus J_0 = \overleftarrow{K^1} \cup  \{s\} \cup \overrightarrow{K^2}
\]
where both $J_0^1\subset I$, $K^1\subset I$ and $s$ is the smallest element of $J \setminus J_0$.\\
Fix now a $J_0^1$ and a $K^1$ and consider the decomposition 
\[
\omega^{\mathcal{T}}(\overleftarrow{J_0^1} \cup \{0\}\cup \overrightarrow{J_0^{2}})\omega(\overleftarrow{K^1} \cup  \{s\} \cup \overrightarrow{K^2}).
\]
This decomposition is present only in two elements of the sum  at the right hand side of \eqref{eq:decomposition}. One is indexed by $I=J_0^1\cup K^1 \cup \{ s \}$ with $I^c=J^2_0\cup K^2$, while the other is indexed by $\tilde{I}=J_0^1\cup K^1 $ with $\tilde{I}^c=\{s\}\cup J^2_0\cup K^2$. Since
$|\tilde{I}^c|=|I^c|+1$, these two elements appear with opposite sign in the sum and, therefore, they mutually cancel. This proves that:
\[
\omega(\tau_t^V(A))
=
\omega^{\mathcal{T}}(\tau_t(A)) + \sum_{n\geq 1} i^n \int_{tS_n} \omega^{\mathcal{T}}([ \tau_{t_1}(K),[\dots ,[\tau_{t_n}(K),\tau_t(A) ]_\otimes \dots  ]_\otimes ]_\otimes)   \di t_1 \dots \di t_n.
\]
The lowest perturbative order which appears in the expansion of $K$ is $1$. Hence, at fixed perturbative order, only a finite number of terms in the sum over $n$ can contribute to $\omega(\tau^V_t(A))$. Furthermore, each of the addends is estimated using assumption \eqref{eq:truncated-function-hypothesis}. Therefore, these are absolutely integrable functions in the first $n$ variables with resulting primitive that is uniformly bounded in the remaining variable (the $n+1$-th included). The proof follows.
\end{proof}
As a first consequence of this theorem, we have that all the states that are time invariant under the free evolution and show a reasonable decay property for big time separations, cannot feature secular growths.
\begin{cor}\label{cor: NoSecInvariant}
Consider a quasifree, Hadamard state $\omega$ invariant under the free time translation $\tau_t$. Moreover, assume that its two-point function satisfies the following bound for timelike future pointing separations $y-x$: 
\begin{equation}\label{eq:clustering=Derivative}
    \big| \partial_x^{(\alpha)} \partial_y^{(\beta)}\omega_2(x;y_0+t,\textbf{y})\big| \leq \frac{C}{t^{1 + \epsilon}},\;\;\; t> 1+(x_0-y_0)+|\textbf{x}-\textbf{y}|,
\end{equation}
where $\alpha,\beta\in\mathbb{N}^4$ are multi-indices and $\partial^{(\alpha)}_x, \partial^{(\beta)}_y$ are partial derivatives of order $\alpha_i, \beta_j$ in the variable $x_i, y_j$, for $i,j\in\{1,2,3,4\}$ and $C$ is a constant which could depend on $\alpha$ and $\beta$. Then, for each set of elements $A_1,\dots, A_n \in \mathcal{A}$ and $(t_1, \ldots, t_n) \in \mathbb{R}^n$ all the truncated functions:
\[
f_{A_1\dots A_n}(t_1,\dots,{t_n}) = \omega^{\mathcal{T}}(\tau_{t_1}(A_1)\otimes \dots \otimes \tau_{t_n}(A_n)) 
\]
are secularily bounded.
\end{cor}
\begin{proof}
We present the proof for the scalar case in full details, the generalization to other fields follows with minor adjustments. If $\omega$ is invariant under time translations also $\omega^{\mathcal{T}}$ is invariant under time translation.
Hence:
    \[
    f_{A_1\dots A_n}(t_1, \dots, t_j,\dots, t_n)
    =
    f_{A_1\dots A_n}(t_1-t_j, \dots,0,\dots, t_n-t_j).
    \]
Namely, by the translation invariance, $f$ can be seen as a function of $n-1$ variables, i.e.~the time separations. The truncated $n$-point function, by the quasifree assumption, admits a diagrammatic expansion via connected and oriented graphs:
\begin{align}
    f_{A_1\dots A_n}&(t_1-t_j, \dots,0,\dots, t_n-t_j) = \nonumber\\
    \bigg(\mathcal{M} \sum_{G \in \mathcal{G}^{c,o}_{n}} c_G \prod_{l \in E(G)} D_{s(l) r(l)} &(\tau_{t_1 - t_j}(A_1) \otimes \ldots \otimes A_j \otimes \ldots \otimes \tau_{t_n - t_j}(A_n)) \bigg)\bigg|_{\phi_1 = \ldots = \phi_n = 0}.\label{eq: truncinvariant}
\end{align}
Here, we have used the representation of $\mathcal{A}$ in $\mathcal{A}_{\omega}$ presented in \eqref{eq:omega-deformation}
and we have kept implicit the application of $\alpha$ in \eqref{eq:alpha}. Furthermore, $c_G$ is a constant which includes the symmetry factor of the diagrammatic expansion and $\mathcal{G}^{c,o}_n$ is a notation for the class of connected, oriented (directed) graphs with $n$ vertices. Moreover, $E(G)$ is the set of lines of $G$ and $D_{s(l) r(l)} : \mathcal{A}_{\omega} \otimes \mathcal{A}_{\omega} \to \mathcal{A}_{\omega} \otimes \mathcal{A}_{\omega}$ is a functional differential operator associated to the line $l$, with source and range vertices $s(l), r(l) \in \{ 1, \ldots, n \}$, that for real scalar fields is of the following type:
\begin{align}
    D_{s(l) r(l)} &= \int \di^4x \, \di^4y \, \omega_{2}(x, y) \frac{\delta}{\delta \phi_{s(l)}(x)} \otimes \frac{\delta}{\delta \phi_{r(l)}(y)} \label{stima_coro}\\
    &= \left \langle \omega_2, \frac{\delta}{\delta \phi_{s(l)}} \otimes \frac{\delta}{\delta \phi_{r(l)}} \right \rangle. \nonumber
\end{align}
where $\frac{\delta}{\delta\phi_{i}}$ is the functional derivative of the $i$-th element in the tensor product of the $n$ observables. So, the product in \eqref{eq: truncinvariant}, is meant as a composition.
We now observe that, under our hypothesis on the two-point function, for every $t_1,t_2\in\mathbb{R}$:
\begin{equation*}
      \left| \left< \omega_2, \frac{\delta}{\delta\phi_1}\tau_{t_1}(A_1) \otimes \frac{\delta}{\delta\phi_2}\tau_{t_2}(A_2) \right>  \right| \leq \frac{C}{(|t_1-t_2|+d)^{1+\epsilon}}\;\;\; \forall A_1,A_2\in 
      \mathcal{A}_{\omega},
\end{equation*}
where the constants $C,d$ may depend on $A_1,A_2$. This is proven following Proposition A.2 of \cite{FaldinoEquilibriumpAQFT}. In particular, since on one side $\omega$ is a time translation invariant Hadamard state, for $|t_1 - t_2|$ large enough $\omega_2(x_0 + t_1, \mathbf{x};y_0 + t_2, \mathbf{y})$ is a smooth function. On the other side, for small $|t_1 - t_2|$ the product of the same two-point function with the microcausal functionals defines a compactly supported distribution. The two cases are joint using a continuity argument of distributions.\\
Finally, in order to show the secular boundedness on $\mathbb{R}^n$ of the truncated $n$-point function, we split $\mathbb{R}^n$ in a union of finitely many simplices, e.g.~$S^{\infty}_n = \{ (t_1, \ldots, t_n) | 0 \leq t_1 \leq t_2 \leq \ldots \leq t_n \leq \infty \}$. \\
Therefore, for $(t_1, \ldots, t_n) \in S_{n}$ (any simplex) we have:
\begin{gather*}
    |f_{A_1\dots A_n}(t_1-t_j, \dots,0,\dots, t_n-t_j)| \leq \sum_{G \in \mathcal{G}^{c,o}_{n}} c_G \prod_{l \in E(G)} \frac{C}{(|t_{s(l)} - t_{r(l)} | + d)^{1+\epsilon}}\\
     \leq C_1
    \bigg(\frac{1}{(|t_{1} - t_{2}| + d) (|t_{2} - t_{3}| + d) \cdots (|t_{j-1} - t_{j}| + d) (|t_{j} - t_{j+1}| + d) \cdots (|t_{n-1} - t_{n}| + d)}\bigg)^{1+\epsilon}.
\end{gather*}
for some constant $C_1$. Here we have used the fact that the number of graphs is finite and that the graphs are all connected. Moreover, since $(t_1, \ldots, t_n) \in S_{n}$ the contribution of each line is estimated with the one of a line joining consecutive vertices. Finally, in the last expression, all time variables are explicit. We have also, without loss of generality, assumed that each vertex has at most two lines attached to it, as the general case gives just a better behaved decay estimate.\\
To show that this is a positive absolutely integrable function in $n-1$ variables and, after the integration, bounded in the $n$-th variable consider:
\begin{equation*}
    \int_{S_n} \di t_{n-1} \cdots \di t_{1}\prod_{i=1}^{n-1} \frac{1}{(|t_{i} - t_{i+1}| + d)^{1+\epsilon}} \leq \int_{\mathbb{R}^{n-1}} \di t_{n-1} \cdots \di t_{1} \prod_{i=1}^{n-1} \frac{1}{(|t_{i} - t_{i+1}| + d)^{1+\epsilon}} = \bigg( \frac{2}{\epsilon d^{\epsilon}} \bigg)^{n-1},
\end{equation*}
where the positivity of the integrand has been explicilty used in the estimate and to apply Tonelli's theorem. The required boundedness is then proven.
\end{proof}
\begin{rem}
Thermal equilibrium (KMS) and ground states satisfy the hypotheses of Corollary \ref{cor: NoSecInvariant}, for $\epsilon = 1/2$, and thus of the Theorem \ref{th:truncated}. This has been shown in Proposition A$.1$ and Proposition $3.3$ of \cite{FaldinoEquilibriumpAQFT} for the real scalar field and it has been generalised in this paper in Proposition \ref{app: 1}, for the case of the Dirac field. Also Non Equilibrium Steady States (NESS), as introduced in \cite{ruelle2000natural, JP02}, are checked to satisfy the hypothesis of the above corollary.
\end{rem}
Finally, as the Hadamard condition is believed to be necessary if we want a state to be of physical relevance (\cite{KW91, Wald1995}), we want to translate the general condition on the $n$-point function, given in Theorem \ref{th:truncated}, in a condition on the two-point function for the situation in which the state is  assumed to be quasifree and not necessarily time translation invariant. On Minkowski spacetime for time translation invariant field equations and by the Hadamard assumption, this translates in a condition only on the smooth part of the state:
\begin{thm}\label{theo:time dependent}
Let $\omega$ be a Hadamard quasifree state with smooth part $W(x,y)$. Assume that for any $k \in \mathbb{N}$, $\alpha$ multi-index:
\begin{equation}\label{eq: assunzione1}
    \sum_{|\alpha| \leq k} \sup_{(x,y) \in \mathbb{M} \times \mathbb{M}} | \partial^{\alpha} W(x,y) | < + \infty.
\end{equation}
and, upon writing $x = (t_x, \mathbf{x})$ and $y=(t_y,\mathbf{y})$, assume that for any $p,q \in \mathbb{N}$ multi-indexes the following conditions hold:
\begin{equation}\label{eq: assunzione2}
    \int_{\mathbb{R}}\di t_x |\partial^{p}_{x} \partial^{q}_{y}W(x,y)| < C^{p,q}_{\mathbf{x}, \mathbf{y}} \quad \land \quad \int_{\mathbb{R}}\di t_y |\partial^{p}_{x} \partial^{q}_{y} W(x,y)| < C'^{p,q}_{\mathbf{x}, \mathbf{y}},
\end{equation}
where $C_{\mathbf{x}, \mathbf{y}}$ respectively $C'_{\mathbf{x}, \mathbf{y}}$ are bounded functions of the space variables, uniform in $t_y$ respectively $t_x$. Then, the truncated $n$-point functions are secularily bounded.
\end{thm}
\begin{rem}
The ideas behind the conditions for the smooth part of the state are the following. Condition \eqref{eq: assunzione1} is needed to be able to break general connected graphs in $1$-particle reducible diagrams while \eqref{eq: assunzione2} to argue that the integrals along the various lines are uniformly bounded.\\
Assumption \eqref{eq: assunzione2} for any $p,q \in \mathbb{N}$, assuring the secular boundedness, can be relaxed asking the condition just up to a fixed number of derivatives depending on the observables.
\end{rem}
\begin{proof}
We present the proof for the scalar case in full details, the generalization to other fields follows with minor adjustments. For every $A_1,\dots ,A_n \in \mathcal{A}$, we focus on 
\begin{equation*}
\omega^{\mathcal{T}}(\tau_{t_1}(A_1)\otimes \dots \otimes \tau_{t_n}(A_n)) =   \mathcal{M}\left( \sum_{G \in \mathcal{G}^{c,o}_{n}} c_G \prod_{1\leq i<j \leq n} D_{ij}^{l_{ij}} (\tau_{t_1}(A_1) \otimes \ldots \otimes \tau_{t_n}(A_n)) \right)_{\phi_1 = \ldots = \phi_n = 0},
\end{equation*}
where we have used the representation of $\mathcal{A}$ in $\mathcal{A}_{\omega}$ presented in \eqref{eq:omega-deformation} and we have kept implicit the application of $\alpha$ in \eqref{eq:alpha}. $ D_{ij} : \mathcal{A}_\omega \otimes \mathcal{A}_\omega \to \mathcal{A}_\omega \otimes \mathcal{A}_\omega$ is a functional differential operator associated to the vertices $i,j$ that are joined by $l_{ij}$ many lines, that for real scalar field is:
\begin{equation*}
    D_{ij}=\left\langle \omega_2,\frac{\delta}{\delta\phi_i}\otimes\frac{\delta}{\delta\phi_j}\right\rangle.
\end{equation*}
The contribution $\omega^{\mathcal{T}}_G$ to $\omega^{\mathcal{T}}(\tau_{t_1}(A_1)\otimes \dots \otimes \tau_{t_n}(A_n))$, due to a single connected graph $G \in \mathcal{G}_n^{c,o}$, can be written out making explicit the two-point functions:
\begin{equation*}
\omega^{\mathcal{T}}_G = 
\int_{\mathbb{M}^n} \bigg(\prod_{1\leq i<j \leq n}\prod_{\alpha_{ij},\beta_{ij}} \left(\partial_{x_i}^{\alpha_{ij }}\partial_{x_j}^{\beta_{ij}}\omega_{2}(x_{i}, x_j)\right)^{m_{\alpha_{ij}\beta_{ij}}}\bigg) \Lambda(x_1 +t_1 e^0, \ldots, x_n + t_n e^0) \di^4x_{1} \cdots \di^4x_n,
\end{equation*}
where the indices $\alpha_{ij},\beta_{ij},m_{\alpha_{ij}\beta_{ij}}$ are fixed natural numbers and the index $m_{\alpha_{ij}\beta_{ij}}$ satisfies the condition $\sum_{\alpha_{ij},\beta_{ij}}m_{\alpha_{ij}\beta_{ji}}=l_{ij}$ for each couple $i,j$ fixed. Moreover, $e^0$ is the unit timelike future pointing vector and $\Lambda(x_1 +t_1 e^0, \ldots, x_n + t_n e^0)$ is a compactly supported smooth function. By the Hadamard assumption, we know:
\begin{equation*}
    \omega_2(x_i,x_j) = \omega_2^{\infty}(x_i - x_j) + W(x_i, x_j),
\end{equation*}
with $\omega_2^{\infty}$ the ground state two-point function. Therefore, by Newton's binomial formula:
\begin{gather*}
    \prod_{\alpha_{ij},\beta_{ij}} \left(\partial_{x_i}^{\alpha_{ij }}\partial_{x_j}^{\beta_{ij}}\omega_{2}(x_{i}, x_j)\right)^{m_{\alpha_{ij}\beta_{ij}}} = 
    \\ \prod_{\alpha_{ij}, \beta_{ij}}\sum_{p_{\alpha_{ij}\beta_{ij}}=0}^{m_{\alpha_{ij}\beta_{ij}}}\binom{m_{\alpha_{ij}\beta_{ij}}}{p_{\alpha_{ij}\beta_{ij}}} \bigg( \big(\partial_{x_i}^{\alpha_{ij }}\partial_{x_j}^{\beta_{ij}} \omega_2^{\infty}(x_i - x_j)\big)^{p_{\alpha_{ij}\beta_{ij}}} \big( \partial_{x_i}^{\alpha_{ij }}\partial_{x_j}^{\beta_{ij}} W(x_i, x_j)\big)^{m_{\alpha_{ij}\beta_{ij}}-p_{\alpha_{ij}\beta_{ij}}} \bigg).
\end{gather*}
Then, the boundedness is studied for the three separate types of contributions:
\begin{align}
    \int_{\mathbb{M}^n} \bigg(&\prod_{1\leq i<j \leq n}\prod_{\alpha_{ij},\beta_{ij}} \left(\partial_{x_i}^{\alpha_{ij }}\partial_{x_j}^{\beta_{ij}}\omega_{2}^{\infty}(x_{i} - x_j)\right)^{m_{\alpha_{ij}\beta_{ij}}} \bigg) \Lambda(x_1 +t_1 e^0, \ldots, x_n + t_n e^0) \di^4x_{1} \cdots \di^4x_n,\\
    \int_{\mathbb{M}^n} \bigg(&\prod_{1\leq i<j \leq n} \prod_{\alpha_{ij},\beta_{ij}} \left(\partial_{x_i}^{\alpha_{ij }}\partial_{x_j}^{\beta_{ij}}W(x_{i}, x_j)\right)^{m_{\alpha_{ij}\beta_{ij}}} \bigg) \Lambda(x_1 +t_1 e^0, \ldots, x_n + t_n e^0) \di^4x_{1} \cdots \di^4x_n,\\
    \int_{\mathbb{M}^n} \bigg(&\prod_{1\leq i<j \leq n} \prod_{\alpha_{ij},\beta_{ij}} \big(\partial_{x_i}^{\alpha_{ij }}\partial_{x_j}^{\beta_{ij}} \omega_2^{\infty}(x_i - x_j)\big)^{p_{\alpha_{ij}\beta_{ij}}} \left( \partial_{x_i}^{\alpha_{ij }}\partial_{x_j}^{\beta_{ij}} W(x_i, x_j)\right)^{m_{\alpha_{ij}\beta_{ij}}-p_{\alpha_{ij}\beta_{ij}}} \bigg)\bigg) \nonumber\\ 
    &\times \Lambda(x_1 +t_1 e^0, \ldots, x_n + t_n e^0) \di^4x_{1} \cdots \di^4x_n.
\end{align}
We start focusing on the first. By a change of variable, the following:
\begin{align*}
    \int_{\mathbb{M}^n} \bigg(&\prod_{1\leq i<j \leq n} \prod_{\alpha_{ij},\beta_{ij}} \big( \partial_{x_i}^{\alpha_{ij }}\partial_{x_j}^{\beta_{ij}} \omega_2^{\infty}(x_i - x_j - (t_i - t_j)e^0)\big)^{m_{\alpha_{ij}\beta_{ij}}}\bigg) \Lambda(x_1, \ldots, x_n) \di^4x_{1} \cdots \di^4x_n
\end{align*}
is secularily bounded by Corollary \ref{cor: NoSecInvariant}. In particular, the hypothesis of the corollary are satisfied by Proposition \ref{app: 1} for Dirac fields and Proposition A$.1$ in the appendix of \cite{FaldinoEquilibriumpAQFT} for scalar fields.\\\\
For the second term, let us start defining:
\begin{equation*}
    \mathcal{W}(x_1, \ldots, x_n, t_n) \coloneqq \int_{\mathbb{R}^{n-1}} \di t_1 \cdots \di t_{n-1} \bigg( \prod_{1\leq i<j \leq n} \prod_{\alpha_{ij},\beta_{ij}} \left(\partial_{x_i}^{\alpha_{ij }}\partial_{x_j}^{\beta_{ij}}W(x_{i} - t_i e^0, x_j - t_j e^0)\right)^{m_{\alpha_{ij}\beta_{ij}}}\bigg)
\end{equation*}
Using assumption \eqref{eq: assunzione1} we can estimate lines simply by a constant. The idea is to do it a number of times sufficient to ensure that, keeping the graph connected, at least one vertex has a single line attached to it:
\begin{equation*}
    |\mathcal{W}(x_1, \ldots, x_n, t_n)| \leq C \int_{\mathbb{R}^{n-1}} \di t_1 \cdots \di t_{n-1} \bigg( \prod_{1\leq i<j \leq n} \prod_{\alpha_{ij},\beta_{ij}} \left|\partial_{x_i}^{\alpha_{ij }}\partial_{x_j}^{\beta_{ij}} W(x_i - t_ie^0, x_j - t_je^0)\right|^{q_{\alpha_{ij} \beta_{ij}}}\bigg)
\end{equation*}
where $q_{\alpha_{ij},\beta_{ij}}$ are such that at least one vertex has just a single line attached. Since the integrand function is a positive measurable function, we can apply Tonelli's theorem and compute the iterative integrals. In particular, we start computing the integral in the variable that appears only once, using assumption \eqref{eq: assunzione2}. Iterating this procedure, if needed, we conclude that the function $\mathcal{W}(x_1,...,x_n,t_n)$ is well defined and that its absolute value is uniformly bounded in the variable $t_n$ by a positive function $\mathfrak{W}$, bounded in the spatial variables:
\begin{equation*}
    |\mathcal{W}(x_1, \ldots, x_n, t_n)| \leq \mathfrak{W}(\mathbf{x}_1, \ldots, \mathbf{x}_n).
\end{equation*}
Therefore, as $\Lambda$ is a smooth function of compact support, by using again Tonelli's theorem:
\begin{align*}
    \bigg|\int_{\mathbb{R}^{n-1}} \di t_1 \cdots \di t_{n-1} \bigg(\int_{\mathbb{M}^n} &\bigg(\prod_{1\leq i<j \leq n} \partial^{\alpha_i}_{x_i} \partial^{\alpha_j}_{x_j} W(x_i - t_ie^0, x_j - t_je^0)^{l_{ij}}\bigg) \Lambda(x_1, \ldots, x_n) \di^4x_{1} \cdots \di^4x_n\bigg) \bigg| \\ 
    &\leq \int_{\mathbb{M}^n} \mathfrak{W}(\mathbf{x}_1, \ldots, \mathbf{x}_n) |\Lambda(x_1, \ldots, x_n)| \di^4x_{1} \cdots \di^4x_n \leq C
\end{align*}
by the assumption of smoothness and compact support.\\\\
Finally, combining the ideas used for the first and second cases one argues analogously for the third contribution. In particular, when both the vacuum two-point function and the smooth part are connected to a same vertex we can get rid of the latter in the integrand using again assumption \eqref{eq: assunzione1}, provided that the graph is still connected.
Finally, combining the ideas used for the first and second cases one argues analogously for the third contribution. In particular, the result obtained in Proposition A.2 of \cite{FaldinoEquilibriumpAQFT} can be easily generalised using condition \eqref{eq: assunzione1}, showing that:
\begin{equation*}
    \left|\left\langle\omega_2^{\infty},\partial^\alpha W\cdot\frac{\delta^2}{\delta\phi_1\delta\phi_2}\tau_{t_1}(A)\otimes\tau_{t_2}(B)\right\rangle\right|\leq\frac{C'}{(|t_1-t_2|+r)^{3/2}}
\end{equation*}
for every $A,B\in\mathcal{F}_{\mu c}$, $t_1,t_2\in\mathbb{R}$ and for some positive constants $C',r$. Using this estimate and the same techniques already applied in the previous part of the proof, the result is obtained.
\end{proof}

\section{Applications and examples}\label{se:applications}
In this section, the previous results are applied in concrete examples. Namely, we study complex scalar and Dirac fields coupled with an external classical electromagnetic potential and a general one loop-order computation for a self interacting scalar field. In the first two examples we illustrate the failure of the return to equilibrium property for interactions supported on the entire space and, at the first order in perturbation theory, how this problem is avoided considering the adiabatic limit only after the large time limit. This is done by constructing directly a KMS state for the interacting theory as discussed in Section \ref{sec: Int States}. Finally, in the last example, we show the application of Theorem \ref{th:truncated} making explicit the mechanism behind the cancellation of secular effects in a loop diagram. For computational purposes, we set $\hbar = c = 1$.
\subsection{Scalar electrodynamics}
\label{se:scalar-ED}
 \subsubsection{General introduction}\label{4.2}
We want to study a complex massive scalar field $\varphi$ coupled to an external electromagnetic potential $A_\mu$. The latter is not quantised but treated as a classical external field. We define the \emph{covariant derivative} $\mathcal{D_\mu}$ in the usual way:
\begin{equation*}
    \mathcal{D}_\mu(\varphi)\coloneqq\partial_\mu\varphi+ieA_\mu\varphi.
\end{equation*}
$e$ denotes the coupling constant with the external field. The Lagrangian density of the interacting scalar field $\varphi$ is obtained by substituting the ordinary derivatives with the covariant ones (\emph{minimal coupling}):
\begin{align*} 
      \mathcal{L}=&-(\mathcal{D}_{\mu}\varphi)^{\dagger}\mathcal{D}^{\mu}\varphi-m^2\varphi^{\dagger}\varphi\\
      &=-\partial_{\mu}\varphi^{\dagger}\partial^{\mu}\varphi-m^2\varphi^{\dagger}\varphi+ieA^{\mu}(\varphi^{\dagger}\partial_{\mu}\varphi-\partial_{\mu}\varphi^{\dagger}\varphi)-e^2A_{\mu}A^{\mu}\varphi^{\dagger}\varphi.
\end{align*}
The interaction Lagrangian is therefore given by the two terms:
\begin{equation*}
    \mathcal{L}_{I}=ieA^{\mu}(\varphi^{\dagger}\partial_{\mu}\varphi-\partial_{\mu}\varphi^{\dagger}\varphi)-e^2A_{\mu}A^{\mu}\varphi^{\dagger}\varphi.
\end{equation*}
Finally, the quadri-potential $A_\mu$, that we will consider in the following, does not depend on the time variable:
\begin{equation}\label{formaexpdelquadrip}
    A_\mu(s)=(A_0(\textbf{s}),\textbf{A}(\textbf{s})).
\end{equation}
With $A_0$ we denoted the scalar potential while $\textbf{A}$ is the vector potential. We recall also that the electromagnetic quadri-potential has a gauge freedom that can be fixed in the most convenient way. As we will see, in our case, this corresponds to the choice of the \emph{Coulomb gauge}.\\

In this examples we compare the large time limit of the correction at the first order in perturbation theory to the $2$-point function:
\begin{equation}\label{duepuntinoncorr}
    \omega^{\beta}(R_V(\varphi^\dagger(x)) R_V(\varphi(y))),
\end{equation}
with the first order expression of the two-point function:
\begin{equation}\label{duepunticorr}
    \omega^{\beta,V}(R_V(\varphi^\dagger(x)) R_V(\varphi(y))).
\end{equation}
As a first result we are able to show how secular effects, and thus the failure of the return to equilibrium property, arises, if the adiabatic limit is taken before the large time limit, in the presence of polynomially growing stationary external potentials. The main cause of this effect is the lack of restrictions on the spatial support of interactions. This, in turn, leads to the momentum conservation at each vertex of the perturbative expansion. Therefore, the combination of the frequencies modes of the propagators appearing in the first order diagrammatic expansion that have concordant sign will lead to time translation dependent contributions. After a partial integration in time, necessary to take into account the presence of the arbitrary choice of the switch-on function $\chi$, the now manifestly time translation dependent contributions are analyzed using stationary phase methods. The general mechanism is the following. A polynomial potential of order $n$ results in a time growing prefactor of order $t^n$. This prefactor is multiplied by a function of time, obtained using the stationary phase method, that decays asymptotically at most as $t^{-(3+n)/2}$. For sufficient large value of $n$ it is now clear the presence of time growths.\\
As proven in the previous sections we know that outside the adiabatic limit return to equilibrium holds, namely that the two-point function cannot grow in time. Therefore, we give the explicit expression of the thermal equilibrium interacting two-point function showing the cancellation of the time translation dependent terms.

\subsubsection{Scalar electric potential}\label{Sec:ScalarEl}
In this subsection we will assume an external static scalar potential:
\begin{equation}
    A(x)=(A_0(\textbf{x}),\textbf{0}).
\end{equation}
In order to make the computation clearer, we split it in several steps each with its own relevance. The first step is to consider the first order correction to \eqref{duepuntinoncorr} in its general form.
\begin{prop}\label{prop_2p_noneq}
    Let's assume a spatially compactly supported interaction of the following form:
    \begin{equation*}
         V=\int \chi(t_x)h(\textbf{x})\left[ ieA_0(\varphi^{\dagger}\partial_{0}\varphi-\partial_{0}\varphi^{\dagger}\varphi) - e^2A_{0}A^{0}\varphi^{\dagger}\varphi\right]\;d^4x.
\end{equation*}
Then, at first order in perturbation theory, it holds:
\begin{equation}\label{first_order_not_correct}
    \omega^{\beta,(1)}\left(R_V\left(\varphi^{\dagger}(x)\right) R_V\left(\varphi(y)\right)\right)=\omega_2^\beta(x,y)+Z^{\mathfrak{A}}(x,y)+Z^{\mathfrak{B}}(x,y),
\end{equation}
where the superscript $(1)$ indicates that the previous equality holds up to first order in the coupling constant, $\omega_2^\beta(x,y)$ is the zeroth-order contribution corresponding to the two-point function of the free thermal state and $Z^{\mathfrak{A}}(x,y) + Z^{\mathfrak{B}}(x,y)$, the first order contribution, are time translation dependent terms whose expression is reported in Appendix \ref{appendix_scalar_el} and whose corresponding Feynman diagrams are depicted in Figure \ref{fig:Z}.
\end{prop}
The proof of this Proposition is extensively reported in Appendix \ref{appendix_scalar_el}. The steps consist in merely applying the expansion of the Bogoliubov map up to the first perturbative order taking into account the distinction between the non-commutative and time-ordered products. The time translation dependence arises from the combination of the energy modes in propagators and in the two-point function that are both of positive or both of negative energy when the $h \to 1$ limit is considered.

  \begin{figure}[H]
      \centering
    \begin{tikzpicture}
        \begin{scope}[thick, decoration={
    markings,
    }
    ]
    \node at (1,-0.3) {$-ie\dot\chi h A$};
    \node at (0.3,0.3) {$\Delta_R$};
    \node at (1.8,0.3) {$\omega_2^\beta$};
    \draw[postaction={decorate}] (-0.5,0)--(1,0);
    \filldraw[black] (-0.5,0) circle (0.5pt) node[anchor=north] {$x$};
    \filldraw[black] (1,0) circle (0.5pt);
    \draw[postaction={decorate} ] (1,0) -- (2.5,0);
    \filldraw[black] (2.5,0) circle (0.5pt) node[anchor=north] {$y$};
    \end{scope}
    \end{tikzpicture}
        \begin{tikzpicture}
        \begin{scope}[thick, decoration={
    markings,
    }
    ]
    \node at (1,-0.3) {$ie\dot\chi h A$};
    \node at (0.3,0.3) {$\omega_2^\beta$};
    \node at (1.8,0.3) {$\Delta_A$};
    \draw[postaction={decorate}] (-0.5,0)--(1,0);
    \filldraw[black] (-0.5,0) circle (0.5pt) node[anchor=north] {$x$};
    \filldraw[black] (1,0) circle (0.5pt);
    \draw[postaction={decorate} ] (1,0) -- (2.5,0);
    \filldraw[black] (2.5,0) circle (0.5pt) node[anchor=north] {$y$};
    \end{scope}
    \end{tikzpicture}
    \hspace{1cm}
        \begin{tikzpicture}
        \begin{scope}[thick, decoration={
    markings,
    }
    ]
    \node at (1,-0.3) {$-ie2\chi h A$};
    \node at (0.3,0.3) {$\partial_0\Delta_R$};
    \node at (1.8,0.3) {$\omega^\beta_2$};
    \draw[postaction={decorate}] (-0.5,0)--(1,0);
    \filldraw[black] (-0.5,0) circle (0.5pt) node[anchor=north] {$x$};
    \filldraw[black] (1,0) circle (0.5pt);
    \draw[postaction={decorate} ] (1,0) -- (2.5,0);
    \filldraw[black] (2.5,0) circle (0.5pt) node[anchor=north] {$y$};
    \end{scope}
    \end{tikzpicture}
        \begin{tikzpicture}
        \begin{scope}[thick, decoration={
    markings,
    }
    ]
    \node at (1,-0.3) {$ie2\chi h A$};
    \node at (0.3,0.3) {$\omega_2^{\beta}$};
    \node at (1.8,0.3) {$\partial_0\Delta_A$};
    \draw[postaction={decorate}] (-0.5,0)--(1,0);
    \filldraw[black] (-0.5,0) circle (0.5pt) node[anchor=north] {$x$};
    \filldraw[black] (1,0) circle (0.5pt);
    \draw[postaction={decorate} ] (1,0) -- (2.5,0);
    \filldraw[black] (2.5,0) circle (0.5pt) node[anchor=north] {$y$};
    \end{scope}
    \end{tikzpicture}
    \caption{The sum of the two diagrams on the left corresponds to $Z^{\mathfrak{A}}$ while the two on the right to $Z^{\mathfrak{B}}$. See Appendix \ref{appendix_scalar_el} for their analytic expression.}
    \label{fig:Z} 
\end{figure}

Starting from this general first order computation, we are now interested in showing through the concrete example of a constant external electromagnetic field how already for this simple model secular growths might arise. In fact, we show that taking the adiabatic limit before the large time limit leads to the failure of the return to the equilibrium property and thus the arising of secular growths.
\begin{prop}\label{prop_expl_pot} Assume the external scalar potential to be of the form $A_0(\textbf{s})=s_i$, for $i=1,2,3$ denoting any spatial coordinate, and consider the adiabatic limit $h\to 1$ for the interaction. If $f,g\in\mathcal{C}^\infty_0(\mathbb{M}^4)$ are two test functions and denoting by $f_t(x)\coloneqq f(t_x-t,\textbf{x})$ the time translation, the absolute value of the smeared distribution \eqref{first_order_not_correct}:
\begin{equation}\label{first_order_not_correct_sm}
 \mathcal{W}_{f,g}(t)\coloneqq \left|\left\langle\omega^{\beta,(1)}\left(R_V(\varphi^{\dagger}(x))  R_V\left(\varphi(y)\right)\right),f_t(x)g_t(y)\right\rangle\right|  
\end{equation}
grows linearly in time:
    \begin{equation*}
        \lim_{t\to\+\infty}\frac{\mathcal{W}_{f,g}(t)}{t}= C_{f,g},\;\;\;C_{f,g}\in\mathbb{R}^+,
    \end{equation*}
with $C_{f,g}\neq 0$ for some $f,g\in\mathcal{C}^\infty_0(\mathbb{M}^4)$.
\end{prop}
The full proof is reported in Appendix \ref{appendix_scalar_el} but we sketch here the main ideas. Consider the $h \to 1$ limit of the distribution $Z^{\mathfrak{A}}$ in \eqref{first_order_not_correct}, corresponding to the sum of the two diagrams depicted at the left hand side of Figure \ref{fig:Z}, and smear it with the functions $f_t,g_t$:
\begin{align*}
    &\left\langle\left(Z^{\mathfrak{A}}\right)(x,y),f_t(x)g_t(y)\right\rangle=\\ 
    &\mathcal{Z}_{\chi,f,g}(t)+et\int\frac{\di^3\textbf{p}}{4\en{p}^3(2\pi)^3}p_i\left[2\bshs{-}{p}\hat{f}(\en{p},-\textbf{p})\hat{g}(-\en{p},\textbf{p})+2\bshs{+}{p}\hat{f}(-\en{p},-\textbf{p})\hat{g}(\en{p},\textbf{p})\right]\\
    &+et\int\frac{\di^3\textbf{p}}{4\en{p}^3(2\pi)^3}p_i(\bshs{-}{p}+\bshs{+}{p})\left[\hat{\dot{\chi}}(-2\en{p})\hat{f}(\en{p},-\textbf{p})\hat{g}(\en{p},\textbf{p})e^{2i\en{p}t}+\hat{\dot{\chi}}(2\en{p})\hat{f}(-\en{p},-\textbf{p})\hat{g}(-\en{p},\textbf{p})e^{-2i\en{p}t}\right].
\end{align*}
Here, the time prefactor arises from the Fourier transform of the external potential, as a momentum derivative, acting on the energy modes and the remaining part, $\mathcal{Z}_{\chi,f,g}(t)$, is a bounded function of $t$. The Boltzmann thermal factors were denoted by:
    \begin{equation*}
            \bshs{+}{p}\coloneqq\frac{1}{1-e^{-\beta\en{p}}} \quad \bshs{-}{p}\coloneqq\frac{1}{e^{\beta\en{p}} - 1}.
    \end{equation*}
Using a slightly modified version of Lemma A.1 in \cite{Meda_2022} (namely by stationary phase methods) it is possible to show that the absolute value of the third addend in the previous expression goes to $0$ in the limit $t\to\infty$. On the contrary, the absolute value of the second term appearing in the previous expression grows linearly in time, giving rise to a secular effect. Finally, remaining smeared distribution $Z^{\mathfrak{B}}$ doesn't cancel the linearly growing term just described above.\\

The previous proposition shows that if the interaction between a complex scalar field and a constant external electric field is considered, in the adiabatic limit, the return to equilibrium property is not satisfied. In particular, the expectation value of the product of two local smeared fields in a KMS state of the free theory grows linearly in time at the first order in perturbation theory. Nevertheless, thanks to the results obtained in the previous section, in particular in Theorem \ref{th:truncated} and Corollary \ref{cor: NoSecInvariant},
we know that for spatially compactly supported interactions the return to equilibrium property is satisfied, so this kind of growths disappear. 
Therefore, the two-point function \eqref{duepuntinoncorr} of the thermal state of the free theory converges, in the large time limit, to the two-point function \eqref{duepunticorr} of the thermal state for the interacting theory. In addition, the KMS state for the interacting theory is well defined also in the limit $h \to 1$. So, we proceed to compute up to the first perturbative order the interacting KMS two-point function when the interaction is spatially compact. Therefore, in order to make explicit the cancellation of such growths, we compute the first order correction to \eqref{duepunticorr}, and we then compare it with the first order correction to \eqref{duepuntinoncorr} obtained in Proposition \ref{prop_2p_noneq}.
\begin{prop}\label{prop_correct_el}
     Consider an interaction of the following form:
    \begin{equation*}
         V=\int \chi(t_x)h(\textbf{x})\left[ ieA_0(\varphi^{\dagger}\partial_{0}\varphi-\partial_{0}\varphi^{\dagger}\varphi) - e^2A_{0}A^{0}\varphi^{\dagger}\varphi\right]\;\di^4x,
\end{equation*}
Then, at first order in perturbation theory, it holds:
\begin{equation}\label{first_order_correct}
    \omega^{\beta,V,(1)}\left(R_V\left(\varphi^{\dagger}(x)\right) R_V\left(\varphi(y)\right)\right)=\omega_2^\beta(x,y)+
Z^{\mathfrak{B},Inv},
\end{equation}
where again $\omega_2^\beta$ is the free KMS state two-point function giving the zeroth order contribution while the first order correction is contained in the remaining part:
 \begin{align*}        
Z^{\mathfrak{B},Inv}
=&2e\int\di^3\textbf{s}\, h(\textbf{s})A_0(\textbf{s})\int\int\frac{\di^3\textbf{p}\di^3\textbf{k}}{(2\pi)^6}\frac{1}{4\en{p}}\left[e^{i\textbf{p}(\textbf{s}-\textbf{y})+i\textbf{k}(\textbf{x}-\textbf{s})}+e^{i\textbf{p}(\textbf{s}-\textbf{x})+i\textbf{k}(\textbf{y}-\textbf{s})}\right]\\
        & \times\left[\frac{2\en{p}}{(\en{p}+\en{k})(\en{p}-\en{k})}\right]\left[\bshs{+}{p}e^{-i\en{p}(t_x-t_y)}-\bshs{-}{p}e^{i\en{p}(t_x-t_y)}\right]
\end{align*}
that is now invariant under time translations.
\end{prop}
Again the proof is explicitly reported in Appendix \ref{appendix_scalar_el}.
The linear contribution $Z^{\mathfrak{B},Inv}$ correspond to the part of $Z^\mathfrak{B}$ which is invariant under time translations. 
More precisely, the difference with respect to the results of Proposition \ref{prop_2p_noneq} is that not only the Bogoliubov map but also the interacting KMS state is perturbatively expanded.
In this way, for a generic local observable $F\in\mathcal{A}$:
\begin{equation}\label{pert_exp}
       \omega^{\beta,V}(R_V(F))=\sum_{n\geq0}(-1)^n\int_{\beta S_n}\di u_1...\di u_n\;\omega^{\beta,\mathcal{T}}\left(R_V(F)\otimes \tau_{iu_1}K\otimes...\otimes\tau_{iu_n}K\right),
\end{equation}
where $\omega^{\beta,\mathcal{T}}$ denotes the truncated $n-$point function of the state $\omega^\beta$ and the specific form of the generator of the cocycle that intertwines the free and the interacting dynamics, $K$, is:
\begin{equation*}
     K=R_V(-\Dot{V})=-\int\Dot{\chi}(t_x)h(\textbf{x})\left[ieA_0(\textbf{x})R_V(\varphi^{\dagger}\partial_{0}\varphi-\partial_{0}\varphi^{\dagger}\varphi) + e^2A_0(\mathbf{x})A^0(\mathbf{x})R_V(\varphi^{\dagger} \varphi)\right]\di^4x.
\end{equation*}
At the first order in the coupling constant $e$, the perturbative expansion \eqref{pert_exp} coincides with \eqref{first_order_not_correct}, except for the presence of the additional contribution $E$ corresponding to $n=1$ in \eqref{pert_exp} and $R_V(F)=F$:
\begin{equation*}
    E\coloneqq\int_0^{\beta}du\;\omega^{\beta,\mathcal{T}}(\varphi^{\dagger}(x)\varphi(y)\otimes\alpha_{iu}\Dot{V}).
\end{equation*}
This correction cancels the terms coming from the perturbative expansion of the Bogoliubov map responsible of the time growths in Proposition \ref{prop_expl_pot}. This was expected, as the interacting KMS two-point function is time translation invariant at each fixed perturbative order.\\

Finally, recall that $\omega^{\beta,V}$ doesn't depend on the switching function $\chi$ and has a convergent $h \to 1$ limit. Therefore, the just proven results emphasize the lack of secular growths also in the case in which the interaction is everywhere supported and has always been turned on (or, equivalently, has been turned on at past infinity).

\subsubsection{Vector magnetic potential}\label{Sec:ScalarMag}
In this subsection we consider the case in which the complex scalar field interacts with an external magnetic vector potential:
\begin{equation}
A_\mu(s)=(0,\textbf{A}(\textbf{s})).
\end{equation}
The aim is to show, for completeness, that the interacting thermal equilibrium two-point function does not present any time growth for the generic choice of the stationary vector potential. In particular, we prove a result analogous to the one presented in Proposition \ref{prop_correct_el}. However, for the sake of simplicity, we report (and prove) the result directly in the adiabatic limit and fixing the gauge freedom for the external potential.
\begin{prop}\label{prop_correct_mag}
Consider an interaction of the following form:
\begin{equation*}
    V=\int \chi(t_x)h(\textbf{x})\left[ ieA^{i}(\varphi^{\dagger}\partial_{i}\varphi-\partial_{i}\varphi^{\dagger}\varphi)-e^2A_{i}A^{i}\varphi^{\dagger}\varphi\right]\;\di^4x,
\end{equation*}
with the external vector potential in the Coulomb gauge ($\partial_iA^i(\textbf{s})=0$). Then, at first order in perturbation theory and in the limit $h \to 1$, it holds:
\begin{equation}\label{first_order_correct_mag}
    \omega^{\beta,V,(1)}\left(R_V\left(\varphi^{\dagger}(x)\right) R_V\left(\varphi(y)\right)\right)=\omega_2^\beta(x,y)+\mathfrak{C},
\end{equation}
where again $\omega_2^\beta$ is the free KMS state two-point function giving the zeroth order contribution while the first order correction:
\begin{align*}       
   \mathfrak{C}\coloneqq\int \di^3\textbf{s}\int \di^3\textbf{p}\int \di^3\textbf{k}\;\frac{eA_i(\mathbf{s})}{(2\pi)^6}\frac{k_i}{2\en{p}\en{k}}
    \left[\bshs{+}{p}\frac{e^{-i\en{p}(t_x-t_y)}}{(\en{k}+\en{p})}+\bshs{+}{p}\frac{e^{-i\en{p}(t_x-t_y)}}{(\en{k}-\en{p})}+\right.\\
   \left.+\bshs{-}{p}\frac{e^{i\en{p}(t_x-t_y)}}{(\en{k}-\en{p})}+\bshs{-}{p}\frac{e^{i\en{p}(t_x-t_y)}}{(\en{k}+\en{p})}\right]\left[e^{i\mathbf{p}(\mathbf{s-y})+i\mathbf{k}(\mathbf{x-s})}-e^{i\mathbf{p}(\mathbf{x-s})+i\mathbf{k}(\mathbf{y-s})}\right]
    \end{align*}
is invariant under time translations.
\end{prop}
The proof of this Proposition is analogous to the one in Proposition \ref{prop_correct_el}, with the only difference consisting in the form of the interaction. The full derivation can be found in the Master thesis \cite{Sangaletti} and it is presented in Appendix \ref{appendix_scalar_mag} for completion.

\subsection{Dirac fields coupled with strong external electromagnetic field}\label{sec: Dirac example}
In this section, we focus on the perturbative computation of the expectation value of an observable for a Dirac field coupled with a classical (strong) external electromagnetic field. In particular, we aim at computing at first order in perturbation theory the expectation value of the Dirac current:
\begin{equation*}
    j^{\mu}(x) = e \overline{\psi} \gamma^{\mu} \psi(x)
\end{equation*}
on the KMS state associated to the interacting theory determined by the following Lagrangian:
\begin{equation*}
    \mathcal{L} = \mathcal{L}_{D} + \mathcal{L}_I,
\end{equation*}
with:
\begin{equation*}
    \mathcal{L}_D = \overline{\psi}(i \cancel{\partial} - m) \psi, \quad \mathcal{L}_I = eA_{\mu} \overline{\psi} \gamma^{\mu} \psi.
\end{equation*}
Here, $A_{\mu} = A_{\mu}(\mathbf{x})$ is the electromagnetic potential, function of just the spatial coordinates. For the discussion of locally covariant Dirac field theories in external potentials see \cite{zahn2012}. In a second step, we compare this computation with the corresponding one on the free KMS state. The motivation and type of analysis is the same as that of the previous subsection with the only difference that here, instead of the two-point function, the expectation value of a concrete local observable is studied.\\
We start assuming the interacting potential to have spatial compact support:
\begin{equation*}
    V = \lambda e \int \di^4x A_{\mu}(\mathbf{x}) \overline{\psi}(x) \gamma^{\mu} \psi(x) \chi(t) h(\mathbf{x}),
\end{equation*}
with $\chi(t)$ the switch-on function given in equation \eqref{eq:accad} and $h(\mathbf{x})$ the space cutoff. The observable, $j^{\mu}(x)$, is supported after the complete switch on. In the following, we first perturbatively expand the expectation value:
\begin{equation*}
    \cancel{\omega}^{\beta,V}(R_V(j^0(x))),
\end{equation*}
showing the disappearing of any dependence from the time $t$ and from the switch-on process. After, we rather assume $V$ to be already in the adiabatic limit, $h \to 1$, where return to equilibrium is not guaranteed. In that context, we discuss the presence of secular growths and the explicit dependence on $\chi$.\\\\
Consider the perturbative expansion of the state (see equation \eqref{eq: espansKMS}) and that of the observable, at first order. Then:
\begin{align}
    &\cancel{\omega}^{\beta,V}\big(R_V(j^0(x))\big) =\nonumber\\
    \cancel{\omega}^{\beta}(:j^0(x):) -i \cancel{\omega}^{\beta}(V :j^0(x):) +&i \cancel{\omega}^{\beta}\big(T(V, :j^0(x):)\big) - \int_0^{\beta}\di u \, \cancel{\omega}^{\beta,{\mathcal{T}}}\big(:j^0(x): \otimes \, \tau_{iu} (K)\big) + O(\lambda^2) \label{eq: DiracCorrente},
\end{align}
where we have denoted the normal ordering with $:\, \cdot \, :$ to remind that without it the expectation value of $j$ is ill-defined. The first term corresponds with the zeroth order, the expectation value of the current of a free Dirac field on the state that is KMS with respect to the free time evolution. The second and third terms correspond to the expansion of the Bogoliubov map of the observable at first order in the coupling $\lambda$. Finally, the fourth term is the first-order expansion of the interacting equilibrium state, where $K$ is the generator of the cocycle intertwining the free and the interacting dynamics \eqref{eq:formaK}.\\
The first thing to notice is that, by normal ordering with respect to the two-point function of $\cancel{\omega}^{\beta}$, the zeroth order contribution vanishes. The computation of the remaining terms is summarised in the following Lemmata, proven in the appendix:
\begin{lem}\label{lem: DiracTerm1}
The term in \eqref{eq: DiracCorrente} corresponding to the first perturbative order correction to the state takes the explicit form:
\begin{align}
    - &\int_0^{\beta}\di u \, \cancel{\omega}^{\beta,{\mathcal{T}}}\big(:j^0(x): \otimes \, \tau_{iu} (K)\big) =\nonumber\\
    &-\frac{\lambda e^2}{(2\pi)^6} \int \di t_1 \di^3\mathbf{x}_1  \Dot{\chi}(t_1) h(\mathbf{x}_1) A_{\mu}(\mathbf{x}_1) \gamma^{0} \gamma^{\mu} \int \frac{\di^3\mathbf{p} \di^3 \mathbf{k}}{4 \en{p} \en{k}}e^{-i(\mathbf{p} - \mathbf{k})(\mathbf{x} - \mathbf{x}_1)} \times \nonumber\\
    &\times \bigg[ \frac{e^{i(\en{p} - \en{k})(t - t_1)}}{\en{p} - \en{k}} (-\gamma^0 \en{p} - \gamma^i p_i + m) (-\gamma^0 \en{k} - \gamma^i k_i + m) \bigg( \frac{1}{1 + e^{\beta \en{k}}} - \frac{1}{1 + e^{\beta \en{p}}}  \bigg)\nonumber\\
    &\,\,\,\, + \frac{e^{i(\en{p} + \en{k})(t - t_1)}}{\en{p} + \en{k}} (-\gamma^0 \en{p} - \gamma^i p_i + m) (\gamma^0 \en{k} - \gamma^i k_i + m)\bigg( \frac{1}{1 + e^{\beta \en{k}}} + \frac{1}{1 + e^{\beta \en{p}}} - 1 \bigg)\nonumber\\
    &\,\,\,\, + \frac{e^{-i(\en{p} + \en{k})(t - t_1)}}{\en{p} + \en{k}} (\gamma^0 \en{p} - \gamma^i p_i + m) (-\gamma^0 \en{k} - \gamma^i k_i + m) \bigg( \frac{1}{1 + e^{\beta \en{k}}} + \frac{1}{1 + e^{\beta \en{p}}} - 1\bigg)\nonumber\\
    &\,\,\,\, + \frac{e^{-i(\en{p} - \en{k})(t - t_1)}}{\en{p} - \en{k}} (\gamma^0 \en{p} - \gamma^i p_i + m) (\gamma^0 \en{k} - \gamma^i k_i + m) \bigg( \frac{1}{1 + e^{\beta \en{k}}} - \frac{1}{1 + e^{\beta \en{p}}} \bigg)\bigg], \label{eq: termine1}
\end{align}
where, for notational convenience, spinorial indices are dropped. 
\end{lem}
\begin{lem}\label{lem: DiracTerm2}
The terms in \eqref{eq: DiracCorrente} corresponding to the first perturbative order expansion of the Bogoliubov map take the explicit form:
\begin{align}
    &-i \cancel{\omega}^{\beta}(V :j^0(x):) +i \cancel{\omega}^{\beta}\big(T(V, :j^0(x):)\big) = \nonumber\\
    &-i\lambda e^2\int \di^4x_1 \chi(t_1) h(\mathbf{x}_1)\gamma^{0} \gamma^{\mu} A_{\mu}(\mathbf{x}_1) \bigg[ \cancel{S}_F(x_1 - x)\cancel{S}_F(x - x_1) + \cancel{S}^{+}_2(x_1 - x)\cancel{S}^{-}_2(x - x_1)+ \nonumber\\
    &\,\,\,\, i\theta(t - t_1)\big(-\cancel{S}(x_1 - x)W^+_{\beta}(x - x_1) +  W^+_{\beta}(x_1 - x)\cancel{S}(x - x_1)\big) \bigg], \label{eq: parti1}
\end{align}
where $\cancel{S}^{\pm}_2(z)$ is the two-point function of the ground state (see equation \eqref{eq: 2puntiGro_+} and \eqref{eq: 2puntiGro_-}), $\cancel{S}_F(z)$ the Feynman propagator of the Dirac equation on Minkowski spacetime, $\cancel{S}(z)$ the corresponding Pauli-Jordan function and $W^+_{\beta}(z)$ the smooth function defined by:
\begin{align*}
    W_{\beta}^{+}(z) &= \cancel{S}^{\beta,+}_2(z) - \cancel{S}^+_{2}(z) \nonumber\\
    &= - \frac{1}{(2\pi)^3} \int \frac{\di^3\mathbf{p}}{2 \en{p}} \frac{e^{i \mathbf{p} \mathbf{z}}}{1 + e^{\beta \en{p}}}\big((-\gamma^0 \en{p} - \gamma^i p_i +m)e^{-i\en{p} t_z} +  (\gamma^0 \en{p} - \gamma^i p_i +m)e^{i\en{p} t_z} \big).
\end{align*}
Moreover, for notational convenience, all spinorial indices were dropped.
\end{lem}
We now combine these results to get the first perturbative order of the expectation value of the conserved current. Again we summarise it in a proposition whose proof will be given in the appendix. As in the previous section, the time translation dependent contributions arising from the expansion of the Bogoliubov map present in Lemma \ref{lem: DiracTerm2}, corresponding to the combination of only positive or negative energy modes in the perturbative expansions, are cancelled by the contributions in Lemma \ref{lem: DiracTerm1}.
\begin{prop}\label{prop: DiracTotalTerm}
The first order perturbative expansion of the expectation value of the conserved Dirac current, computed adding equation \eqref{eq: termine1} together with \eqref{eq: parti1}, is:
\begin{align*}
    &\cancel{\omega}^{\beta,V}\big(R_V(j^0(x))\big) =\\
    &\frac{4\lambda e^2}{(2 \pi)^6}\int \di^3 \mathbf{x}_1 h(\mathbf{x}_1) A^{0}(\mathbf{x}_1) \int \frac{\di^3 \mathbf{p} \di^3 \mathbf{k}}{\en{p}^2 - \en{k}^2} e^{-i(\mathbf{p} - \mathbf{k})(\mathbf{x} - \mathbf{x}_1)} \bigg[  \frac{\en{p}^2 + m^2 + k_i p^i}{\en{p} (1 + e^{\beta \en{p}})} - \frac{\en{k}^2 + m^2 + k_i p^i}{\en{k} (1 + e^{\beta \en{k}})}\bigg] + O(\lambda^2)
\end{align*}
where $\omega_{\textbf{k},M} = \sqrt{|\mathbf{k}|^2 + M^2}$. This expression is correct up to renormalization freedom terms that are time translation invariant.
\end{prop}
It is manifest that the expectation value is independent of the switch-on process and invariant under any time translation. Namely, secular effects are absent as expected by Theorem \ref{thm: NoSecEq} and Theorem \ref{th:truncated}.

\subsubsection{Secular growths for non-compactly supported interactions}
If the limit $h \to 1$ is taken before thermalisation, the system might not satisfy the property of return to equilibrium, as mentioned in Section \ref{sec: Ret Eq}. Therefore, the expectation value at first order in perturbation theory of the zero component of the current $j^0(x)$ corresponds to:
\begin{equation*}
    \langle j^0\rangle_{{\beta},1}\coloneqq -i \cancel{\omega}^{\beta}(V :j^0(x):) +i \cancel{\omega}^{\beta}\big(T(V, :j^0(x):)\big),
\end{equation*}
that is just the expansion of the Bogoliubov map. We now show that this term manifests secular growths following the same procedure discussed at the end of the subsection \ref{4.2}. In order to do so, we assume that the electromagnetic potential has the following form\footnote{We are aware that such a potential might have few physical interests for arbitrary $n > 1$. Our choice aims just at showing the existence of such growths already at the first perturbative order. More physically interesting situations, like the case $n=1$ of a constant electric field, also present secular growths just at higher order in perturbation theory.}:
\begin{equation*}
    A_{\mu}(\mathbf{s}) = (-E (s_1)^n, 0, 0, 0),
\end{equation*}
where $s_1$ denotes the first component of $\mathbf{s} \in \mathbb{R}^3$, $n \in \mathbb{N}$ and $E \in \mathbb{R}$ is a constant.\\
Let's start analysing the term \eqref{eq: termine1} since, as we discussed above, is the one that has an explicit dependence on the time translation. Notice that:
\begin{equation*}
    (s_1)^n e^{-i(\mathbf{p} - \mathbf{k}) \mathbf{s}} = \bigg(\frac{1}{i} \partial_{p_1} \bigg)^n e^{-i(\mathbf{p} - \mathbf{k}) \mathbf{s}},
\end{equation*}
to get:
\begin{align*}
\eqref{eq: termine1}
=
    \frac{(-1)^ne^2E}{(2\pi)^3(i)^{n-1}} \int_{-\infty}^t &\di t_s \chi(t_s) \int \di^3 \mathbf{p}\, \partial_{p_1}^n \bigg[ \frac{1}{\en{p} \en{k}} e^{i(\mathbf{p} - \mathbf{k}) \mathbf{x}} 
    \\
 &
    \bigg[ e^{i(\en{p} - \en{k})(t-t_s)}(\en{p} \en{k} - p^i k_i - m^2) \bigg( \frac{1}{1 + e^{\beta \en{p}}} - \frac{1}{1 + e^{\beta \en{k}}} \bigg)\\ 
    &+ e^{i(\en{p} + \en{k})(t-t_s)}(-\en{p} \en{k} - p^i k_i - m^2) \bigg( -\frac{1}{1 + e^{\beta \en{p}}} - \frac{1}{1 + e^{\beta \en{k}}} \bigg)\\
    &+ e^{-i(\en{p} + \en{k})(t-t_s)}(-\en{p} \en{k} - p^i k_i - m^2) \bigg( \frac{1}{1 + e^{\beta \en{p}}} + \frac{1}{1 + e^{\beta \en{k}}} \bigg)\\
    &+e^{-i(\en{p} - \en{k})(t-t_s)}(\en{p} \en{k} - p^i k_i - m^2) \bigg( -\frac{1}{1 + e^{\beta \en{p}}} + \frac{1}{1 + e^{\beta \en{k}}} \bigg)\bigg] \bigg]_{p=k}.
\end{align*}
In this expression the most growing contribution in time, denoted by $J_{t,1}$, is:
\begin{align*}
    J_{t,1} = &-\frac{i(-1)^n2e^2E}{(2\pi)^3} \int_{-\infty}^t \di t_s \chi(t_s) (t - t_s)^n \int \di^3 \mathbf{p}  \bigg(\frac{p_1}{\en{p}}\bigg)^n \bigg( \frac{e^{2i\en{p} (t - t_s)} + (-1)^n e^{-2i \en{p} (t - t_s)}}{1 + e^{\beta \en{p}}} \bigg)\\
    &\sim \frac{(-1)^neE}{(2\pi)^3} \int \di^3 \mathbf{p}  \bigg(\frac{p_1}{\en{p}}\bigg)^n \frac{1}{1 + e^{\beta \en{p}}} \int_{-\infty}^{+\infty} \di t_s \Dot{\chi}(t_s) \frac{(t-t_s)^n}{\en{p}} \big( e^{2i\en{p} (t - t_s)} + (-1)^n e^{-2i \en{p} (t - t_s)}\big)\\
    &= \frac{(-1)^neE}{(2\pi)^3} \int \di^3 \mathbf{p}  \frac{(p_1)^n}{\en{p}^{n+1}} \frac{1}{1 + e^{\beta \en{p}}} \bigg(\frac{1}{2i}\bigg)^n \partial_{\en{p}}^n \bigg( \Hat{\Dot{\chi}}(-2 \en{p}) e^{2i\en{p} t} + \Hat{\Dot{\chi}}(2 \en{p})e^{-2i \en{p} t}\bigg),
\end{align*}
where, at the second step, we have integrated by parts in $t_s$ and again kept only the most divergent contribution in $t$ arising from it. Performing a change of variable, similar to the one outlined in the proof of Proposition \ref{app: 1}, and noticing that by symmetry arguments only the case $n = 2l$ for $l\in \mathbb{N}$ does not vanish we get:
\begin{align*}
  J_{t,1} = \frac{(-1)^le^2E}{(2\pi)^3 2^{2l}} \frac{t^{2l}}{t^{l + 3/2}}\int_0^{+\infty} &\di w   \frac{w^{l + 1/2}}{(\frac{w}{t} + m )^{2l - 1}} \frac{( \frac{w}{t} + 2m )^{l + \frac{1}{2}}}{1 + e^{\beta (\frac{w}{t} + m)}} 
  \\
  &\times\partial_{w}^{2l} 
  \bigg[ \Hat{\Dot{\chi}}\bigg(-2 \bigg(\frac{w}{t} + m\bigg)\bigg) e^{2i(\frac{w}{t} + m) t} + \Hat{\Dot{\chi}}\bigg(2 \bigg(\frac{w}{t} + m\bigg)\bigg)e^{-2i (\frac{w}{t} + m) t}\bigg].
\end{align*}
By the arguments in the proof of Proposition \ref{app: 1} the above integral in $w$ is proven to be convergent uniformly in $t$. Acting on a general family of time translated test functions $f_{t'}$ does not tame the shown divergence following the same arguments of Proposition \ref{prop_expl_pot}. Therefore, for $l \geq 2$ ($n = 4$), secular growths arise at first order in perturbation theory.\\
The contributions associated with the renormalised Feynman propagators are computed in the same way. We will not do it here as, for the purpose of showing the presence of secular growths, we just need to notice that the term involving the renormalization of the Feynman propagators is independent of the temperature, so it cannot cancel the just computed growths.

\subsection{Fields with more than quadratic interaction}
\label{sec: Loop Example}
We consider an uncharged scalar Klein-Gordon field $\varphi$ of mass $m$ with an interaction Lagrangian
\[
V= \lambda \int \di t \di \mathbf{x}  \chi(t)h(\mathbf{x}) \varphi^n(t,\mathbf{x}),
\]
which propagates on a Minkowski spacetime $\mathbb{M}$. We are interested in analyzing the two-point function
\[
\omega^I_2(x_1,x_2) = \omega(R_V(\varphi(x_1)) R_V(\varphi(x_2)))
\]
for fields in an Hadamard state $\omega$ which for simplicity is assumed to be quasifree, translation and rotation invariant and its two-point function satisfies the decay estimate for large time given in Corollary \ref{cor: NoSecInvariant}.
The two-point function of such a state is
\begin{equation}\label{eq:2ptexample}
\omega_2(x_1,x_2) = \frac{1}{(2\pi)^3} \int \dd^{3}\mathbf{p} \;  e^{i \mathbf{p}(\mathbf{x}_1-\mathbf{x}_2)} 
\left(
e^{i \omega_{\mathbf{p}}(t_{{x}_1}-t_{{x}_2})}\frac{1+F(\omega_{\mathbf{p}})}{2 \omega_{\mathbf{p}}}
+
e^{-i \omega_{\mathbf{p}}(t_{{x}_1}-t_{{x}_2})}\frac{F(\omega_{\mathbf{p}})}{2 \omega_{\mathbf{p}}}\right)
\end{equation}
where $F$ is a smooth function on $[m,\infty)$ which is bounded together with its derivatives near $m$ and which decays rapidly for large value of its argument. For equilibrium state at inverse temperature $\beta$ with respect to the time evolution olong the Minkowski time $F(w) = (e^{\beta w}-1)^{-1}$.
\\
The analysis is split in two parts. In the first, we take the limit $h\to 1$ and consider, at second perturbative order, the limit $t_{x_1}+t_{x_2} \to + \infty$ of $\omega^I_2(x_1,x_2)$. 
The result is a linear growth in $t_{x_1}+t_{x_2}$ when the support of the Fourier transform of the distribution expressing the internal loop intersect the positive and/or negative mass-shell, e.g. for a $\lambda \varphi^4$ interacting theory in a generic state but not for $\lambda \varphi^3$. 
In the second part, we verify the statement of Theorem \ref{th:truncated} studying the same above configuration in the limit $t_{x_1}+t_{x_2} \to + \infty$ but for fixed  $h \in \mathcal{C}^{\infty}_0(\Sigma)$. The result is a bounded function of  $t_{x_1}+t_{x_2}$.

We preliminary observe that the second order contribution to $\omega^{I}_2(x_1,x_2)$ can be written as, 
\begin{align}
\omega_2^{I,(2)}(x_1,x_2) 
=& \lambda^2\int \di y_1 \di y_2  \Delta_R(x_1,y_1) \chi(t_{y_1}) h({y}_1) S(y_1-y_2) \chi(t_{y_2}) h({y}_2)  \Delta_A(y_2,x_2) \nonumber
\\
&+ \lambda^2 \int \di y_1 \di y_2  \Delta_R(x_1,y_1) \chi(t_{y_1}) h({y}_1) S_A(y_1-y_2) \chi(t_{y_2}) h({y}_2)  \omega_2(y_2,x_2) \label{eq:omega2I}
\\
&+ \lambda^2 \int \di y_1 \di y_2  \omega_2(x_1,y_1) \chi(t_{y_1}) h({y}_1) S_A(y_1-y_2) \chi(t_{y_2}) h({y}_2)  \Delta_A(y_2,x_2) \nonumber
\end{align}
where the integrals are understood in the distributional sense. Furthermore $S$ and $S_A$ are the integral kernels of two distributions obtained as the sum of the various propagators which joins the internal vertices. The precise form of $S$ and $S_A$ depends on the order of the interaction Lagrangian, on the quantum state in which the background theory is assumed to be and furthermore, $S_A$ depends also on some renormalization constants. A careful analysis shows that at late time the dominant contribution to 
$\omega_2^{I,(2)}(x_1,x_2)$ in \eqref{eq:omega2I} is the first one. The other two terms are in fact bounded functions when both points $x_1$ and $x_2$ are translated in the future time direction.
We thus discuss with some care the diagram with external lines consisting of advanced and retarded propagators:
 \begin{equation}
     W(x_1, x_2) = \lambda^2 \int \di y_1 \di y_2  \Delta_R(x_1,y_1) \chi(t_{y_1}) h({y}_1) S(y_1-y_2) \chi(t_{y_2}) h({y}_2)  \Delta_A(y_2,x_2)
 \end{equation}
whose general diagrammatic representation is depicted here:
\begin{figure}[H]
\centering
\begin{tikzpicture}
\begin{scope}[thick, decoration={
    markings,
    }
    ]
    \node at (1.25,1) {S};
    \node at (0,0.3) {$\Delta_R$};
   \node at (2.5,0.3) {$\Delta_A$};
   \node at (0.3,-0.3) {$\chi h$};
   \node at (2.2,-0.3) {$\chi h$};
    \draw[postaction={decorate}] (-0.5,0)--(0.5,0);
  \filldraw[black] (-0.5,0) circle (0.5pt) node[anchor=north] {$x$};
  \filldraw[black] (0.5,0) circle (0.5pt) ;
 \draw[pattern=horizontal lines] (1.25,0) circle (0.75);
  \filldraw[black] (2,0) circle (0.5pt) ;
  \draw[postaction={decorate}] (2,0)--(3,0);
  \filldraw[black] (3,0) circle (0.5pt) node[anchor=north] {$y$};
  \end{scope}
  \end{tikzpicture}
  \caption{Generic diagram}
  \end{figure}
The first remark concerns the explicit form of $S(y_1-y_2)$. In a general $\varphi^n$-self interacting real scalar field theory it has the explicit form: $S(y_1-y_2) = (\omega_2(y_1 - y_2))^{n-1}$, where $\omega_2$ denotes the two-point function of the Hadamard state for the free theory that we are perturbatively expanding around given in \eqref{eq:2ptexample}.
Writing it in Fourier space we have:
\begin{equation*}
    (\omega_2(x))^{n-1} = \frac{1}{(2 \pi)^{3(n-1)}}\int \di^4 p \int \di^4 q_1 \cdots \int \di^4 q_{n-1} \hat{\omega}_2(q_1) \cdots \hat{\omega}_2(q_{n-1}) \delta^{(4)}(q_1 + \ldots + q_{n-1} - p) e^{ipx}.
\end{equation*}
This shows that, if $\hat{\omega}_2(q_i)$ 
is supported both on the positive and negative mass shells, namely when $F$ in \eqref{eq:2ptexample} is not vanishing,
the support of $S$ contains four vectors $p$ obtained as  all possible combination of the $n-1$ four-vectors $q_i$ supported on either the positive or negative mass-shells. In particular, in the $n=4$ case this means that the support of $S(y_1-y_2)$ 
in 
momentum space is above the positive mass-shell of mass $m$ and below the negative mass shell of mass $m$. 
Instead, for $n=3$ the loop function $S(y_1-y_2)$ is supported at vanishing four momentum, above the positive mass-shell of mass $2m$ and below the negative mass-shell of mass $2m$.
We move on discussing the contributions of the considered diagram. Let us recall here the explicit form of the retarded propagator:
\[
\Delta_{R}(x,y) = \frac{i}{(2\pi)^3} \theta(t_x-t_y) \int \frac{\di^3 \mathbf{p}}{2\omega_{\mathbf{p}}}
\left(
e^{i \omega_{\mathbf{p}}(t_x-t_y)} - e^{- i \omega_{\mathbf{p}}(t_x-t_y)} \right) e^{-i \mathbf{p}(\mathbf{x}-\mathbf{y})}
\]
and $\Delta_{A}(x,y)=\Delta_{R}(y,x)$. Therefore, in $W(x_1,x_2)$ there are four contributions: one with positive frequency of $\Delta_R$ combined with positive frequency of $\Delta_A$, the second obtained combining negative frequencies of $\Delta_R$ and $\Delta_A$ and other two combining frequencies of different signs.\\
We discuss with some care the contribution obtained combining positive frequency of $\Delta_R$ with negative frequency of $\Delta_A$. The other contribution with frequencies of opposite signs can be treated similarly. Finally, when frequencies have the same signs, the behavior at large time can be studied using the same methods outlined in the previous examples in this Section and therefore here we do not report it in full details.\\
We instead consider:
\begin{align*}
    A_h(x_1,x_2) = \lambda^2 &\int_{-\infty}^{t_{x_1}} \di t_{y_1} \chi(t_{y_1}) \int_{-\infty}^{t_{x_2}} \di t_{y_2}\chi(t_{y_2})  \int \frac{\di^3 \mathbf{q}_1 \di^3 \mathbf{q}_2 \di^4 p}{(2\pi)^{4}4 \omega_{\mathbf{q}_1}\omega_{\mathbf{q}_2}}e^{i\mathbf{q}_1 \mathbf{x}_1 + i \mathbf{q}_2 \mathbf{x}_2} e^{-ip^0(t_{y_1} - t_{y_2})} \\
    &\times \hat{h}(\mathbf{q}_1 - \mathbf{p})\hat{h}(\mathbf{q}_2 + \mathbf{p}) \hat{S}(p^0, \mathbf{p}) e^{i\omega_{\mathbf{q}_1}(t_{x_1} - t_{y_1}) - i\omega_{\mathbf{q}_2}(t_{x_2} - t_{y_2})}
\end{align*}
For $t_{x_1}$ and $t_{x_2}$ which are large and positive, we shall for simplicity consider the limit where the $\epsilon$ in the definition of $\chi$ tends to $0$. In this way $\chi(t)\to\theta(t)$:
\begin{align}
    A_h(x_1,x_2) = \lambda^2 &\int_{0}^{t_{x_1}} \di t_{y_1} \int_{0}^{t_{x_2}} \di t_{y_2}  \int \frac{\di^3 \mathbf{q}_1 \di^3 \mathbf{q}_2 \di^4 p}{(2\pi)^{4}4 \omega_{\mathbf{q}_1}\omega_{\mathbf{q}_2}}e^{i\mathbf{q}_1 \mathbf{x}_1 + i \mathbf{q}_2 \mathbf{x}_2} e^{-ip^0(t_{y_1} - t_{y_2})} \nonumber\\
    &\times \hat{h}(\mathbf{q}_1 - \mathbf{p})\hat{h}(\mathbf{q}_2 + \mathbf{p}) \hat{S}(p^0, \mathbf{p}) e^{i\omega_{\mathbf{q}_1}(t_{x_1} - t_{y_1}) - i\omega_{\mathbf{q}_2}(t_{x_2} - t_{y_2})} \label{eq: Diagramma Loop}.
\end{align}
We now specify to the two cases, the first in which the limit $h \to 1$ is taken and the second in which $h$ is kept fixed.

\subsubsection{Secular growths for loop diagrams in the adiabatic limit}
\label{sec: Loop Example sec}
In the limit $h \to 1$ the integrals over $\mathbf{q}_1$ and $\mathbf{q}_2$ can be taken. Establishing, essentially, energy and momentum conservation in the interaction. The result becomes:
\begin{align*}
    A_1(x_1,x_2) = \lambda^2 &\int_{0}^{t_{x_1}} \di t_{y_1} \int_{0}^{t_{x_2}} \di t_{y_2}  \int \frac{\di^4 p}{(2\pi)^{4}4 \omega_{\mathbf{p}}^2}e^{i\mathbf{p} (\mathbf{x}_1 - \mathbf{x}_2)} e^{-ip^0(t_{y_1} - t_{y_2})} \hat{S}(p^0, \mathbf{p}) e^{i\omega_{\mathbf{p}}(t_{x_1} - t_{y_1}) - i\omega_{\mathbf{p}}(t_{x_2} - t_{y_2})}.
\end{align*}
In order to compute this expression explicitly, we move to another system of coordinates $s = t_{y_1} + t_{y_2}$ and $r = t_{y_1} - t_{y_2}$. 
Performing the integration in $s$ and evaluating the result for $t_{x_1} = t_{x_2} = t$ we obtain
\begin{align}
{A}_1(t,\mathbf{x}_1;t,\mathbf{x}_2) 
&=  
\frac{\lambda^2}{(2\pi)^4} \int \dd^4 p 
 \frac{e^{i \mathbf{p}(\mathbf{x}_1-\mathbf{x}_2)}}
{4\omega_{\mathbf{p}}^2}
\left(
t\int\limits_{-t}^{t} \dd r 
e^{i r (p^0+\omega_{\mathbf{p}})}
\hat{S}(p^0,\mathbf{p})
-
\int\limits_{-t}^{t} \dd r 
e^{i r (p^0+\omega_{\mathbf{p}})}
|r|
\hat{S}(p^0,\mathbf{p})
\right).
\label{eq: Contributi secular}
\end{align}
Here,  we are interested in the leading order contribution for large $t$, hence, we 
consider first of all the term that grows linearly with $t$. We notice that in the limit $t \to \infty$:
\begin{equation*}
    \lim_{t \to \infty} \int_{-t}^{t} \di r   e^{i(p^0 + \omega_{\mathbf{p}})r} = \delta(p^0 + \omega_{\mathbf{p}}).
\end{equation*}
Therefore, depending on whether $(-\omega_{\mathbf{p}}, \mathbf{p})$ is in the support of $\hat{S}(p^0, \mathbf{p})$ there is a linear growth in $t$ corresponding to a secular effect. To this extent we need to make a couple of remarks. It is possible to check, with an analogous computation, that the same contribution arising from 
the negative frequency of $\Delta_R$ and positive frequency of $\Delta_A$ does not cancel this growth. This is because $\hat{S}(p^0, \mathbf{p})$ does not have any symmetry for $p_0 \to - p_0$, being the Fourier transform of a product of Hadamard states. In the case in which $(-\omega_{\mathbf{p}}, \mathbf{p}) \not\in \mathrm{supp}(\hat{S}(p^0, \mathbf{p}))$, 
also a more carefully analysis shows that 
$A_1$ vanishes. \\
We are left with the evaluation of the  second contribution in \eqref{eq: Contributi secular}. In particular, computing the limit in the sense of distribution theory:
\begin{equation*}
    \lim_{t \to \infty} \int_{-t}^{t} \di r   e^{i(p^0 + \omega_{\mathbf{p}})r} |r| = \partial_{p^0} \mathrm{PV}\left( \frac{2}{p^0 + \omega_{\mathbf{p}}}\right).
\end{equation*}
This is now explicitly independent from $t$. In particular, inserting this expression in the second addend of \eqref{eq: Contributi secular} leads to a well defined and finite expression as the expression we started with in \eqref{eq: Diagramma Loop} is finite. Therefore, in the $t \to \infty$ limit, the eventual secular growth
in \eqref{eq: Contributi secular} is not cancelled.\\
Therefore, in the case of a quasifree state $\omega$ with two-point function both with positive and negative frequency modes, for $n=3$ secular growths do not arise in view of the support properties of the corresponding $\hat{S}(p_0,\mathbf{p})$ while for $n=4$ they might arise due to the above discussed support properties of $\hat{S}$. 
The observation summarized in this subsection is in   accordance with the results of 
Akhmedov, and Astrakhantsev and Popov, see e.g. Section 2.1 in \cite{Akhmedov_2014}.

\subsubsection{Absence of secular growths for loop diagrams}\label{sec: Loop Example no sec}
We know from Theorem \ref{th:truncated} and from 
Corollary \ref{cor: NoSecInvariant}
that if the state of the background theory is invariant under time translation and if the two-point function decays sufficiently fast for large timelike separation, no secular growths are present outside the adiabatic limit.
This is for example the case when the state in which the theory is analyzed is the KMS state of the free theory and if the support of $h$ is kept compact in space.

In order to see this explicitly we  start again from Equation \eqref{eq: Diagramma Loop}. We study this expression in the limit $t_{x_1} + t_{x_2} \to \infty$ aiming to show that secular growths are absent.\\
We first notice that $\hat{S}(p^0,\mathbf{p})$ can be expressed in the form:
\begin{equation}\label{eq: formaS}
    \hat{S}(p^0,\mathbf{p}) = ((p^0)^2 + |\mathbf{p}|^2 + 1)^k\hat{\mathscr{W}}(p^0,\mathbf{p})
\end{equation}
where, for sufficiently large $k \in \mathbb{N}$, 
$\hat{\mathscr{W}}(p^0,\mathbf{p})$ is an absolutely integrable function and hence it is the integral kernel of a distribution of order $0$.
The most singular contribution in $\hat{S}$ is the one which does not depend on $F$ in \eqref{eq:2ptexample}.
That contribution, which is essentially the weight of the ordinary K\"allen-Lehmann decomposition of the interacting Feynman propagator at perturbative order two, for $n=3$ it is 
\[
\hat{S}(p^0,\mathbf{p})
= \sqrt{1-\frac{4m^2}{(p^0)^2- |\mathbf{p}|^2}}
\theta(\sqrt{(p^0)^2- |\mathbf{p}|^2}-2m)
\]
while for $n=4$ it is a function which grows at most as $(p^0)^2- |\mathbf{p}|^2$ for large momentum.

Let us now perform the explicit computation starting from \eqref{eq: Diagramma Loop}. First, as a distribution in $\mathbf{x}_1$ and $\mathbf{x}_2$, we smear it with test functions $f_1,f_2 \in \mathcal{D}(\mathbb{R}^3)$ and insert the above expression for $\hat{S}(p^0, \mathbf{p})$:
\begin{align}
    A=&\int \di^3 \mathbf{x}_1 \di^3 \mathbf{x}_2 A_{h}(x_1,x_2) f_1(\mathbf{x}_1) f_2(\mathbf{x}_2) \nonumber\\
    =& \lambda^2 \int_{0}^{t_{x_1}} \di t_{y_1} \int_{0}^{t_{x_2}} \di t_{y_2}  \int \frac{\di^3 \mathbf{q}_1 \di^3 \mathbf{q}_2 \di^4 p}{(2\pi)^{4}4 \omega_{\mathbf{q}_1}\omega_{\mathbf{q}_2}}\hat{f}_1(-\mathbf{q}_1)\hat{f}_2(-\mathbf{q}_2) e^{-ip^0(t_{y_1} - t_{y_2})} \label{eq: senzah}\\
    &\times \hat{h}(\mathbf{q}_1 - \mathbf{p})\hat{h}(\mathbf{q}_2 + \mathbf{p}) ((p^0)^2 - |\mathbf{p}|^2 + a_1)^k\hat{\mathscr{W}}(p^0,\mathbf{p}) e^{i\omega_{\mathbf{q}_1}(t_{x_1} - t_{y_1}) - i\omega_{\mathbf{q}_2}(t_{x_2} - t_{y_2})}. \nonumber
\end{align}
For simplicity, having in mind the $\lambda \varphi^4$ example, let us choose $k=1$. Nevertheless, the case for general $k$, follows straightforwardly. Moreover, we 
analyze here only the contribution $(p^0)^2 \hat{\mathscr{W}}(p^0,\mathbf{p})$ as the other are can be studied with the same method.\\
Notice, from Equation \eqref{eq: senzah}, that the contributions associated to the external legs, the integrals in $\mathbf{q}_1$ and $\mathbf{q}_2$, can be factorized and studied separately. In doing so, let us further include a factor of $p^0$ in each of the two contributions and focus on the $\mathbf{q}_1$ integral:
\begin{align*}
    g_{t_{x_1}}(p):=&\int_{0}^{t_{x_1}} \di t_{y_1} p^0 e^{-ip^0 t_{y_1}} \int \frac{\di^3 \mathbf{q}_1}{2\omega_{\mathbf{q}_1}}\hat{f}_1(-\mathbf{q}_1)  \hat{h}(\mathbf{q}_1 - \mathbf{p}) e^{i\omega_{\mathbf{q}_1}(t_{x_1} - t_{y_1})}\\
    =& i e^{-ip^0 t_{x_1}} \int \frac{\di^3 \mathbf{q}_1}{2\omega_{\mathbf{q}_1}}\hat{f}_1(-\mathbf{q}_1)  \hat{h}(\mathbf{q}_1 - \mathbf{p}) - i \int \frac{\di^3 \mathbf{q}_1}{2\omega_{\mathbf{q}_1}}\hat{f}_1(-\mathbf{q}_1)  \hat{h}(\mathbf{q}_1 - \mathbf{p}) e^{i\omega_{\mathbf{q}_1}t_{x_1}}\\
    &- \int_{0}^{t_{x_1}} \di t_{y_1} e^{-ip^0 t_{y_1}} \int \frac{\di^3 \mathbf{q}_1}{2}\hat{f}_1(-\mathbf{q}_1)  \hat{h}(\mathbf{q}_1 - \mathbf{p}) e^{i\omega_{\mathbf{q}_1}(t_{x_1} - t_{y_1})}.
\end{align*}
The first integral is a convolution between two Schwartz functions and so it gives a Schwartz function of $\mathbf{p}$ which is also bounded in $p_0$. The remaining terms, for large $t_{x_1}-t_{y_1}$, are estimated using stationary phase methods:
\begin{equation*}
    \int \di^3 \mathbf{q}_1 \omega_{\mathbf{p}}^{n-1}\hat{f}_1(-\mathbf{q}_1)  \hat{h}(\mathbf{q}_1 - \mathbf{p}) e^{i\omega_{\mathbf{q}_1}(t_{x_1} - t_{y_1})} \leq C (t_{x_1} - t_{y_1})^{-\frac{3 + 2n}{2}}\hat{h}(-\mathbf{p})
\end{equation*}
for $C \in \mathbb{R}^+$, while it is bounded for 
small $t_{x_1}-t_{y_1}$.
This estimate are sufficient to prove that 
$g_{t_{x_1}}(p)$ tends to $0$ in the uniform norm in $p$.
In the same way, the contribution of the $\mathbf{q}_2$-leg, denoted by $\tilde{g}_{t_{x_2}}$ is estimated analogously.
The corresponding contribution to $A$ in \eqref{eq: senzah} is 
\[
\langle \hat{\mathscr{W}}, g_{t_{x_1}}g_{t_{x_2}}\rangle,
\]
which stays bounded 
in the limit of large $t_{x_1}$ and large $t_{x_2}$ in view of the uniform estimates given above and the character of the distribution with the integral kernel $\hat{\mathscr{W}}$.
\\

${}$ \\ \\ ${}$ \\
{\bf  Acknowledgments}
S.G.~and N.P.~would like to thank the National Group of Mathematical Physics (GNFM-INdAM) for the support. The research performed by S.G.~and N.P.~was supported in part by the MIUR Excellence Department Project 2023-2027 awarded to the Dipartimento di Matematica of the University of Genova, CUP\textunderscore$\,$D33C23001110001. L.S.~would like to thank Daniela Cadamuro for the fruitful discussions and the Deutsche Forschungsgemeinschaft (DFG, German Research Foundation) $---$ project no.~396692871 within the Emmy Noether grant CA1850/1-1 for the financial support.\\
The authors are grateful to Tommaso Bruno, Edoardo D'Angelo, Simone Murro and Gabriel Schmid for the useful discussions. They also benefited from conversations with Markus Fr\"ob and Claudio Iuliano.

\appendix 
\section{Perturbative expansions}\label{app: espansioni}
As it is customary in perturbation theory, we provide here the perturbative expansions in $K$ of the objects used in this work. Namely, for any $A \in \mathcal{A}$:
\begin{align}
    U_V(t) &= 1 + \sum_{n=1}^{\infty} i^n \int_0^t \di t_1 \int_0^{t_1} \di t_2 \cdots \int_0^{t_{n-1}} \di t_n \tau_{t_n}(K) \cdots \tau_{t_1}(K) \label{eq: espansCociclo},\\
    \tau_t^V(A) &= \tau_t(A) + \sum_{n = 1}^{\infty} i^n \int_{t S_n} \di t_1 \cdots \di t_n [\tau_{t_1}(K), [\ldots [\tau_{t_n}(K), \tau_t(A)] \ldots]]\label{eq: espansDynam},\\
    \omega^{\beta, V}(A) &= \sum_{n = 0}^{\infty} (-1)^{n} \int_{\beta S_n} \di u_1 \cdots \di u_n \, \omega^{\beta,{\mathcal{T}}}\bigg( A \otimes \bigotimes_{k=1}^{n} \tau_{i u_k}(K) \bigg)\label{eq: espansKMS},
\end{align}
where as usual $S_n \coloneqq \{ (t_1, \ldots, t_n) \in \mathbb{R}^n : 0 \leq t_1 \leq \cdots \leq t_n \leq 1 \}$ is the unit simplex, $S_n' \coloneqq \{ (t_1, \ldots, t_n) \in \mathbb{R}^n : -1 \leq t_1 \leq \cdots \leq t_n \leq 1 \}$ and $\omega^{\beta,{\mathcal{T}}}$ is the truncated (connected) part of $\omega^{\beta}$.

\section{Technical Propositions and Lemmata}
\subsection{Time decay of Dirac vacuum and KMS two-point functions}
\begin{prop}\label{app: 1}
Consider the two-point functions for free fermionic massive KMS states:
\begin{align*}
    \cancel{S}^{\beta, +}_{2}(x-y) &= \frac{1}{(2 \pi)^3} \int \frac{\di^3\mathbf{p}}{2 \en{p}} \bigg( \frac{(-\gamma^0 \en{p} - \gamma^i p_i + m)e^{-i\en{p} (t_x - t_y)}}{(1 + e^{-\beta \en{p}})} - \frac{(\gamma^0 \en{p} - \gamma^i p_i + m)e^{i\en{p} (t_x - t_y)}}{(1 + e^{\beta \en{p}})}  \bigg) e^{i \mathbf{p} (\mathbf{x}- \mathbf{y})},\\
    \cancel{S}^{\beta, -}_{2}(x-y) &= \frac{1}{(2 \pi)^3} \int \frac{\di^3\mathbf{p}}{2 \en{p}} \bigg(\frac{(-\gamma^0 \en{p} - \gamma^i p_i + m)e^{-i\en{p} (t_x - t_y)}}{(1 + e^{\beta \en{p}})} - \frac{(\gamma^0 \en{p} - \gamma^i p_i + m)e^{i\en{p} (t_x - t_y)}}{(1 + e^{-\beta \en{p}})}  \bigg) e^{i \mathbf{p} (\mathbf{x}- \mathbf{y})}.
\end{align*}
If $y-x$ is a causal future pointing vector then:
\begin{equation*}
    |\partial^{(\alpha)}_x \partial^{(\beta)}_y \cancel{S}^{\beta, \pm}_{2}(x; t_y + t, \mathbf{y})| \leq \frac{C_{\alpha}}{t^{3/2}} \, , \quad t > 1,
\end{equation*}
with $\alpha, \beta$ a multiindices denoting the partial derivations in $x$ and $y$.
\end{prop}
\begin{proof}
We know that both $\cancel{S}^{\beta, \pm}_{2}$ and $\cancel{S}^{\pm}_{2}$ are Hadamard, therefore their difference $\cancel{S}^{\beta, \pm}_{2} -\cancel{S}^{\pm}_{2} = W^{\pm}_{\beta}$ is a smooth function. The vector joining the points $y$ and $(t_y + t, \mathbf{y})$ is timelike future pointing, thus, by composition, the vector joining $x$ and $(t_y + t, \mathbf{y})$ is also timelike future pointing. Follows that the point $(x; t_y + t, \mathbf{y})$ is neither in the singular support of $\cancel{S}^{\beta, \pm}_{2}$ nor in that of $\cancel{S}^{\pm}_{2}$ and thus both the two-point functions are smooth in a neighborhood of the point $(x; t_y + t, \mathbf{y})$. We start noticing that:
\begin{align*}
    \cancel{S}^{\pm}_{2}(x; t_y+t, \mathbf{y}) &= \pm (i \cancel{\partial}_x + m) \omega_2^{\infty}(\pm(t_x - t_y - t, \mathbf{x} - \mathbf{y}))\\
    &= \pm (i \cancel{\partial}_x + m) \bigg( 4 \pi \frac{m}{i \sqrt{(t + t_y - t_x)^2 - |\mathbf{x} - \mathbf{y}|^2}} K_1(im \sqrt{(t + t_y - t_x)^2 - |\mathbf{x} - \mathbf{y}|^2})\bigg),
\end{align*}
where $\omega_{2}^{\infty}$ is the ground two-point function of the free real Klein-Gordon field
\begin{equation*}
    \omega_{2}^{\infty}(x-y) = \frac{1}{(2\pi)^3} \int \frac{\di^3 \mathbf{p}}{2 \omega_{\mathbf{p}}} e^{-i (\omega_{\mathbf{p}}(t_x - t_y) - \mathbf{p} (\mathbf{x} - \mathbf{y}))}
\end{equation*}
and $K_1$ the modified Bessel function of the second kind and index $1$. However, derivatives applied to the factor in front of the Bessel function just improve the decaying properties for large $t$. Furthermore, by properties of Bessel functions we also know that \cite{gradshteyn2007}:
\begin{equation*}
    \frac{d}{dx} K_n(x) = \frac{n}{x} K_{n}(x) - K_{n+1}(x)\, , \quad |K_n(y)| \leq \frac{c_n}{\sqrt{|y|}} \,\,\,\,\, y \gg n.
\end{equation*}
Follows that, for large $t$, the decaying properties of $\cancel{S}^{\pm}_{2}(x; t_y+t, \mathbf{y})$ are not worse than those of $\omega^{\infty}_2(x; t_y+t, \mathbf{y})$. Therefore:
\begin{equation*}
    |\cancel{S}^{\pm}_{2}(x; t_y+t, \mathbf{y})| \leq \frac{C}{t^{3/2}}
\end{equation*}
for $C \in \mathbb{R}^+$. Compute now:
\begin{align*}
    t^{3/2}(\cancel{S}^{\beta, +}_{2} - \cancel{S}^{+}_{2})(0; t, \mathbf{x}) &= t^{3/2} \bigg( -\frac{1}{(2\pi)^3} \int\frac{\di^3\mathbf{p}}{2 \en{p}} e^{-i \mathbf{p} \mathbf{x}} \frac{(- \gamma^0 \en{p} - \gamma^i p_i + m)e^{i \en{p} t} + (\gamma^0 \en{p} - \gamma^i p_i + m)e^{-i \en{p} t}}{1 + e^{\beta \en{p}}} \bigg)\\
    &= t^{3/2} \bigg( -\frac{1}{(2\pi)^3} \int\frac{\sin(\theta) \di \theta \, d \varphi \, d|\mathbf{p}| \, |\mathbf{p}|^2}{2 \en{p}} e^{-i \mathbf{p} \mathbf{x}} \frac{(- \gamma^0 \en{p} + m)e^{i \en{p} t} + (\gamma^0 \en{p} + m)e^{-i \en{p} t}}{1 + e^{\beta \en{p}}} \bigg)\\
    &= c t^{3/2} \bigg( -\frac{1}{(2\pi)^3} \int_0^{\infty}\frac{\di|\mathbf{p}| \, |\mathbf{p}|^2}{2 \en{p}} \frac{\sin(|\mathbf{p}| \, |\mathbf{x}|)}{|\mathbf{p}| \, |\mathbf{x}|} \frac{(- \gamma^0 \en{p} + m)e^{i \en{p} t} + (\gamma^0 \en{p} + m)e^{-i \en{p} t}}{1 + e^{\beta \en{p}}} \bigg)\\
    &= c t^{3/2} \bigg( -\frac{1}{(2\pi)^3} \int_{m}^{\infty}\di E \, \sqrt{E^2 - m^2} \frac{\sin(\sqrt{E^2 - m^2} \, |\mathbf{x}|)}{\sqrt{E^2 - m^2} \, |\mathbf{x}|} \frac{(- \gamma^0 E + m)e^{i E t} + (\gamma^0 E + m)e^{-i E t}}{1 + e^{\beta E}} \bigg)\\
    &= c \int_{0}^{\infty}\di w \, \sqrt{w} \bigg(\sqrt{\frac{w}{t} + 2m} \bigg)\frac{\sin\bigg(\sqrt{\frac{w}{t}} \sqrt{\frac{w}{t} + 2m}  \, |\mathbf{x}|\bigg)}{\sqrt{\frac{w}{t}} \sqrt{\frac{w}{t} + 2m}  \, |\mathbf{x}|} \\
    &\,\,\,\,\,\,\,\,\,\,\times \frac{(- \gamma^0 (w/t + m) + m)e^{i (w/t + m) t} + (\gamma^0 (w/t + m) + m)e^{-i (w/t + m) t}}{1 + e^{\beta (w/t + m)}}.\\
\end{align*}
In the second step, we have used the vanishing of the integration on the angular variables of the momenta $p_i$. We have also included all numeric factors outside the integration in the constant $c \in \mathbb{R}$. Let us introduce the functions:
\begin{equation*}
    b_{\pm}\bigg( \frac{w}{t} \bigg) = e^{\pm i m t} \bigg(\sqrt{\frac{w}{t} + 2m}\bigg) \mathrm{sinc}\bigg( \sqrt{\frac{w}{t}} \sqrt{\frac{w}{t} + 2m} |\mathbf{x}|\bigg) \frac{\pm\gamma^0(w/t + m) + m}{1 + e^{\beta(w/t + m)}}
\end{equation*}
bounded and of fast decay due to the Fermi factor. Therefore, the above can be rewritten as:
\begin{align*}
    t^{3/2}(\cancel{S}^{\beta, +}_{2} - \cancel{S}^{+}_{2})(0; t, \mathbf{x}) &= c \int_{0}^{\infty}\di w \, \sqrt{w} \bigg(b_+\bigg( \frac{w}{t} \bigg) e^{i w} +  b_-\bigg( \frac{w}{t} \bigg) e^{-i w} \bigg)\\
    &= \lim_{\epsilon \to 0^+} \int_{0}^{\infty}\di w \, \sqrt{w} \bigg(b_+\bigg( \frac{w}{t} \bigg) e^{i w - \epsilon w} - b_+(0) e^{i w - \epsilon w} +  b_-\bigg( \frac{w}{t} \bigg) e^{-i w - \epsilon w} - b_-(0) e^{-i w - \epsilon w}\bigg)\\
    &+ \lim_{\epsilon \to 0^+} \int_{0}^{\infty}\di w \, \sqrt{w} \bigg(b_+( 0 ) e^{i w - \epsilon w} +  b_-( 0) e^{-i w - \epsilon w} \bigg).
\end{align*}
By boundedness, we know $b_{\pm}(0)$ is finite. Therefore, by explicit computation, the second integral gives a finite contribution in the limit $\epsilon \to 0$ as:
\begin{align*}
    \int_0^{\infty} \di w \sqrt{w} \, b_{\sigma}(0) e^{(i \sigma - \epsilon)w} &= 2 b_{\sigma}(0) \int_0^{\infty} \di s \, s^2 e^{(i \sigma - \epsilon)s^2}\\
    &= \frac{b_{\sigma}(0)}{2}\bigg( \frac{\pi}{ \epsilon - i \sigma} \bigg)^{3/2},
\end{align*}
with $\sigma$ either $+$ or $-$. For the other terms we have:
\begin{align*}
    \int_0^{\infty} \di w \sqrt{w} \bigg( b_{\sigma}\bigg( \frac{w}{t} \bigg) - b_{\sigma}(0) \bigg) e^{(i \sigma - \epsilon)w} &= \frac{1}{(i \sigma - \epsilon)^2} \int_0^{\infty} \di w \sqrt{w} \bigg( b_{\sigma}\bigg( \frac{w}{t} \bigg) - b_{\sigma}(0) \bigg) \partial^2_{w} (e^{(i \sigma - \epsilon)w})\\
    &= \frac{1}{(i \sigma - \epsilon)^2} \int_0^{\infty} \di w (w)^{-3/2} c_{\sigma}\bigg( \frac{w}{t} \bigg) (e^{(i \sigma - \epsilon)w} - 1),
\end{align*}
where we have performed twice an integration by parts and we have called:
\begin{equation*}
    c_{\sigma}\bigg( \frac{w}{t} \bigg) = -\frac{1}{4} \bigg( b_{\sigma}\bigg( \frac{w}{t} \bigg) - b_{\sigma}(0) \bigg) + w \partial_w \bigg( b_{\sigma}\bigg( \frac{w}{t} \bigg) - b_{\sigma}(0) \bigg) + w^2 \partial^2_w \bigg( b_{\sigma}\bigg( \frac{w}{t} \bigg) - b_{\sigma}(0) \bigg).
\end{equation*}
The latter, is a bounded function due to the exponential decay in the argument of $b_{\sigma}( \frac{w}{t} )$. Now, due to the $-3/2$ power of $w$, the above integral is bounded by a constant in time. Follows that we can find a constant $C \in \mathbb{R}^+$ such that:
\begin{equation*}
    |(\cancel{S}^{\beta, +}_{2} - \cancel{S}^{+}_{2})(x; t_y + t, \mathbf{y})| \leq \frac{C}{t^{3/2}}.
\end{equation*}
Combining this estimate with the one for the ground state two-point function we have the first part of the claim for $\cancel{S}^{\beta,+}_2$. For the other two-point function, $\cancel{S}^{\beta,-}_2$, noticing that $W_{\beta}^+(x) = -W_{\beta}^-(x)$ we have by the same steps the same claim.\\
The proof is completed by using again the asymptotic behaviour of modified Bessel functions of the second kind and index 1. In particular, any derivative of the ground two-point function will just increase the decay in $t$. While, for $(\cancel{S}^{\beta, +}_{2} - \cancel{S}^{+}_{2})$, derivatives act on $\frac{\sin(\sqrt{E^2 - m^2} |\mathbf{x}|)}{\sqrt{E^2 - m^2} |\mathbf{x}|} e^{i\sigma E t}$. Again, they cannot spoil the decay in $t$ by the presence of the Fermi factor.
\end{proof}

\subsection{Proofs of Section \ref{Sec:ScalarEl}}\label{appendix_scalar_el}
\begin{proof}[\textbf{Proof of Proposition \ref{prop_2p_noneq}}]
Expanding the Bogoliubov maps in \eqref{first_order_not_correct} and keeping only the first order in the interaction we obtain (with obvious notation):
\begin{equation*}
  Z_1+Z_2\coloneqq \omega^\beta\left([-iV\star\varphi^{\dagger}(x)+iV\cdot_T\varphi^{\dagger}(x)]\star\varphi(y)\right)
  +\omega^\beta\left(\varphi^{\dagger}(x)\star[-iV\star\varphi(y)+iV\cdot_T\varphi(y)]\right).
\end{equation*}
We also recall here the explicit expression of the functional $V$:
\begin{equation*}
         V=\int \chi(t_x)h(\textbf{x})\left[ ieA_0(\textbf{x})(\varphi^{\dagger}\partial_{0}\varphi-\partial_{0}\varphi^{\dagger}\varphi)\right]\;\di^4x,
\end{equation*}
where we have already neglected the terms at second order in the coupling constant and we have substituted our explicit expression for the quadri-potential. To evaluate the first term ($Z_1$), we compute the first functional derivative of the potential. Then, we use the following relation between the two-point function and the Feynman propagator (in which we directly neglect the terms quadratic in the fields, that will go to zero once evaluated on the state):
\begin{equation}\label{eq: FeyRet}
    -iV\star\varphi^\dagger(x)+i V\cdot_T\varphi^\dagger(x)=-\int \di z\frac{\delta V}{\delta \varphi}(z)\Delta_R(x,z)
\end{equation}
to exclude any boundary term. Finally, computing also the functional derivative in the adjoint field $\varphi^{\dagger}$ the final result is:
\begin{align*}
    &Z_1=\overbrace{-ie\int \di s \dot{\chi}(t_s)h(\textbf{s})A_0(\textbf{s})\Delta_R(x,s)\omega_2^\beta(s,y)}^{Z_1^{\mathfrak{A}}}\overbrace{-ie\int \di ss2\chi(t_s)h(\textbf{s})A_0(\textbf{s})\partial_0^s\Delta_R(x,s)\omega_2^\beta(s,y)}^{Z_1^{\mathfrak{B}}},\\
    & Z_2=\underbrace{ie\int \di s \dot{\chi}(t_s)h(\textbf{s})A_0(\textbf{s})\Delta_R(y,s)\omega_2^\beta(x,s)}_{Z_2^{\mathfrak{A}}}+ie\underbrace{\int \di s 2\chi(t_s)h(\textbf{s})A_0(\textbf{s})\partial_0^s\Delta_R(y,s)\omega_2^\beta(x,s)}_{Z_2^{\mathfrak{B}}},
\end{align*}
where we have also reported the result for the term $Z_2$ computed in the same manner. We underline here the presence of terms proportional to $\dot{\chi}$, as a consequence of the presence of a time derivative of the fields in the interaction Lagrangian density.\\ 
We now focus on computing explicitly $Z_1^{\mathfrak{A}}$. Substituting the expression for the retarded propagator and the two-point function \eqref{eq: 2pfscalar}, and using the support properties of the function $\dot{\chi}(t_s)$ to extend the integral in $t_s$ to +$\infty$, we get: 
\begin{align*}
    Z_1^{\mathfrak{A}}=&e\int\di^3\textbf{s}\int_{-\infty}^{+\infty}\di t_s\dot{\chi}(t_s)h(\textbf{s})A_0(\textbf{s})\int\int\frac{\di^3\textbf{p}\di^3\textbf{k}}{(2\pi)^6}\frac{1}{4\en{p}\en{k}}e^{i\textbf{p}(\textbf{s}-\textbf{y})+i\textbf{k}(\textbf{x}-\textbf{s})}\left[\bshs{+}{p}e^{i\en{k}t_x+i\en{p}t_y-it_s(\en{p}+\en{k})}\right.\\\nonumber
    & \left.+\bshs{-}{p}e^{i\en{k}t_x-i\en{p}t_y+it_s(\en{p}-\en{k})}-\bshs{+}{p}e^{-i\en{k}t_x+i\en{p}t_y-it_s(\en{p}-\en{k})}-\bshs{-}{p}e^{-i\en{k}t_x-i\en{p}t_y+it_s(\en{p}+\en{k})}\right].
\end{align*}
Focusing instead on the term $Z_1^{\mathfrak{B}}$, with similar computations one obtains:
\begin{align*}
    Z_1^{\mathfrak{B}}=&-2ie\int\di^3\textbf{s}\int_{-\infty}^{t_x}\di t_s\chi(t_s)h(\textbf{s})A_0(\textbf{s})\int\int\frac{\di^3\textbf{p}\di^3\textbf{k}}{(2\pi)^6}\frac{1}{4\en{p}}e^{i\textbf{p}(\textbf{s}-\textbf{y})+i\textbf{k}(\textbf{x}-\textbf{s})}\left[\bshs{+}{p}e^{i\en{k}t_x+i\en{p}t_y-it_s(\en{p}+\en{k})}+\right.\\\nonumber
    & \left.+\bshs{-}{p}e^{i\en{k}t_x-i\en{p}t_y+it_s(\en{p}-\en{k})}+\bshs{+}{p}e^{-i\en{k}t_x+i\en{p}t_y-it_s(\en{p}-\en{k})}+\bshs{-}{p}e^{-i\en{k}t_x-i\en{p}t_y+it_s(\en{p}+\en{k})}\right].
\end{align*}
Analogous steps and notation give for $Z_2$:
\begin{align*}
    Z_2^{\mathfrak{A}} =&-e\int\di^3\textbf{s}\int_{-\infty}^{+\infty}\di t_s\dot{\chi}(t_s)h(\textbf{s})A_0(\textbf{s})\int\int\frac{\di^3\textbf{p}\di^3\textbf{k}}{(2\pi)^6}\frac{1}{4\en{p}\en{k}}e^{i\textbf{p}(\textbf{x}-\textbf{s})+i\textbf{k}(\textbf{y}-\textbf{s})}\left[\bshs{+}{p}e^{i\en{k}t_y-i\en{p}t_x+it_s(\en{p}-\en{k})}\right.\\\nonumber
    & \left.+\bshs{-}{p}e^{i\en{p}t_x+i\en{k}t_y-it_s(\en{p}+\en{k})}-\bshs{+}{p}e^{-i\en{p}t_x-i\en{k}t_y+it_s(\en{p}+\en{k})}-\bshs{-}{p}e^{i\en{p}t_x-i\en{k}t_y+it_s(\en{k}-\en{p})}\right]
\end{align*}    
and:
\begin{align*}
    Z_2^{\mathfrak{B}}=&2ie\int\di^3\textbf{s}\int_{-\infty}^{t_y}\di t_s\chi(t_s)h(\textbf{s})A_0(\textbf{s})\int\int\frac{\di^3\textbf{p}\di^3\textbf{k}}{(2\pi)^6}\frac{1}{4\en{p}}e^{i\textbf{p}(\textbf{s}-\textbf{x})+i\textbf{k}(\textbf{y}-\textbf{s})}\left[\bshs{+}{p}e^{i\en{k}t_y-i\en{p}t_x+it_s(\en{p}-\en{k})}\right.\\\nonumber
    & \left.+\bshs{-}{p}e^{i\en{p}t_x+i\en{k}t_y-it_s(\en{p}+\en{k})}+\bshs{+}{p}e^{-i\en{p}t_x-i\en{k}t_y+it_s(\en{p}+\en{k})}+\bshs{-}{p}e^{i\en{p}t_x-i\en{k}t_y+it_s(\en{k}-\en{p})}\right].
\end{align*}
This concludes the proof of the proposition by calling $Z^{\mathfrak{A}} = Z_1^{\mathfrak{A}} + Z_2^{\mathfrak{A}}$ and $Z^{\mathfrak{B}} = Z_1^{\mathfrak{B}} + Z_2^{\mathfrak{B}}$,
which correspond respectively to the first two diagrams and to the second two diagrams in Figure \ref{fig:Z}.
\end{proof}

   \begin{proof}[\textbf{Proof of Proposition \ref{prop_expl_pot}}]
We consider the adiabatic limit $h\to 1$ and the scalar potential:
\begin{equation*}
    A_0(\textbf{s})=s_i.
\end{equation*}
We start evaluating the term $Z^{\mathfrak{A}}$ 
formed by the sum of
$Z_1^{\mathfrak{A}}$ and $Z_2^{\mathfrak{A}}$ in \eqref{first_order_not_correct}, which correspond respectively to the first and the second diagram in Figure \ref{fig:Z}. By direct computation we obtain:
\begin{align*}
  Z_1^{\mathfrak{A}}=&-ie\int_{-\infty}^{+\infty}\di t_s\dot{\chi}(t_s)\int\frac{\di^3\textbf{p}}{(2\pi)^3}\frac{e^{i\textbf{p}(\textbf{x}-\textbf{y})}}{4\omega_\textbf{p}^2}(ix_i-\frac{p_i}{\en{p}^2})\left[\bshs{+}{p}e^{i\en{p}(t_x+t_y-2t_s)}+\right.\\
        & \left.+\bshs{-}{p}e^{i\en{p}(t_x-t_y)}-\bshs{+}{p}e^{-i\en{p}(t_x-t_y)}-\bshs{-}{p}e^{-i\en{p}(t_x+t_y-2t_s)}\right]+\\
        &+\frac{ip_i(t_x-t_s)}{4(2\pi)^3\en{p}^3}e^{i\textbf{p}(\textbf{x}-\textbf{y})}\left[\bshs{+}{p}e^{i\en{p}(t_x+t_y-2t_s)}+\bshs{-}{p}e^{i\en{p}(t_x-t_y)}+\bshs{+}{p}e^{-i\en{p}(t_x-t_y)}+\bshs{-}{p}e^{-i\en{p}(t_x+t_y-2t_s)}\right].\\
    Z_2^{\mathfrak{A}}=&-ie\int_{-\infty}^{+\infty}\di t_s\dot{\chi}(t_s)\int\frac{\di^3\textbf{p}}{(2\pi)^3}\frac{e^{i\textbf{p}(\textbf{x}-\textbf{y})}}{4\omega_\textbf{p}^2}(-iy_i-\frac{p_i}{\en{p}^2})\left[-\bshs{+}{p}e^{-i\en{p}(t_x+t_y-2t_s)}+\right.\\
        & \left.-\bshs{-}{p}e^{i\en{p}(t_x-t_y)}+\bshs{+}{p}e^{-i\en{p}(t_x-t_y)}+\bshs{-}{p}e^{i\en{p}(t_x+t_y-2t_s)}\right]+\\
        &+\frac{ip_i(t_y-t_s)}{4(2\pi)^3\en{p}^3}e^{i\textbf{p}(\textbf{x}-\textbf{y})}\left[\bshs{+}{p}e^{-i\en{p}(t_x+t_y-2t_s)}+ \bshs{-}{p}e^{i\en{p}(t_x-t_y)}+\bshs{+}{p}e^{-i\en{p}(t_x-t_y)}+\bshs{-}{p}e^{i\en{p}(t_x+t_y-2t_s)}\right].  
\end{align*}
Summing the two addends $Z_1^{\mathfrak{A}}$ and $Z_2^{\mathfrak{A}}$ and smearing against the two families of test functions $f_t(x)=f(t_x-t,\textbf{x})$, $g_t(y)=g(t_y-t,\textbf{y})$ we get:
\begin{multline*}
    \left(Z_1^{\mathfrak{A}}+Z_2^{\mathfrak{A}}\right)(f_t,g_t)=\mathcal{Z}_{\chi,f,g}(t)+et\int\frac{\di^3\textbf{p}}{4\en{p}^3(2\pi)^3}p_i\left[2\bshs{-}{p}\hat{f}(\en{p},-\textbf{p})\hat{g}(-\en{p},\textbf{p})+2\bshs{+}{p}\hat{f}(-\en{p},-\textbf{p})\hat{g}(\en{p},\textbf{p})\right]\\
    +et\int\frac{\di^3\textbf{p}}{4\en{p}^3(2\pi)^3}p_i(\bshs{-}{p}+\bshs{+}{p})\left[\hat{\dot{\chi}}(-2\en{p})\hat{f}(\en{p},-\textbf{p})\hat{g}(\en{p},\textbf{p})e^{2i\en{p}t}+\hat{\dot{\chi}}(2\en{p})\hat{f}(-\en{p},-\textbf{p})\hat{g}(-\en{p},\textbf{p})e^{-2i\en{p}t}\right],
\end{multline*}
where $\mathcal{Z}_{\chi,f,g}(t)$ is a bounded function of $t$ and the relation $\int \di t_s\dot{\chi}(t_s)=1$ has been used. It is now easy to see that the absolute value of the second addend in the previous expression grows linearly with the time $t$. On the contrary, using a slightly modified version of Lemma A.1 in \cite{Meda_2022} (stationary phase method) or by Appendix \ref{app: 1}, it is possible to show that the absolute value of the third addend in the previous expression goes to $0$ in the limit $t\to\infty$ (even if slower than expected for the two-point function).\\ 
We proceed considering $Z^{\mathfrak{B}}$ in \eqref{first_order_not_correct} which is formed by the sum of the 
terms $Z_1^{\mathfrak{B}}$ and $Z_2^{\mathfrak{B}}$, which correspond to the third and fourth diagram in Figure \ref{fig:Z}. We start evaluating explicitly the contribution $Z_1^{\mathfrak{B}}$ that, after Fourier transforming in $\mathbf{s}_1$:
\begin{align*}
    Z_1^{\mathfrak{B}}=&-2e\int_{-\infty}^{t_x}\di t_s\chi(t_s)\int\di^3\textbf{p}\frac{1}{4(2\pi)^3}ix_ie^{-i\textbf{p}(\textbf{y}-\textbf{x})}\frac{1}{\en{p}}\left[\bshs{+}{p}e^{i\en{p}(t_y+t_x)-i2\en{p}t_s}+\right.\\\nonumber
        & \left.+\bshs{-}{p}e^{i\en{p}(t_x-t_y)}+\bshs{+}{p}e^{-i\en{p}(t_x-t_y)}+\bshs{-}{p}e^{-i\en{p}(t_x+t_y)+i2\en{p}t_s}\right]+\\
        & \frac{1}{4(2\pi)^3}e^{-i\textbf{p}(\textbf{y}-\textbf{x})}\frac{p_i}{\en{p}^2}\left[i(t_x-t_s)\bshs{+}{p}e^{i\en{p}(t_y+t_x)-i2\en{p}t_s}+\right.\\\nonumber
        & \left.+i(t_x-t_s)\bshs{-}{p}e^{i\en{p}(t_x-t_y)}+i(t_s-t_x)\bshs{+}{p}e^{-i\en{p}(t_x-t_y)}+i(t_s-t_x)\bshs{-}{p}e^{-i\en{p}(t_x+t_y)+i2\en{p}t_s}\right].
\end{align*}
Calling the terms in the integrand that contain a dependence on the variable of integration $t_s$ in the exponents $Z_{1,+}^{\mathfrak{B}}$, we obtain the following boundary term:
\begin{align*}
    Z_{1,+}^{\mathfrak{B},Inv}=&-2e\int\di^3\textbf{p}\frac{1}{4(2\pi)^3}ix_ie^{-i\textbf{p}(\textbf{y}-\textbf{x})}\frac{1}{\en{p}}\left[\frac{\bshs{+}{p}}{-2i\en{p}}e^{i\en{p}(t_y-t_x)}+\frac{\bshs{-}{p}}{2i\en{p}}e^{-i\en{p}(t_y-t_x)}\right]+\\\nonumber
    &\frac{1}{4(2\pi)^3}e^{-i\textbf{p}(\textbf{y}-\textbf{x})}\frac{p_i}{\en{p}^2}\left[\frac{i\bshs{+}{p}e^{i\en{p}(t_y-t_x)}}{-4\en{p}^2}-\frac{i\bshs{-}{p}e^{-i\en{p}(t_y-t_x)}}{-4\en{p}^2}\right]
\end{align*}
and the following bulk term:
\begin{align*}
    Z_{1,+}^{\mathfrak{B},Bulk}=&-2e\int_{-\infty}^{+\infty}\di t_s\dot{\chi}(t_s)\int\di^3\textbf{p}\frac{1}{4(2\pi)^3}x_ie^{-i\textbf{p}(\textbf{y}-\textbf{x})}\frac{1}{2\en{p}^2}\left[\bshs{+}{p}e^{i\en{p}(t_y+t_x)-i2\en{p}t_s}-\bshs{-}{p}e^{-i\en{p}(t_y+t_x)+2i\en{p}t_s}\right]+\\\nonumber
    &\frac{1}{4(2\pi)^3}e^{-i\textbf{p}(\textbf{y}-\textbf{x})}\frac{p_i}{\en{p}^2}\left[-t_xe^{i\en{p}(t_y+t_x)-i2\en{p}t_s}\frac{\bshs{+}{p}}{2\en{p}}-(i2\en{p}t_s+1)e^{i\en{p}(t_x+t_y)-i2\en{p}t_s}\frac{i\bshs{+}{p}}{4\en{p}^2}\right.\\\nonumber
    &\left.-t_xe^{-i\en{p}(t_y+t_x)+i2\en{p}t_s}\frac{\bshs{-}{p}}{2\en{p}}-(i2\en{p}t_s-1)e^{-i\en{p}(t_x+t_y)+i2\en{p}t_s}\frac{i\bshs{-}{p}}{4\en{p}^2}\right].
\end{align*}
To obtain the previous expression we used the form of $\chi$ and the fact that the time $t_x$ is supposed to be after the end of the complete switch-on. In conclusion we have:
\begin{equation}    
    Z_1^{\mathfrak{B}} = Z_{1,+}^{\mathfrak{B},Inv}+Z_{1,+}^{\mathfrak{B},Bulk}+Z_{1,-}^{\mathfrak{B}},
\end{equation}
where we used the notation $Z_{1,-}^{\mathfrak{B}}$ to denote the contribution in $Z_1^{\mathfrak{B}}$ that does not show any dependence on the integration variable $t_s$ in the exponents:
\begin{align*}
    Z_{1,-}^{\mathfrak{B}}=&-2e\int_{-\infty}^{t_x}\di t_s\chi(t_s)\int\di^3\textbf{p}\frac{1}{4(2\pi)^3}ix_ie^{-i\textbf{p}(\textbf{y}-\textbf{x})}\frac{1}{\en{p}}\left[
        \bshs{-}{p}e^{i\en{p}(t_x-t_y)}+\bshs{+}{p}e^{-i\en{p}(t_x-t_y)}\right]+\\
        & \frac{1}{4(2\pi)^3}e^{-i\textbf{p}(\textbf{y}-\textbf{x})}\frac{p_i}{\en{p}^2}\left[i(t_x-t_s)\bshs{-}{p}e^{i\en{p}(t_x-t_y)}+i(t_s-t_x)\bshs{+}{p}e^{-i\en{p}(t_x-t_y)}\right].
\end{align*}

In a similar way we obtain the following expression for the contribution $Z_2^{\mathfrak{B}}$:
\begin{equation*}
    Z_2^{\mathfrak{B}}=Z_{2,+}^{\mathfrak{B},Inv}+Z_{2,+}^{\mathfrak{B},Bulk}+Z_{2,-}^{\mathfrak{B}},
\end{equation*}
where, with analogous notation:
\begin{align*}
    Z_{2,-}^{\mathfrak{B}}&=2e\int_{-\infty}^{t_y}\di t_s\chi(t_s)\int\di^3\textbf{p}\frac{1}{4(2\pi)^3}iy_ie^{-i\textbf{p}(\textbf{y}-\textbf{x})}\frac{1}{\en{p}}\left[
    \bshs{-}{p}e^{i\en{p}(t_x-t_y)}+\bshs{+}{p}e^{-i\en{p}(t_x-t_y)}\right]+\\
    & \frac{1}{4(2\pi)^3}e^{-i\textbf{p}(\textbf{y}-\textbf{x})}\frac{p_i}{\en{p}^2}\left[i(t_y-t_s)\bshs{-}{p}e^{i\en{p}(t_x-t_y)}+i(t_s-t_y)\bshs{+}{p}e^{-i\en{p}(t_x-t_y)}\right].\\
    Z_{2,+}^{\mathfrak{B},Inv}&=-2e\int\di^3\textbf{p}\frac{1}{4(2\pi)^3}iy_ie^{-i\textbf{p}(\textbf{y}-\textbf{x})}\frac{1}{\en{p}}\left[\frac{\bshs{+}{p}}{-2i\en{p}}e^{i\en{p}(t_y-t_x)}+\frac{\bshs{-}{p}}{2i\en{p}}e^{-i\en{p}(t_y-t_x)}\right]+\\\nonumber
    &\frac{1}{4(2\pi)^3}e^{-i\textbf{p}(\textbf{y}-\textbf{x})}\frac{p_i}{\en{p}^2}\left[\frac{i\bshs{+}{p}e^{i\en{p}(t_y-t_x)}}{-4\en{p}^2}+\frac{i\bshs{-}{p}e^{-i\en{p}(t_y-t_x)}}{-4\en{p}^2}\right]\\
    Z_{2,+}^{\mathfrak{B},Bulk}&=-2e\int_{-\infty}^{+\infty}\di t_s\dot{\chi}(t_s)\int\di^3\textbf{p}\frac{1}{4(2\pi)^3}y_ie^{-i\textbf{p}(\textbf{y}-\textbf{x})}\frac{1}{2\en{p}^2}\left[\bshs{+}{p}e^{-i\en{p}(t_y+t_x)+i2\en{p}t_s}-\bshs{-}{p}e^{i\en{p}(t_y+t_x)-2i\en{p}t_s}\right]+\\\nonumber
    &\frac{1}{4(2\pi)^3}e^{-i\textbf{p}(\textbf{y}-\textbf{x})}\frac{p_i}{\en{p}^2}\left[t_ye^{i\en{p}(t_y+t_x)-i2\en{p}t_s}\frac{\bshs{-}{p}}{2\en{p}}+(i2\en{p}t_s+1)e^{i\en{p}(t_x+t_y)-i2\en{p}t_s}\frac{i\bshs{-}{p}}{4\en{p}^2}\right.\\\nonumber
    &\left.+t_ye^{-i\en{p}(t_y+t_x)+i2\en{p}t_s}\frac{\bshs{+}{p}}{2\en{p}}+(i2\en{p}t_s-1)e^{-i\en{p}(t_x+t_y)+i2\en{p}t_s}\frac{i\bshs{+}{p}}{4\en{p}^2}\right].
\end{align*}
The addends $Z_{1,+}^{\mathfrak{B},Inv}$, $Z_{2,+}^{\mathfrak{B},Inv}$ are invariant under time translations. Concerning the terms $Z_{1,+}^{\mathfrak{B},Bulk}$, $Z_{2,+}^{\mathfrak{B},Bulk}$, instead, it is again possible to show, by stationary phase methods or adapting Lemma A.1 in \cite{Meda_2022}, that they decay in the limit $t\to\infty$. The last contributes to consider are the one coming from the terms $Z_{1,-}^{\mathfrak{B}}$ and $Z_{2,-}^{\mathfrak{B}}$. For simplicity, we consider the case in $t_x=t_y=t/2$ and we smear the distribution $Z_{1,-}^{\mathfrak{B}}+Z_{2,-}^{\mathfrak{B}}(t,\textbf{x},\textbf{y})$ only in space, against two test functions $f(x),g(y)\in\mathcal{C}^\infty_0(\mathbb{R}^3)$:
\begin{equation*}
    \int \di^3\textbf{x}\int\di^3\textbf{y} f(\textbf{x})g(\textbf{y})\left(Z_{1,-}^{\mathfrak{B}}(t,\textbf{x},\textbf{y})+Z_{2,-}^{\mathfrak{B}}(t,\textbf{x},\textbf{y})\right)=-et\int\di^3\textbf{p}\frac{\bshs{-}{p}+\bshs{+}{p}}{4(2\pi)^3\en{p}}\partial_{p_i}\left(\hat{f}(-\textbf{p})\hat{g}(\textbf{p})\right).
\end{equation*}
With an analogous computation, we get for the terms in $Z_1^{\mathfrak{A}}+Z_2^{\mathfrak{A}}$ that grow linearly in $t$ ($Z_1^{\mathfrak{A},l}+Z_2^{\mathfrak{A},l}$):
\begin{equation*}
    \int \di^3\textbf{x}\int\di^3\textbf{y} f(\textbf{x})g(\textbf{y})\left(Z_1^{\mathfrak{A},l}(t,\textbf{x},\textbf{y})+Z_2^{\mathfrak{B},l}(t,\textbf{x},\textbf{y})\right)=et\int\di^3\textbf{p}\frac{p_i(\bshs{-}{p}+\bshs{+}{p})}{4(2\pi)^3\en{p}^3}\left(\hat{f}(-\textbf{p})\hat{g}(\textbf{p})\right).
\end{equation*}
It is now evident that it is possible to chose the function $f$ and $g$ in a way such that the linearly growing term in $Z_1^{\mathfrak{A}}+Z_2^{\mathfrak{A}}$ is not canceled by a term in $Z_1^{\mathfrak{B}}+Z_2^{\mathfrak{B}}$. The proof of the proposition follows by taking the absolute value of the smeared distribution \eqref{first_order_not_correct_sm} and the large time limit $t\to\infty$.
   \end{proof}

\begin{proof}[\textbf{Proof of Proposition \ref{prop_correct_el}}]
We want to compute:
\begin{equation*}
    \omega^{\beta,V}\left(R_V(\varphi^\dagger(x)\star\varphi(y))\right).
\end{equation*}
We recall that the following perturbative expression holds for every 
$F\in\mathcal{A}$:
\begin{equation}
       \omega^{\beta,V}(R_V(F))=\sum_{n\geq0}(-1)^n\int_{\beta S_n}\di u_1...\di u_n\;\omega^{\beta,\mathcal{T}}\left(R_V(F)\otimes \tau_{iu_1}K\otimes...\otimes\tau_{iu_n}K\right)
\end{equation}
and that, in this specific case, we have:
\begin{equation*}
     K=R_V(-\Dot{V})=-\int\Dot{\chi}(t_x)h(\textbf{x})\left[ieA^{\mu}(x)R_V(\varphi^{\dagger}\partial_{\mu}\varphi-\partial_{\mu}\varphi^{\dagger}\varphi)-e^2A_{\mu}(x)A^{\mu}(x)R_V(\varphi^{\dagger}\varphi)\right]\di^4x.
\end{equation*}
The term of order $0$ in the interaction coincides with the two-point function of the KMS state for the free theory. The corrections of order one coincide with the one obtained in \eqref{first_order_not_correct}, except for the additional term $E$:
\begin{equation*}
    E\coloneqq\int_0^{\beta}\di u\;\omega^{\beta,\mathcal{T}}(\varphi^{\dagger}(x)\star\varphi(y)\otimes\alpha_{iu}\Dot{V}).
\end{equation*}
With our specific choice of the potential, we have (at first order in $e$):
\begin{equation*}
    \dot{V}=ie\int\di^4x\dot{\chi}(t_x)h(\textbf{x})A_0(\textbf{x})(\varphi^\dagger(x)\partial_0\varphi(x)-\partial_0\varphi^\dagger(x)\varphi(x)).
\end{equation*}
Neglecting the vanishing contributions once evaluated on the state we get:
\begin{equation*}\label{3.2.37}
    \varphi(y)\star\dot{V}=2ie\int\di^4s\dot{\chi}(t_s)h(\textbf{s})A_0(\textbf{s})\omega_2^\beta(y,s)\partial_0^s\varphi(s)+ie\int\di^4s\omega_2^\beta(y,s)\ddot{\chi}(t_s)h(\textbf{s})A_0(\textbf{s})\varphi(s).
\end{equation*}
and:
\begin{equation*}
    \varphi^\dagger(x)\star\varphi(y)\star\dot{V}=-ie\int\di^4s\ddot{\chi}(t_s)h(\textbf{s})A_0(\textbf{s})\omega_2^\beta(x,s)\omega_2^\beta(y,s)-2ie\int\di^4s\dot{\chi}(t_s)h(\textbf{s})A_0(\textbf{s})\omega_2^\beta(x,s)\partial_0^s\omega_2^\beta(y,s)
\end{equation*}
In conclusion, at the first order in $e$, the correction $E$ is given by:
\begin{multline*}
    E^1=\overbrace{-ie\int\di^4s\int_0^\beta \di u\ddot{\chi}(t_s)h(\textbf{s})A_0(\textbf{s})\omega_2^\beta(x,s+iue_0)\omega_2^\beta(y,s+iue_0)}^{E^{1,\mathfrak{A}}}\\
    \underbrace{-2ie\int\di^4s\int_0^\beta\di u\dot{\chi}(t_s)h(\textbf{s})A_0(\textbf{s})\omega_2^\beta(x,s+iue_0)\partial_0^s\omega_2^\beta(y,s+iue_0)}_{E^{1,\mathfrak{B}}}.
\end{multline*}
Here we denoted $e_0\equiv(1,0,0,0)$. Substituting the explicit expression for the two-point function of the free KMS state, integrating by parts in $t_s$ the term $E^{1,\mathfrak{A}}$ and finally integrating everything in $u$ we get:
\begin{align*}
    E^{1,\mathfrak{A}}=&-e\int\di^4s\dot{\chi}(t_s)h(\textbf{s})A_0(\textbf{s})\int\int\frac{\di^3\textbf{k}\di^3\textbf{p}}{(2\pi)^6}e^{i\textbf{p}(\textbf{x}-\textbf{s})+i\textbf{k}(\textbf{y}-\textbf{s})}\frac{1}{4\en{p}\en{k}}\times\\\nonumber
    &\times\left[-\bshs{+}{p}\bshs{+}{k}e^{-i\en{p}(t_x-t_s)-i\en{k}(t_y-t_s)}\left(e^{-(\en{k}+\en{p})\beta}-1\right)-\bshs{-}{p}\bshs{+}{k}e^{i\en{p}(t_x-t_s)-i\en{k}(t_y-t_s)}\left(e^{(\en{p}-\en{k})\beta}-1\right)\right.\\\nonumber
    &\left.-\bshs{+}{p}\bshs{-}{k}e^{-i\en{p}(t_x-t_s)+i\en{k}(t_y-t_s)}\left(e^{(\en{k}-\en{p})\beta}-1\right)-\bshs{-}{p}\bshs{-}{k}e^{i\en{p}(t_x-t_s)+i\en{k}(t_y-t_s)}\left(e^{(\en{p}+\en{k})\beta}-1\right)\right]\\
    E^{1,\mathfrak{B}}=&2e\int\di^4s\dot{\chi}(t_s)h(\textbf{s})A_0(\textbf{s})\int\int\frac{\di^3\textbf{k}\di^3\textbf{p}}{(2\pi)^6}e^{i\textbf{p}(\textbf{x}-\textbf{s})+i\textbf{k}(\textbf{y}-\textbf{s})}\frac{1}{4\en{p}}\times\\\nonumber
     &\times\left[-\frac{\bshs{+}{p}\bshs{+}{k}}{\en{p}+\en{k}}e^{-i\en{p}(t_x-t_s)-i\en{k}(t_y-t_s)}\left(e^{-(\en{k}+\en{p})\beta}-1\right)-\frac{\bshs{-}{p}\bshs{+}{k}}{\en{k}-\en{p}}e^{i\en{p}(t_x-t_s)-i\en{k}(t_y-t_s)}\left(e^{(\en{p}-\en{k})\beta}-1\right)\right.\\\nonumber
    &\left.+\frac{\bshs{+}{p}\bshs{-}{k}}{\en{p}-\en{k}}e^{-i\en{p}(t_x-t_s)+i\en{k}(t_y-t_s)}\left(e^{(\en{k}-\en{p})\beta}-1\right)-\frac{\bshs{-}{p}\bshs{-}{k}}{\en{p}+\en{k}}e^{i\en{p}(t_x-t_s)+i\en{k}(t_y-t_s)}\left(e^{(\en{p}+\en{k})\beta}-1\right)\right].
\end{align*}
Our aim now is to sum the correction term $E$ to the expression \eqref{first_order_not_correct}. For this purpose, we perform in \eqref{first_order_not_correct} an integration by parts in the variable $t_s$ getting boundary terms:
    \begin{align}\label{1binv}
        Z_1^{\mathfrak{B},Inv}=&2e\int\di^3\textbf{s}h(\textbf{s})A_0(\textbf{s})\int\int\frac{\di^3\textbf{p}\di^3\textbf{k}}{(2\pi)^6}\frac{1}{4\en{p}}e^{i\textbf{p}(\textbf{s}-\textbf{y})+i\textbf{k}(\textbf{x}-\textbf{s})}\left[\frac{\bshs{+}{p}}{(\en{k}+\en{p})}e^{i\en{p}(t_y-t_x)}\right.\\\nonumber
        & \left.+\frac{\bshs{-}{p}}{(\en{k}-\en{p})}e^{-i\en{p}(t_y-t_x)}-\frac{\bshs{+}{p}}{(\en{k}-\en{p})}e^{i\en{p}(t_y-t_x)}-\frac{\bshs{-}{p}}{(\en{k}+\en{p})}e^{-i\en{p}(t_y-t_x)}\right],\\
        Z_2^{\mathfrak{B},Inv}=&2e\int\di^3\textbf{s}h(\textbf{s})A_0(\textbf{s})\int\int\frac{\di^3\textbf{p}\di^3\textbf{k}}{(2\pi)^6}\frac{1}{4\en{p}}e^{i\textbf{p}(\textbf{s}-\textbf{x})+i\textbf{k}(\textbf{y}-\textbf{s})}\left[\frac{\bshs{+}{p}}{(\en{k}+\en{p})}e^{i\en{p}(t_y-t_x)}\right.\\\nonumber
        & \left.+\frac{\bshs{-}{p}}{(\en{k}-\en{p})}e^{-i\en{p}(t_y-t_x)}-\frac{\bshs{+}{p}}{(\en{k}-\en{p})}e^{i\en{p}(t_y-t_x)}-\frac{\bshs{-}{p}}{(\en{k}+\en{p})}e^{-i\en{p}(t_y-t_x)}\right],
    \end{align}
    whose sum is $Z^{\mathfrak{B},Inv}$ mentioned in the thesis of the Proposition
    and the following bulk terms:
    \begin{align}\label{1bulk}
        Z_1^{\mathfrak{B},bulk}=&2e\int\di^3\textbf{s}\int_{-\infty}^{+\infty}\di t_s\dot{\chi}(t_s)h(\textbf{s})A_0(\textbf{s})\int\int\frac{\di^3\textbf{p}\di^3\textbf{k}}{(2\pi)^6}\frac{1}{4\en{p}}e^{i\textbf{p}(\textbf{s}-\textbf{y})+i\textbf{k}(\textbf{x}-\textbf{s})}\\\nonumber
        &\times\left[\frac{\bshs{+}{p}}{-(\en{p}+\en{k})}e^{i\en{k}t_x+i\en{p}t_y-it_s(\en{p}+\en{k})}+\frac{\bshs{-}{p}}{\en{p}-\en{k}}e^{i\en{k}t_x-i\en{p}t_y+it_s(\en{p}-\en{k})}\right.\\\nonumber
        &\left.+\frac{\bshs{+}{p}}{\en{k}-\en{p}}e^{-i\en{k}t_x+i\en{p}t_y-it_s(\en{p}-\en{k})}+\frac{\bshs{-}{p}}{\en{p}+\en{k}}e^{-i\en{k}t_x-i\en{p}t_y+it_s(\en{p}+\en{k})}\right].\\
        Z_2^{\mathfrak{B},bulk}=&-2e\int\di^3\textbf{s}\int_{-\infty}^{+\infty}\di t_s\dot{\chi}(t_s)h(\textbf{s})A_0(\textbf{s})\int\int\frac{\di^3\textbf{p}\di^3\textbf{k}}{(2\pi)^6}\frac{1}{4\en{p}}e^{i\textbf{p}(\textbf{s}-\textbf{x})+i\textbf{k}(\textbf{y}-\textbf{s})}\\\nonumber
        &\times\left[\frac{\bshs{+}{p}}{\en{p}-\en{k}}e^{i\en{k}t_y-i\en{p}t_x+it_s(\en{p}-\en{k})}+\frac{\bshs{-}{p}}{-\en{p}-\en{k}}e^{i\en{p}t_x+i\en{k}t_y-it_s(\en{p}+\en{k})}\right.\\\nonumber
        & \left.+\frac{\bshs{+}{p}}{\en{p}+\en{k}}e^{-i\en{p}t_x-i\en{k}t_y+it_s(\en{p}+\en{k})}+\frac{\bshs{-}{p}}{\en{k}-\en{p}}e^{i\en{p}t_x-i\en{k}t_y+it_s(\en{k}-\en{p})}\right].
    \end{align}
In \eqref{1binv} was used that the fields are considered for times consecutive to the switching on of the interaction (since we are interested in  the large time limit) and in \eqref{1bulk} the support properties of the function $\dot{\chi}(t_s)$, in order to extend the integral in $t_s$ to $+\infty$.\\
We now sum the coefficients in $E^1, Z_1^{\mathfrak{B}, bulk}, Z_2^{\mathfrak{B}, bulk}$ and $Z^{\mathfrak{A}}$ which have the same exponential modes.
For simplicity, we report here the computation only for a specific mode (all the others are computed in the same way):
\begin{align*}
    &e\int\di^4s\dot{\chi}(t_s)A_0(\textbf{s})\int\int\frac{\di^3\textbf{k}\di^3\textbf{p}}{(2\pi)^6}e^{i\textbf{p}(\textbf{s}-\textbf{y})+i\textbf{k}(\textbf{x}-\textbf{s})}e^{i\en{k}t_x+i\en{p}t_y-it_s(\en{p}+\en{k})}\left[\frac{\bshs{+}{p}}{4\en{p}\en{k}}-\frac{2\bshs{+}{p}}{4\en{p}(\en{p}+\en{k})}\right.\\\nonumber
    &\left.-\frac{\bshs{-}{k}}{4\en{p}\en{k}}+\frac{2\bshs{-}{k}}{4\en{k}(\en{p}+\en{k})}+\frac{\bshs{-}{k}\bshs{-}{p}}{4\en{p}\en{k}}\left(e^{(\en{p}+\en{k})\beta}-1\right)-\frac{2\bshs{-}{k}\bshs{-}{p}}{4\en{k}(\en{p}+\en{k})}\left(e^{(\en{p}+\en{k})\beta}-1\right)\right]\\\nonumber
    &=e\int\di^4s\dot{\chi}(t_s)A_0(\textbf{s})\int\int\frac{\di^3\textbf{k}\di^3\textbf{p}}{(2\pi)^6}e^{i\textbf{p}(\textbf{s}-\textbf{y})+i\textbf{k}(\textbf{x}-\textbf{s})}e^{i\en{k}t_x+i\en{p}t_y-it_s(\en{p}+\en{k})}\left[\frac{\bshs{+}{p}}{2\en{p}\en{k}}-\frac{\bshs{+}{p}}{2\en{p}\en{k}}\right]=0.
\end{align*}
In conclusion, at first order in perturbation theory, only the boundary terms $Z^{\mathfrak{B},Inv}$ survive. We thus obtain the following equality:
\begin{equation*}
    \omega_2^{\beta,V}(x,y)=\omega_2^\beta(x,y)+Z_1^{\mathfrak{B},Inv}+Z_2^{\mathfrak{B},Inv}+\mathcal{O}(e^2)
\end{equation*}
that proves the proposition.
\end{proof}

\subsection{Proof of section \ref{Sec:ScalarMag}}\label{appendix_scalar_mag}
\begin{proof}[\textbf{Proof of Proposition \ref{prop_correct_mag}}]     
We study the correction at the first order in the coupling constant $e$ to the two-point function:
\begin{equation}\label{cosadacuipartireincap4}
    \omega^{\beta,V}\left(R_V(\varphi^\dagger(x)\star\varphi(y))\right),  
\end{equation}
and we start by the terms arising from the expansion of the Bogoliubov map. Recall that in this specific case the interaction Lagrangian density $V$ has the following form:
\begin{equation*}
     V=\int \chi(t_x)h(\textbf{x})\left[ ieA^{i}(\varphi^{\dagger}\partial_{i}\varphi-\partial_{i}\varphi^{\dagger}\varphi)-e^2A_{i}A^{i}\varphi^{\dagger}\varphi\right]\;\di^4x.
\end{equation*}
Expanding Bogoliubov formula in \eqref{cosadacuipartireincap4} we get the following two contributes linear in the interaction $V$:
\begin{multline}\label{2terminicap4}
  Z_1+Z_2\coloneqq \omega^\beta\left([-iV\star\varphi^{\dagger}(x)+iV\cdot_T\varphi^{\dagger}(x)]\star\varphi(y)\right)
  +\omega^\beta\left(\varphi^{\dagger}(x)\star[-iV\star\varphi(y)+iV\cdot_T\varphi(y)]\right).
\end{multline}
We now explicitly compute the first of the two addends ($Z_1$). After an integration by parts, the following integral kernel for the functional derivative of $V$ is obtained:
\begin{equation*}\label{funcdervcap4}
    \frac{\delta V}{\delta\varphi}=\chi h\left[-ie\partial^{i}(A_{i}\varphi^{\dagger})-ieA^{i}\partial_{i}\varphi^{\dagger}-e^2A_{i}A^{i}\varphi^{\dagger}\right].
\end{equation*}
In the previous expression, the term proportional to the spatial derivative of the cut-off function $h(\textbf{s})$ has been neglected. This is due to the fact that, in the upcoming computation, we are interested in considering only the case in which the adiabatic limit ($h\to1$) is taken. As shown in \cite{FredenhagenLindnerKMS_2014} the adiabatic limit can be exchanged with the space integration, via dominated convergence theorem, thanks to the properties of decay of the two-point functions of the KMS state of the free theory (for the massive scalar field). In our case, the spatial derivative of the cut-off function $\partial_ih(\textbf{s})$ converges pointwise to the constant $0$. Therefore, we can conclude that in the adiabatic limit the terms proportional to $\partial_i h(\textbf{s})$ vanish.\\
Using now again Equation \eqref{eq: FeyRet} and the convenient choice of the \emph{Coulomb gauge} for the electromagnetic potential, the first addend in \eqref{2terminicap4} reduces to:
\begin{equation}\label{primoaddasost4cap}
      Z_1^{(1)}=-\int \di s\; \omega_2^{\beta,(0)}(s,y)\chi(t_s)h(\textbf{s})\left[2ieA_{i}(s)\partial^{i}\Delta_R(x,s)\right].
\end{equation}
Here, we have already neglected terms of second order in the coupling constant $e$ and in the vector potential $A_i$.\\
From an analogous computation for the second term in \eqref{2terminicap4} ($Z_2$) we get the following result at first order in $e$:
\begin{equation}\label{secondocontrlinearecap4}
    Z_2^{(1)}\coloneqq\int \di s\; \omega_2^{\beta,(0)}(x,s)\chi(t_s)h(\textbf{s})\left[2ieA_{i}(s)\partial^{i}\Delta_R(y,s)\right].
\end{equation}
We now substitute the explicit expression for the two-point function of the free KMS state $\omega_2^{\beta}$ and for the retarded propagator of the free theory $\Delta_R$. In addition, focusing on the term $Z_1^{(1)}$ in the limit $h \to 1$:
\begin{multline*}
    Z_1^{(1)}=\int_{-\infty}^{t_x}\di t_s\int \di^3\textbf{s}\int \di^3\textbf{p}\int \di^3\textbf{k}\;k_i\frac{2eA_i(\mathbf{s})}{(2\pi)^6}\chi(t_s)\;\frac{\sin{\en{k}(t_x-t_s)}}{2\en{p}\en{k}}\cdot\\
    \left[\bshs{+}{p}e^{-i\en{p}(t_s-t_y)}+\bshs{-}{p}e^{i\en{p}(t_s-t_y)}\right]e^{i\mathbf{p}(\mathbf{s-y})+i\mathbf{k}(\mathbf{x-s})}.
\end{multline*}
Integrating by parts, assuming the fields $\varphi^\dagger(x)$ and $\varphi(y)$ to be supported after the complete switch on of the interaction, we obtain the following two terms:
\begin{multline*}
   \mathfrak{C}_1\coloneqq\int \di^3\textbf{s}\int \di^3\textbf{p}\int \di^3\textbf{k}\;\frac{eA_i(\mathbf{s})}{(2\pi)^6}\frac{k_i}{2\en{p}\en{k}}
    \left[\frac{\bshs{+}{p}e^{-i\en{p}(t_x-t_y)}}{(\en{k}+\en{p})}+\frac{\bshs{+}{p}e^{-i\en{p}(t_x-t_y)}}{(\en{k}-\en{p})}\right.\\
   \left.+\frac{\bshs{-}{p}e^{i\en{p}(t_x-t_y)}}{(\en{k}-\en{p})}+\frac{\bshs{-}{p}e^{i\en{p}(t_x-t_y)}}{(\en{k}+\en{p})}\right]e^{i\mathbf{p}(\mathbf{s-y})+i\mathbf{k}(\mathbf{x-s})},
\end{multline*}
\begin{multline}\label{primoterminenoninvcap4}
    \mathfrak{B}_1\coloneqq\int_{-\infty}^{+\infty}\di t_s\int \di^3\textbf{s}\int \di^3\textbf{p}\int \di^3\textbf{k}\frac{eA_i(\mathbf{s})}{(2\pi)^6}\Dot{\chi}(t_s)\frac{k_i}{2\en{p}\en{k}}
    \left[-\frac{\bshs{+}{p}e^{i(\en{k}(t_x-t_s)-\en{p}(t_s-t_y))}}{(\en{k}+\en{p})}\right.\\
    \left.-\frac{\bshs{+}{p}e^{-i(\en{k}(t_x-t_s)+\en{p}(t_s-t_y))}}{(\en{k}-\en{p})}
   -\frac{\bshs{-}{p}e^{i(\en{k}(t_x-t_s)+\en{p}(t_s-t_y))}}{(\en{k}-\en{p})}-\frac{\bshs{-}{p}e^{-i(\en{k}(t_x-t_s)-\en{p}(t_s-t_y))}}{(\en{p}+\en{k})}\right]e^{i\mathbf{p}(\mathbf{y-s})+i\mathbf{k}(\mathbf{x-s})} .
\end{multline}
First, observe that the term $\mathfrak{A}$ represents a quantity invariant under time translations of the system. On the contrary the integral \eqref{primoterminenoninvcap4} depends on the switch-on function and it is not invariant under time translations. Repeating a similar analysis for \eqref{secondocontrlinearecap4}, we get the following:
\begin{multline*}
    \mathfrak{C}_2\coloneqq\int \di^3\textbf{s}\int \di^3\textbf{p}\int \di^3\textbf{k}\;\frac{eA_i(\mathbf{s})}{(2\pi)^6}\frac{k_i}{2\en{p}\en{k}}
    \left[\frac{\bshs{+}{p}e^{-i\en{p}(t_x-t_y)}}{(\en{p}-\en{k})}-\frac{\bshs{+}{p}e^{-i\en{p}(t_x-t_y)}}{(\en{k}+\en{p})}\right.+\\
   \left.-\frac{\bshs{-}{p}e^{i\en{p}(t_x-t_y)}}{(\en{k}+\en{p})}-\frac{\bshs{-}{p}e^{i\en{p}(t_x-t_y)}}{(\en{k}-\en{p})}\right]e^{i\mathbf{p}(\mathbf{x-s})+i\mathbf{k}(\mathbf{y-s})},
\end{multline*}
\begin{multline}\label{secondoterminenoninvcap4}
     \mathfrak{B}_2\coloneqq-\int_{-\infty}^{+\infty}\di t_s\int \di^3\textbf{s}\int \di^3\textbf{p}\int \di^3\textbf{k}\;\frac{eA_i(\mathbf{s})}{(2\pi)^6}\Dot{\chi}(t_s)\frac{p_i}{2\;\en{p}\en{k}}
    \left[-\frac{\bshs{+}{k}e^{i(\en{p}(t_y-t_s)-\en{k}(t_x-t_s))}}{(\en{p}-\en{k})}+\right.\\
    \left.-\frac{\bshs{+}{k}e^{-i(\en{p}(t_y-t_s)+\en{k}(t_x-t_s))}}{(\en{k}+\en{p})}
   -\frac{\bshs{-}{k}e^{i(\en{p}(t_y-t_s)+\en{k}(t_x-t_s))}}{(\en{p}+\en{k})}-\frac{\bshs{-}{k}e^{-i(\en{p}(t_y-t_s)-\en{k}(t_x-t_s))}}{(\en{p}-\en{k})}\right]e^{i\mathbf{p}(\mathbf{y-s})+i\mathbf{k}(\mathbf{x-s})}. 
\end{multline}
To finally determine the first order in $e$ of \eqref{cosadacuipartireincap4}, we still need (see \eqref{pert_exp}):
\begin{equation}\label{correzionekmsintcap4}
    E\coloneqq\int_0^{\beta}\di u\;\omega^{\beta,\mathcal{T}}(\varphi^{\dagger}(x)\star\varphi(y)\otimes\alpha_{iu}\Dot{V}).
\end{equation}
 Substituting in it the explicit expression for $\Dot{V}$, using our assumption about the potential ($A_0=0$), taking the adiabatic limit $h \to 1$, inserting the expression for the two-point function arising from the definition of the product and, at the very end, integrating in $u$ we get:
\begin{align}\label{lunghissimocap4}
\nonumber
     E^{(1)}=ie\int_{-\infty}^{+\infty} \int\int\int \frac{\Dot{\chi}(t_z)A_i(\mathbf{z})}{(2\pi)^6}(ip_i-ik_i)\left[\overbrace{-\frac{e^{-i\en{k}(t_x-t_z-i\beta)-i\en{p}(t_y-t_z-i\beta)}}{(1-e^{-\beta\en{k}})(1-e^{-\beta\en{p}})(\en{k}+\en{p})4\en{p}\en{k}}}\right.\\\nonumber
     -\frac{e^{i\en{k}(t_x-t_z-i\beta)-i\en{p}(t_y-t_z-i\beta)}}{(1-e^{\beta\en{k}})(1-e^{-\beta\en{p}})(\en{k}-\en{p})4\en{p}\en{k}}+\frac{e^{-i\en{k}(t_x-t_z-i\beta)+i\en{p}(t_y-t_z-i\beta)}}{(1-e^{-\beta\en{k}})(1-e^{\beta\en{p}})(\en{k}-\en{p})4\en{p}\en{k}}\\\nonumber
      +\frac{e^{i\en{k}(t_x-t_z-i\beta)+i\en{p}(t_y-t_z-i\beta)}}{(1-e^{\beta\en{k}})(1-e^{\beta\en{p}})(\en{p}+\en{k})4\en{p}\en{k}}+\overbrace{\frac{e^{-i\en{k}(t_x-t_z)-i\en{p}(t_y-t_z)}}{(1-e^{-\beta\en{k}})(1-e^{-\beta\en{p}})(\en{k}+\en{p})4\en{p}\en{k}}}\\\nonumber
     +\frac{e^{i\en{k}(t_x-t_z)-i\en{p}(t_y-t_z)}}{(1-e^{\beta\en{k}})(1-e^{-\beta\en{p}})(\en{k}-\en{p})4\en{p}\en{k}}-\frac{e^{-i\en{k}(t_x-t_z)+i\en{p}(t_y-t_z)}}{(1-e^{-\beta\en{k}})(1-e^{\beta\en{p}})(\en{k}-\en{p})4\en{p}\en{k}}\\
      \left.-\frac{e^{i\en{k}(t_x-t_z)+i\en{p}(t_y-t_z)}}{(1-e^{\beta\en{k}})(1-e^{\beta\en{p}})(\en{p}+\en{k})4\en{p}\en{k}}\right]e^{i\mathbf{k}(\mathbf{x-z})+i\mathbf{p}(\mathbf{y-z})}\;\di t_z\di^3\textbf{z}\;\di^3\textbf{k}\;\di^3\textbf{p}.
\end{align}
We now want to check how the correction just found combines with the contributions to the two-point function ($\mathfrak{C}_1,\mathfrak{C}_2,\mathfrak{B}_1,\mathfrak{B}_2$) to give a time translations invariant distribution. For this purpose, we sum the terms in \eqref{lunghissimocap4} ($E^{(1)}$), \eqref{primoterminenoninvcap4}($\mathfrak{C}_1$) and  \eqref{secondoterminenoninvcap4}($\mathfrak{C}_2$) that have the same exponential modes. For the sake of simplicity, we consider only the modes underlined in \eqref{lunghissimocap4}(the computation for the other terms is analogous):
\begin{align}\nonumber\label{sideveannullarecap4}
    e\int_{-\infty}^{\infty}\di t_s\int \di^3\textbf{s}\;\di^3\textbf{p}\;\di^3\textbf{k}\; \frac{A_i(\mathbf{s})}{(2\pi)^6}\frac{\Dot{\chi}(t_s)}{4\en{p}\en{k}}\frac{e^{-i(\en{k}(t_x-t_s)+\en{p}(t_y-t_s))}}{(\en{p}+\en{k})}e^{i\mathbf{p}(\mathbf{y-s})+i\mathbf{k}(\mathbf{x-s})}\\\nonumber
    \left[k_i\left[\frac{1}{1-e^{-\beta\en{p}}}-\frac{1}{e^{\beta\en{p}}-1}-\frac{1}{(1-e^{-\beta\en{k}})(e^{\beta\en{p}}-1)}+\frac{1}{(e^{\beta\en{k}}-1)(1-e^{-\beta\en{p}})}\right]-\right.\\
    \left.p_i\left[-\frac{1}{1-e^{-\beta\en{k}}}+\frac{1}{e^{\beta\en{k}}-1}+\frac{1}{(1-e^{-\beta\en{k}})(e^{\beta\en{p}}-1)}-\frac{1}{(e^{\beta\en{k}}-1)(1-e^{-\beta\en{p}})}\right]\right].
\end{align}
To understand why the expression in \eqref{sideveannullarecap4} is actually equal to $0$ we have to make use of the chosen gauge condition. In fact, by taking the Fourier transform of the gauge condition, we get the following \emph{transversality condition}:
\begin{equation}\label{trasversalitàcap4}
    \partial_i A_i(\mathbf{s})=0\xrightarrow[]{\mathcal{F}}q_i\hat{A_i}(\mathbf{q})=0.
\end{equation}
However, the integration in the variable $\textbf{s}$ in \eqref{sideveannullarecap4} produces the Fourier transform of the vector potential as a function of the three-momentum $\textbf{p}+\textbf{k}$. Follows that Equation \eqref{sideveannullarecap4} vanishes.\\
Repeating the same computation for all the exponential modes, it's easy to check that $E^{(1)}+\mathfrak{B}_1+\mathfrak{B}_2=0$ and, therefore:
\begin{equation*}
     \omega^{\beta,V}(R_V(\varphi^\dagger(x))\star R_V(\varphi(y)))=\omega_2^{\beta}(x,y)+\mathfrak{C}_1+\mathfrak{C}_2+\mathcal{O}(e^2)
\end{equation*}
Observing that $\mathfrak{C}=\mathfrak{C}_1+\mathfrak{C}_2$ concludes the proof of the proposition.
\end{proof}

\subsection{Proofs of Section \ref{sec: Dirac example}}
\begin{proof}[\textbf{Proof of Lemma \ref{lem: DiracTerm1}}]
We start noticing that the normal ordering of the conserved current can be removed as its contribution will just provide disconnected diagrams. Therefore, writing everything out explicitly:
\begin{align*}
    -\int_0^{\beta}\di u \, \cancel{\omega}^{\beta,{\mathcal{T}}}\big(j^0(x) \otimes \tau_{iu}(K)\big) &= -e \int_0^{\beta}\di u \, \cancel{\omega}^{\beta,{\mathcal{T}}}\big(\overline{\psi} \gamma^0 \psi(x) \otimes \tau_{iu}(K)\big)\\
    &= -\lambda e^2 \int_0^{\beta} \di u \int \di t_1 \di^3\mathbf{x}_1  \Dot{\chi}(t_1) h(\mathbf{x}_1)\cancel{\omega}^{\beta,{\mathcal{T}}}\big((\overline{\psi} \gamma^0 \psi)(t - iu, \mathbf{x}) \otimes \overline{\psi}(x_1) \gamma^{\mu} A_{\mu}(\mathbf{x}_1) \psi(x_1)\big).
\end{align*}
Here we used the stationarity of $\cancel{\omega}^{\beta,{\mathcal{T}}}$ together with its analytic extension to the complex strip $\Im(z) \in (0,\beta)$. Inserting the two-point functions and introducing Spinor indices\footnote{We introduce them in this expression just to highlight and keep in mind the order in which gamma matrices are multiplied. In what follows we drop spinorial indices again.} we get:
\begin{align*}
    &-\int_0^{\beta}\di u \, \cancel{\omega}^{\beta,{\mathcal{T}}}\big(j^0(x) \otimes \tau_{iu}(K)\big) =\\
    &-\lambda e^2\int_0^{\beta} \di u \int \di t_1 \di^3\mathbf{x}_1  \Dot{\chi}(t_1) h(\mathbf{x}_1) A_{\mu}(\mathbf{x}_1)
    \big(\gamma^{0}\big)^{A}_{\,\,\,B} \big(\gamma^{\mu} \big)^{C}_{\,\,\,D}\big(\cancel{S}^{\beta,-}_2(t_1 - t + iu, \mathbf{x}_1 - \mathbf{x})\big)^{D}_{\,\,\,A} \big(\cancel{S}^{\beta,+}_{2}(t - iu - t_1, \mathbf{x} - \mathbf{x}_1)\big)^{B}_{\,\,\, C}.
\end{align*}
Using the explicit form of the KMS two-point functions (see equations \eqref{eq: 2puntiKMS_+} and \eqref{eq: 2puntiKMS_-}) and computing the integral in the $u$-variable:
\begin{align*}
    -\int_0^{\beta}\di u \, \cancel{\omega}^{\beta,{\mathcal{T}}}\big(j^0(x) \otimes &\tau_{iu}(K)\big) = -\frac{\lambda e^2}{(2\pi)^6} \int \di t_1 \di^3\mathbf{x}_1  \Dot{\chi}(t_1) h(\mathbf{x}_1) A_{\mu}(\mathbf{x}_1) \gamma^{0} \gamma^{\mu} \int \frac{\di^3\mathbf{p} \di^3 \mathbf{k}}{4 \en{p} \en{k}}e^{-i(\mathbf{p} - \mathbf{k})(\mathbf{x} - \mathbf{x}_1)}\\
    &\times \bigg[ e^{i(\en{p} - \en{k})(t - t_1)} \bigg( \frac{e^{(\en{p} - \en{k}) \beta} - 1}{\en{p} - \en{k}}  \bigg) \frac{(-\gamma^0 \en{p} - \gamma^i p_i + m) (-\gamma^0 \en{k} - \gamma^i k_i + m)}{(1 + e^{\beta \en{p}}) (1 + e^{-\beta \en{k}})}\\
    &\,\,\,\,  - e^{i(\en{p} + \en{k})(t - t_1)} \bigg( \frac{e^{(\en{p} + \en{k}) \beta} - 1}{\en{p} + \en{k}}  \bigg)\frac{(-\gamma^0 \en{p} - \gamma^i p_i + m) (\gamma^0 \en{k} - \gamma^i k_i + m)}{(1 + e^{\beta \en{p}}) (1 + e^{\beta \en{k}})}\\
    &\,\,\,\, + e^{-i(\en{p} + \en{k})(t - t_1)} \bigg( \frac{e^{-(\en{p} + \en{k}) \beta} - 1}{\en{p} + \en{k}}  \bigg)\frac{(\gamma^0 \en{p} - \gamma^i p_i + m) (-\gamma^0 \en{k} - \gamma^i k_i + m)}{(1 + e^{-\beta \en{p}}) (1 + e^{-\beta \en{k}})}\\
    &\,\,\,\, - e^{-i(\en{p} - \en{k})(t - t_1)} \bigg( \frac{e^{-(\en{p} - \en{k}) \beta} - 1}{\en{p} - \en{k}}  \bigg)\frac{(\gamma^0 \en{p} - \gamma^i p_i + m) (\gamma^0 \en{k} - \gamma^i k_i + m)}{(1 + e^{-\beta \en{p}}) (1 + e^{\beta \en{k}})}\bigg].
\end{align*}
However, for what concerns Fermi factors, we may use the following decomposition:
\begin{equation*}
    \frac{e^{\beta (\en{p} - \en{k})} - 1}{(1 + e^{\beta \en{p}}) (1 + e^{-\beta \en{k}})} = \frac{e^{\beta \en{p}}(e^{ - \beta \en{k}} + 1) - 1 - e^{\beta \en{p}}}{(1 + e^{\beta \en{p}}) (1 + e^{-\beta \en{k}})} = - \frac{1}{(1 + e^{-\beta \en{k}})} + \frac{1}{(1 + e^{-\beta \en{p}})},
\end{equation*}
together with:
\begin{equation*}
    \frac{1}{1+ e^{-\beta \en{p}}} - 1 = \frac{-e^{-\beta \en{p}}}{1 + e^{-\beta \en{p}}} = -\frac{1}{1 + e^{\beta \en{p}}},
\end{equation*}
to conclude:
\begin{align*}
    -\int_0^{\beta}\di u \, \cancel{\omega}^{\beta,{\mathcal{T}}}\big(j^0(x) \otimes &\tau_{iu}(K)\big) = -\frac{\lambda e^2}{(2\pi)^6} \int \di t_1 \di^3\mathbf{x}_1  \Dot{\chi}(t_1) h(\mathbf{x}_1) A_{\mu}(\mathbf{x}_1) \gamma^{0} \gamma^{\mu} \int \frac{\di^3\mathbf{p} \di^3 \mathbf{k}}{4 \en{p} \en{k}}e^{-i(\mathbf{p} - \mathbf{k})(\mathbf{x} - \mathbf{x}_1)} \cdot\\
    &\cdot \bigg[ \frac{e^{i(\en{p} - \en{k})(t - t_1)}}{\en{p} - \en{k}} (-\gamma^0 \en{p} - \gamma^i p_i + m) (-\gamma^0 \en{k} - \gamma^i k_i + m) \bigg( \frac{1}{1 + e^{\beta \en{k}}} - \frac{1}{1 + e^{\beta \en{p}}}  \bigg)\\
    &\,\,\,\, + \frac{e^{i(\en{p} + \en{k})(t - t_1)}}{\en{p} + \en{k}} (-\gamma^0 \en{p} - \gamma^i p_i + m) (\gamma^0 \en{k} - \gamma^i k_i + m)\bigg( \frac{1}{1 + e^{\beta \en{k}}} + \frac{1}{1 + e^{\beta \en{p}}} - 1 \bigg)\\
    &\,\,\,\, + \frac{e^{-i(\en{p} + \en{k})(t - t_1)}}{\en{p} + \en{k}} (\gamma^0 \en{p} - \gamma^i p_i + m) (-\gamma^0 \en{k} - \gamma^i k_i + m) \bigg( \frac{1}{1 + e^{\beta \en{k}}} + \frac{1}{1 + e^{\beta \en{p}}} - 1\bigg)\\
    &\,\,\,\, + \frac{e^{-i(\en{p} - \en{k})(t - t_1)}}{\en{p} - \en{k}} (\gamma^0 \en{p} - \gamma^i p_i + m) (\gamma^0 \en{k} - \gamma^i k_i + m) \bigg( \frac{1}{1 + e^{\beta \en{k}}} - \frac{1}{1 + e^{\beta \en{p}}} \bigg)\bigg].
\end{align*}
This concludes the proof of the Lemma.
\end{proof}

\begin{proof}[\textbf{Proof of Lemma \ref{lem: DiracTerm2}}]
We want to compute:
\begin{equation*}
    -i \cancel{\omega}^{\beta}(V :j^0(x):) +i \cancel{\omega}^{\beta}\big(T(V, :j^0(x):)\big).
\end{equation*}
Start noticing that the normal ordering of the conserved current can be dropped as its contribution in each term exactly cancels the one arising from the other factor. By the explicit form of the potential and of the current we have:
\begin{align*}
    -i\cancel{\omega}^{\beta}\big(V j^0(x)\big) = -i\lambda e\int \di^4x_1 \chi(t_1) h(\mathbf{x}_1)\big( \gamma^{\mu} \big)^{C}_{\,\,\,D} A_{\mu}(\mathbf{x}_1) \big( \cancel{S}^{\beta,+}_2(x_1 - x) \big)^{D}_{\,\,\,A} \big( \gamma^0 \big)^{A}_{\,\,\,B} \big(\cancel{S}^{\beta,-}_2(x - x_1)\big)^{B}_{\,\,\, C}
\end{align*}
and analogously with the same spinorial contractions from the definition of the time-ordered product:
\begin{equation*}
    i\cancel{\omega}^{\beta}\big(T(V , j^0(x)) \big) = -i\lambda e^2\int \di^4x_1 \chi(t_1) h(\mathbf{x}_1)(\gamma^{\mu})^C_{\,\,D} A_{\mu}(\mathbf{x}_1) \big(\cancel{S}_F^{\beta}(x_1 - x)\big)^D_{\,\,A} (\gamma^0)^A_{\,\,B} \big(\cancel{S}_F^{\beta}(x - x_1)\big)^B_{\,\,C}.
\end{equation*}
Adding them together gives\footnote{In what follows we drop again spinorial indices being aware of the correct order of contraction that we have just given.}:
\begin{align*}
    &-i\cancel{\omega}^{\beta}(V j^0(x)) + i\cancel{\omega}^{\beta}\big(T(V, j^0(x))\big) =\\
    &-i\lambda e^2\int \di^4x_1 \chi(t_1) h(\mathbf{x}_1)\gamma^{\mu} \gamma^0 A_{\mu}(\mathbf{x}_1) \big( \cancel{S}_F^{\beta}(x_1 - x) \cancel{S}_F^{\beta}(x - x_1) + \cancel{S}^{\beta,+}_2(x_1 - x) \cancel{S}^{\beta, -}_2(x - x_1)  \big).
\end{align*}
The above involves Feynman propagators and two-point functions at finite temperature. However, as it is usually done in this context, we rewrite them in terms of their equivalents at zero temperature plus the smooth functions:
\begin{align}
    W_{\beta}^{+}(x_1 - x) &= \cancel{S}^{\beta,+}_2(x_1 - x) - \cancel{S}^+_{2}(x_1 - x) \nonumber\\
    &= - \frac{1}{(2\pi)^3} \int \frac{\di^3\mathbf{p}}{2 \en{p}} \frac{e^{i \mathbf{p} (\mathbf{x}_1 - \mathbf{x})}}{1 + e^{\beta \en{p}}}\big((-\gamma^0 \en{p} - \gamma^i p_i +m)e^{-i\en{p}(t_1 - t)} +  (\gamma^0 \en{p} - \gamma^i p_i +m)e^{i\en{p}(t_1 - t)} \big)\nonumber,\\
    W_{\beta}^{-}(x_1 - x) &= \cancel{S}_2^{\beta,-}(x_1 - x) - \cancel{S}^-_{2}(x_1 - x) \nonumber\\
    &= - W_{\beta}^+(x_1 - x)\nonumber
\end{align}
and:
\begin{align*}
    W_{\beta}^F(x_1 - x) &= \cancel{S}_F^{\beta}(x_1 - x) - \cancel{S}_F(x_1 - x) \\
    &= \cancel{S}^{\beta,+}_2(x_1 - x)\theta(t_1 - t) - \cancel{S}_2^{\beta,-}(x_1 - x)\theta(t - t_1) - \cancel{S}^+_2(x_1 - x)\theta(t_1 - t) + \cancel{S}^-_2(x_1 - x)\theta(t - t_1)\\
    &= W_{\beta}^+(x_1 - x).
\end{align*}
The smoothness follows from the states being Hadamard. The above remark allows us to write:
\begin{align*}
    \cancel{S}_F^{\beta}(x_1 - x) \cancel{S}_F^{\beta}(x - x_1) + \cancel{S}^{\beta,+}_2(x_1 - x) \cancel{S}^{\beta,-}_2(x - x_1)    &= \cancel{S}_F(x_1 - x)\cancel{S}_F(x - x_1) + \cancel{S}^{+}_2(x_1 - x)\cancel{S}^{-}_2(x - x_1)\\
    &\hspace{3mm} + \cancel{S}_F(x_1 - x)W^+_{\beta}(x - x_1) + W^+_{\beta}(x_1 - x)\cancel{S}_F(x - x_1)\\
    &\hspace{3mm} - \cancel{S}^{+}_2(x_1 - x)W^+_{\beta}(x - x_1) + W^+_{\beta}(x_1 - x)\cancel{S}^{-}_2(x - x_1).
\end{align*}
Writing the Feynman propagator in terms of the Pauli-Jordan function we obtain:
\begin{align*}
    &-i\cancel{\omega}^{\beta}\big(Vj^0(x)\big) +i \cancel{\omega}^{\beta}\big(T(V,j^0(x))\big) = \\
    &-i\lambda e^2\int \di^4x_1 \chi(t_1) h(\mathbf{x}_1)\gamma^{\mu} \gamma^0 A_{\mu}(\mathbf{x}_1) \bigg[ \cancel{S}_F(x_1 - x)\cancel{S}_F(x - x_1) + \cancel{S}^{+}_2(x_1 - x)\cancel{S}^{-}_2(x - x_1)+\\
    &\,\,\,\,i\theta(t - t_1)\big(-\cancel{S}(x_1 - x)W^+_{\beta}(x - x_1) +  W^+_{\beta}(x_1 - x)\cancel{S}(x - x_1)\big) \bigg].
\end{align*}
This concludes the proof of the Lemma.
\end{proof}

\begin{proof}[\textbf{Proof of Proposition \ref{prop: DiracTotalTerm}}]
Consider the term in equation \eqref{eq: termine1} and separate the two contributions:
\begin{align*}
    Z_1 &= -\frac{\lambda e^2}{(2\pi)^6} \int \di t_1 \di^3\mathbf{x}_1  \Dot{\chi}(t_1) h(\mathbf{x}_1) A_{\mu}(\mathbf{x}_1) \gamma^{0} \gamma^{\mu} \int \frac{\di^3\mathbf{p} \di^3 \mathbf{k}}{4 \en{p} \en{k}}e^{-i(\mathbf{p} - \mathbf{k})(\mathbf{x} - \mathbf{x}_1)} \cdot\\
    &\cdot \bigg[ \frac{e^{i(\en{p} - \en{k})(t - t_1)}}{\en{p} - \en{k}} (-\gamma^0 \en{p} - \gamma^i p_i + m) (-\gamma^0 \en{k} - \gamma^i k_i + m) \bigg( \frac{1}{1 + e^{\beta \en{k}}} - \frac{1}{1 + e^{\beta \en{p}}}  \bigg)\\
    &\,\,\,\, + \frac{e^{i(\en{p} + \en{k})(t - t_1)}}{\en{p} + \en{k}} (-\gamma^0 \en{p} - \gamma^i p_i + m) (\gamma^0 \en{k} - \gamma^i k_i + m)\bigg( \frac{1}{1 + e^{\beta \en{k}}} + \frac{1}{1 + e^{\beta \en{p}}}\bigg)\\
    &\,\,\,\, + \frac{e^{-i(\en{p} + \en{k})(t - t_1)}}{\en{p} + \en{k}} (\gamma^0 \en{p} - \gamma^i p_i + m) (-\gamma^0 \en{k} - \gamma^i k_i + m) \bigg( \frac{1}{1 + e^{\beta \en{k}}} + \frac{1}{1 + e^{\beta \en{p}}}\bigg)\\
    &\,\,\,\, + \frac{e^{-i(\en{p} - \en{k})(t - t_1)}}{\en{p} - \en{k}} (\gamma^0 \en{p} - \gamma^i p_i + m) (\gamma^0 \en{k} - \gamma^i k_i + m) \bigg( \frac{1}{1 + e^{\beta \en{k}}} - \frac{1}{1 + e^{\beta \en{p}}} \bigg)\bigg]\\
\end{align*}
and:
\begin{align}
    Z_2 &= -\frac{\lambda e^2}{(2\pi)^6} \int \di t_1 d^3\mathbf{x}_1  \Dot{\chi}(t_1) h(\mathbf{x}_1) A_{\mu}(\mathbf{x}_1) \gamma^{0} \gamma^{\mu} \int \frac{\di^3\mathbf{p} \di^3 \mathbf{k}}{4 \en{p} \en{k}}e^{-i(\mathbf{p} - \mathbf{k})(\mathbf{x} - \mathbf{x}_1)} \nonumber\\
    &\times \bigg[ -\frac{e^{i(\en{p} + \en{k})(t - t_1)}}{\en{p} + \en{k}} (-\gamma^0 \en{p} - \gamma^i p_i + m) (\gamma^0 \en{k} - \gamma^i k_i + m) \nonumber\\
    &\,\,\,\, - \frac{e^{-i(\en{p} + \en{k})(t - t_1)}}{\en{p} + \en{k}} (\gamma^0 \en{p} - \gamma^i p_i + m) (-\gamma^0 \en{k} - \gamma^i k_i + m)\bigg] \label{eq: DiracZ2}.
\end{align}
The term \eqref{eq: parti1} needs to be renormalised as $\cancel{S}_F(x_1 - x)\cancel{S}_F(x - x_1) + \cancel{S}^{+}_2(x_1 - x)\cancel{S}^{-}_2(x - x_1)$ diverges in the coinciding point limit. Using K\"allén--Lehmann spectral representation, this can be extended to a well-defined distribution also on the diagonal and supported in the region $t \geq t_1$:
\begin{align*}
    &\cancel{S}_{F}(x_1 - x) \cancel{S}_{F}(x - x_1) + \cancel{S}^+_{2}(x_1 - x) \cancel{S}^-_{2}(x - x_1) =\\
    &-i \int_{(2m)^2}^{\infty} \di M^2 \rho_1^D(M^2) \bigg(\frac{(\square - c_1)(\square - c_2)}{(M^2 - c_1)(M^2 - c_2)}\Delta_A(x_1 - x,M) \bigg), \numberthis \label{eq: A_Rin}
\end{align*}
where:
\begin{equation*}
    \rho_1^D(M^2) = -\frac{1}{(2\pi)^{3}} \int \di^3 \mathbf{k} \, \delta(2\omega(\mathbf{k}) - M) \frac{(-\gamma_0 \omega(\mathbf{k}) - \gamma_i k^i + m)( \gamma_0 \omega(\mathbf{k}) - \gamma_i k^i + m)}{4 \omega^2(\mathbf{k})},
\end{equation*}
and $c_1, c_2 \in \mathbb{R}$ are renormalization constants. Now, the formal integration by parts of $\eqref{eq: parti1}$ can be taken leading to a boundary and a bulk term:
\begin{align*}
    Z_{BO} &= -i\lambda e^2\int \di^3\mathbf{x}_1 \chi(t_1) h(\mathbf{x}_1)\gamma^{\mu} \gamma^0 A_{\mu}(\mathbf{x}_1) (\mathfrak{A} + \mathfrak{B})\big|_{t_1 = -\infty}^{t_1 = t},\\
    Z_{BU} &= i\lambda e^2\int_{-\infty}^{t} \di t_1 \int \di^3\mathbf{x}_1  h(\mathbf{x}_1)\gamma^{\mu} \gamma^0 A_{\mu}(\mathbf{x}_1) \big( \Dot{\chi}(t_1)\mathfrak{A} + \Dot{\chi}(t_1) \mathfrak{B} \big).
\end{align*}
Here $\mathfrak{A}$ is the indefinite integral in $t_1$ of the renormalised expression and $\mathfrak{B}$ is the indefinite integral of the remaining part.\\ 
For the bulk term, using the support properties of $\Dot{\chi}(t_1)$, the integration domain is extended to:
\begin{equation*}
    Z_{BU} = i\lambda e^2\int \di^4x_1 \Dot{\chi}(t_1) h(\mathbf{x}_1)\gamma^{\mu} \gamma^0 A_{\mu}(\mathbf{x}_1) (\mathfrak{A} + \mathfrak{B}).
\end{equation*}
Moreover, as $t_1 < t$, we have $\cancel{S}_F(x_1 - x) = - \cancel{S}^{-}_2(x_1 - x)$ and $\cancel{S}_F(x - x_1) = \cancel{S}^{+}_2(x - x_1)$. Therefore, the $\mathfrak{A}$-term is:
\begin{align*}
    \mathfrak{A} &= -\frac{i}{(2\pi)^6} \int \frac{\di^3 \mathbf{p} \di^3 \mathbf{k}}{4 \en{k} \en{p} (\en{p} + \en{k})}e^{i(\mathbf{p} - \mathbf{k}) (\mathbf{x}_1 - \mathbf{x})} e^{-i(\en{p} + \en{k})(t_1 - t)}(-\gamma^0 \en{p} - \gamma^i p_i + m)(\gamma^0 \en{k} - \gamma^i k_i + m)\\
    &\,\,\,\,- \frac{i}{(2\pi)^6} \int \frac{\di^3 \mathbf{p} \di^3 \mathbf{k}}{4 \en{k} \en{p} (\en{p} + \en{k})}e^{i(\mathbf{p} - \mathbf{k}) (\mathbf{x}_1 - \mathbf{x})} e^{i(\en{p} + \en{k})(t_1 - t)}(\gamma^0 \en{p} - \gamma^i p_i + m)(-\gamma^0 \en{k} - \gamma^i k_i + m),
\end{align*}
giving a contribution:
\begin{align*}
    Z_{BU}^{\mathfrak{A}} &= \frac{\lambda e^2}{(2 \pi)^6}\int \di^4x_1 \Dot{\chi}(t_1) h(\mathbf{x}_1)(\gamma^{\mu})^{C}_{\,\,\,D} (\gamma^0)^A_{\,\,\,B} A_{\mu}(\mathbf{x}_1) \int \frac{\di^3 \mathbf{p} \di^3 \mathbf{k}}{4 \en{k} \en{p} (\en{p} + \en{k})}e^{i(\mathbf{p} - \mathbf{k}) (\mathbf{x}_1 - \mathbf{x})}\\
    &\times \bigg[ e^{-i(\en{p} + \en{k})(t_1 - t)}(-\gamma^0 \en{p} - \gamma^i p_i + m)^{D}_{\,\,\,A}(\gamma^0 \en{k} - \gamma^i k_i + m)^{B}_{\,\,\,C}\\
    &\,\,\,\,\,\,+ e^{i(\en{p} + \en{k})(t_1 - t)}(\gamma^0 \en{p} - \gamma^i p_i + m)^{D}_{\,\,\,A}(-\gamma^0 \en{k} - \gamma^i k_i + m)^{B}_{\,\,\,C} \bigg].
\end{align*}
Finally, remembering that this terms is summed with $Z_2$ (Eq \eqref{eq: DiracZ2}), we get:
\begin{equation*}
    Z_2 + Z^{\mathfrak{A}}_{BU} = 0.
\end{equation*}
Now, we show that the same holds for the sum between $Z_1$ and $Z^{\mathfrak{B}}_{BU}$. Therefore, we look at the remaining contribution to the bulk:
\begin{align*}
    W_{\beta}^+(x_1 - x) \cancel{S}(x-x_1) - \cancel{S}(x_1 &- x)W_{\beta}^+(x - x_1) =\\
    \frac{i}{(2 \pi)^6}\int \frac{\di^3 \mathbf{p} \di^3 \mathbf{k}}{4 \en{k} \en{p}} e^{-i(\mathbf{p} - \mathbf{k})\cdot(\mathbf{x} - \mathbf{x}_1)} &\bigg[ e^{i(\en{p} - \en{k})(t-t_1)}(-\gamma^0 \en{p} - \gamma^i p_i + m)(-\gamma^0 \en{k} -\gamma^i k_i + m)\bigg( \frac{1}{1 + e^{\beta \en{p}}} - \frac{1}{1 + e^{\beta \en{k}}} \bigg)\\ 
    &+ e^{i(\en{p} + \en{k})(t-t_1)}(-\gamma^0 \en{p} - \gamma^i p_i + m)(\gamma^0 \en{k} -\gamma^i k_i + m)\bigg( -\frac{1}{1 + e^{\beta \en{p}}} - \frac{1}{1 + e^{\beta \en{k}}} \bigg) \\
    &+ e^{-i(\en{p} + \en{k})(t-t_1)}(\gamma^0 \en{p} - \gamma^i p_i + m)(-\gamma^0 \en{k} -\gamma^i k_i + m)\bigg( \frac{1}{1 + e^{\beta \en{p}}} + \frac{1}{1 + e^{\beta \en{k}}} \bigg)\\
    &+ e^{-i(\en{p} - \en{k})(t-t_1)}(\gamma^0 \en{p} - \gamma^i p_i + m)(\gamma^0 \en{k} -\gamma^i k_i + m)\bigg( -\frac{1}{1 + e^{\beta \en{p}}} + \frac{1}{1 + e^{\beta \en{k}}} \bigg)\bigg].
\end{align*}
Integrating it in $t_1$:
\begin{align*}
    Z_{BU}^{\mathfrak{B}} &= -\frac{\lambda e^2}{(2\pi)^6} \int \di^4x_1\Dot{\chi}(t_1) h(\mathbf{x}_1)(\gamma^{\mu})^{C}_{\,\,\,D} (\gamma^0)^{A}_{\,\,\,B} A_{\mu}(\mathbf{x}_1) \int \frac{\di^3 \mathbf{p} \di^3 \mathbf{k}}{4 \en{k} \en{p}} e^{-i(\mathbf{p} - \mathbf{k})\cdot(\mathbf{x} - \mathbf{x}_1)} \\
    &\times \bigg[ \frac{e^{i(\en{p} - \en{k})(t-t_1)}}{\en{p} - \en{k}}(-\gamma^0 \en{p} - \gamma^i p_i + m)^{D}_{\,\,\,A}(-\gamma^0 \en{k} -\gamma^i k_i + m)^{B}_{\,\,\,C}\bigg( -\frac{1}{1 + e^{\beta \en{p}}} + \frac{1}{1 + e^{\beta \en{k}}} \bigg)\\ 
    &\,\,\,\,+ \frac{e^{i(\en{p} + \en{k})(t-t_1)}}{\en{p} + \en{k}}(-\gamma^0 \en{p} - \gamma^i p_i + m)^{D}_{\,\,\,A}(\gamma^0 \en{k} -\gamma^i k_i + m)^{B}_{\,\,\,C}\bigg( \frac{1}{1 + e^{\beta \en{p}}} + \frac{1}{1 + e^{\beta \en{k}}} \bigg) \\
    &\,\,\,\,+ \frac{e^{-i(\en{p} + \en{k})(t-t_1)}}{\en{p} + \en{k}}(\gamma^0 \en{p} - \gamma^i p_i + m)^{D}_{\,\,\,A}(-\gamma^0 \en{k} -\gamma^i k_i + m)^{B}_{\,\,\,C}\bigg( \frac{1}{1 + e^{\beta \en{p}}} + \frac{1}{1 + e^{\beta \en{k}}} \bigg)\\
    &\,\,\,\,+ \frac{e^{-i(\en{p} - \en{k})(t-t_1)}}{\en{p} - \en{k}}(\gamma^0 \en{p} - \gamma^i p_i + m)^{D}_{\,\,\,A}(\gamma^0 \en{k} -\gamma^i k_i + m)^{B}_{\,\,\,C}\bigg( -\frac{1}{1 + e^{\beta \en{p}}} + \frac{1}{1 + e^{\beta \en{k}}} \bigg)\bigg].
\end{align*}
It follows that:
\begin{equation*}
    Z_1 + Z^{\mathfrak{B}}_{BU} = 0.
\end{equation*}
The only remaining part is the boundary term $Z_{BO}$. Again, computing separately the two terms, using the form of the switch on:
\begin{align*}
    Z^{\mathfrak{B}}_{BO} &= \frac{\lambda e^2}{(2 \pi)^6}\int \di^3 \mathbf{x}_1 h(\mathbf{x}_1) (\gamma^{\mu})^{C}_{\,\,\,D} (\gamma^0)^{A}_{\,\,\,B} A_{\mu}(\mathbf{x}_1) \int \frac{\di^3 \mathbf{p} \di^3 \mathbf{k}}{4 \en{k} \en{p}} e^{-i(\mathbf{p} - \mathbf{k})\cdot(\mathbf{x} - \mathbf{x}_1)} \\
    &\bigg[ \frac{1}{\en{p} - \en{k}}(-\gamma^0 \en{p} - \gamma^i p_i + m)^{D}_{\,\,\,A}(-\gamma^0 \en{k} -\gamma^i k_i + m)^{B}_{\,\,\,C}\bigg( -\frac{1}{1 + e^{\beta \en{p}}} + \frac{1}{1 + e^{\beta \en{k}}} \bigg)\\ 
    &+ \frac{1}{\en{p} + \en{k}}(-\gamma^0 \en{p} - \gamma^i p_i + m)^{D}_{\,\,\,A}(\gamma^0 \en{k} -\gamma^i k_i + m)^{B}_{\,\,\,C}\bigg( \frac{1}{1 + e^{\beta \en{p}}} + \frac{1}{1 + e^{\beta \en{k}}} \bigg) \\
    &+ \frac{1}{\en{p} + \en{k}}(\gamma^0 \en{p} - \gamma^i p_i + m)^{D}_{\,\,\,A}(-\gamma^0 \en{k} -\gamma^i k_i + m)^{B}_{\,\,\,C}\bigg( \frac{1}{1 + e^{\beta \en{p}}} + \frac{1}{1 + e^{\beta \en{k}}} \bigg)\\
    &+ \frac{1}{\en{p} - \en{k}}(\gamma^0 \en{p} - \gamma^i p_i + m)^{D}_{\,\,\,A}(\gamma^0 \en{k} -\gamma^i k_i + m)^{B}_{\,\,\,C}\bigg( -\frac{1}{1 + e^{\beta \en{p}}} + \frac{1}{1 + e^{\beta \en{k}}} \bigg)\bigg],
\end{align*}
that using the relation for traces of $\gamma$ matrices:
\begin{align*}
    (\gamma^{\mu})^{A}_{\,\,\,B} (\gamma^{\nu})^{B}_{\,\,\,A} &= - 4 \eta^{\mu \nu,}\\
    (\gamma^{\mu})^{A}_{\,\,\,B} (\gamma^{\nu})^{B}_{\,\,\,C} (\gamma^{\tau})^{C}_{\,\,\,A} &= 0,\\
    (\gamma^{\mu})^{A}_{\,\,\,B} (\gamma^{\nu})^{B}_{\,\,\,C} (\gamma^{\tau})^{C}_{\,\,\,D} (\gamma^{\rho})^{D}_{\,\,\,A} &= 4 \bigg( \eta^{\mu \nu} \eta^{\tau \rho} - \eta^{\mu \tau} \eta^{\nu \rho} + \eta^{\mu \rho} \eta^{\nu \tau}\bigg),
\end{align*} 
gives the final result:
\begin{align*}
    Z^{\mathfrak{B}}_{BO} = \frac{4\lambda e^2}{(2 \pi)^6}\int \di^3 \mathbf{x}_1 h(\mathbf{x}_1) A^{0}(\mathbf{x}_1) \int \frac{\di^3 \mathbf{p} \di^3 \mathbf{k}}{\en{p}^2 - \en{k}^2} e^{-i(\mathbf{p} - \mathbf{k})(\mathbf{x} - \mathbf{x}_1)} \bigg[  \frac{\en{p}^2 + m^2 + k_i p^i}{\en{p} (1 + e^{\beta \en{p}})} - \frac{\en{k}^2 + m^2 + k_i p^i}{\en{k} (1 + e^{\beta \en{k}})}\bigg]. 
\end{align*}
For what concerns $Z^{\mathfrak{A}}_{BO}$, considering the renormalised expression \eqref{eq: A_Rin}, integrating it in $t_1$ one gets:
\begin{align*}
    Z^{\mathfrak{A}}_{BO} = 0
\end{align*}
by noticing that the overall contraction of spinorial indices in $\rho_1^D(M^2)$ with the remaining gamma matrices gives a vanishing contribution. This result is consistent with the Ward identities. Finally, computing the remaining renormalization freedoms arising from the extension of the product of Feynman propagators, proportional to the Dirac deltas and a number of its derivatives corresponding to the scaling degree of the product of propagators, one gets:
\begin{equation*}
    4 \lambda e^2 \left( a_0 h(\mathbf{x}) A^0(\mathbf{x}) + \sum_{i=1}^{3} a_1^i \left( \partial_{y_i} (h(\mathbf{y}) A^0(\mathbf{y}))\right)\big|_{\mathbf{y} = \mathbf{x}} + \sum_{i,j = 1}^{3} a_2^{ij} \left( \partial_{y_i} \partial_{y_j} (h(\mathbf{y}) A^0(\mathbf{y}))\right)\big|_{\mathbf{y} = \mathbf{x}}\right).
\end{equation*}
Summing these freedoms with $Z^{\mathfrak{B}}_{BO}$ concludes the proof.
\end{proof}

${}$ \\ \\ ${}$ \\
{\bf  Conflict of interest}
The authors have no relevant financial or non-financial interests to disclose. The authors have no conflict of interest to declare that are relevant to the content of this article.

 ${}$ \\
{\bf  Data statement}
Data sharing is not applicable to this article as no new data were created or analysed in this study.

\printbibliography
\end{document}